\documentclass[bja,noinfoline,preprint]{imsart}

\RequirePackage{amsthm,amsmath,amsfonts,amssymb}
\RequirePackage[numbers]{natbib}
\usepackage{todonotes}
\usepackage{enumerate}

\usepackage[linesnumbered,ruled,vlined]{algorithm2e}

\usepackage{hyperref}
\hypersetup{
    colorlinks=true,
    linkcolor=blue,
    filecolor=magenta,      
    urlcolor=cyan,
}

\startlocaldefs

\theoremstyle{plain}
\newtheorem{thm}{Theorem}[section]
\newtheorem{cor}[thm]{Corollary}
\newtheorem{lem}[thm]{Lemma}
\newtheorem{prop}[thm]{Proposition}
\newtheorem{defn}[thm]{Definition}

\newtheorem{aspt}[thm]{Assumption}
\newtheorem{rem}[thm]{Remark}

\newcommand{\Vhat}{\widehat{V}}
\newcommand{\Lambdahat}{\widehat{\Lambda}}
\newcommand{\Sigmahat}{\widehat{\Sigma}}

\newcommand{\reals}{\mathbb{R}}

\newcommand{\Hhat}{\widehat{H}}

\newcommand{\calX}{\mathcal{X}}

\newcommand{\res}{\mathcal{R}}
\newcommand{\var}{{\rm var}}
\newcommand{\Hell}{D_{H}}
\newcommand{\trace}{\tr}

\newcommand{\Phat}{\widehat{P}}
\newcommand{\Xhat}{\widehat{\mathcal{X}}}

\newcommand{\E}{\mathbb{E}}
\newcommand{\Error}{\mathcal{E}}

\newcommand{\tr}{{\rm tr}}
\newcommand{\Mmu}{m}
\newcommand{\apdf}{\varphi}

\usepackage{tikz}
\usetikzlibrary{positioning,shapes,calc,3d}
\usepackage{pgfplots}
\usepgfplotslibrary{fillbetween}

\pgfplotsset{%
      every axis/.append style={width=.9\textwidth,
                                axis x line=bottom, axis y line=left,
                                x axis line style={very thick,->}, y axis line style={very thick,->},
                                tick align=inside, tick style={thick},
                                every x tick label/.style={font=\tiny},
                                every y tick label/.style={font=\tiny,text width=3em,align=right},
                                },
      every axis legend/.append style={
                                legend columns=1,
                                font=\tiny,
                                draw=none,
                                fill=white,
                                },
      every axis x label/.style={at={(0.5,-0.12)},below,fill=none,fill opacity=1,text opacity=1,font=\small},
      every axis y label/.style={at={(-0.3,0.5)},fill=none,fill opacity=1,text opacity=1,rotate=90,font=\small},
      }

\endlocaldefs

\usepackage{xspace}
\def\isarxiv{1}

\begin{document}

\newcommand{\appname}{
\ifx\isarxiv\undefined
Supplementary Material\xspace
\else
Appendix\xspace
\fi
}

\newcommand{\smartref}[1]{
\ifx\isarxiv\undefined
Section \ref{#1} of the Supplementary Material\xspace
\else
Appendix \ref{#1}\xspace
\fi
}

\begin{frontmatter}
\title{A unified performance analysis of likelihood-informed subspace methods}
\runtitle{Accuracy of LIS methods}

\begin{aug}
\author[A]{\fnms{Tiangang} \snm{Cui}\ead[label=e1]{tiangang.cui@monash.edu}},
\author[B]{\fnms{Xin} \snm{T. Tong}\ead[label=e2]{mattxin@nus.edu.sg}}
\address[A]{Monash University, School of Mathematics \printead{e1}}

\address[B]{National University of Singapore, Department of Mathematics \printead{e2}}
\end{aug}

\begin{abstract}
The likelihood-informed subspace (LIS) method 
offers a viable route to reducing the dimensionality of high-dimensional probability distributions arising in Bayesian inference. 
LIS identifies an intrinsic low-dimensional linear subspace where the target distribution differs the most from some tractable reference distribution. Such a subspace can be identified using the leading eigenvectors of a Gram matrix of the gradient of the log-likelihood function. 
Then, the original high-dimensional target distribution is approximated through various forms of marginalization of the likelihood function, in which the approximated likelihood only has support on the intrinsic low-dimensional subspace. 
This approximation enables the design of inference algorithms that can scale sub-linearly with the apparent dimensionality of the problem. 
Intuitively, the accuracy of the approximation, and hence the performance of the inference algorithms, are influenced by three factors---the dimension truncation error in identifying the subspace, Monte Carlo error in estimating the Gram matrices, and Monte Carlo error in constructing marginalizations. 
This work establishes a unified framework to analyze each of these three factors and their interplay. 
Under mild technical assumptions, we establish error bounds for a range of existing dimension reduction techniques based on the principle of LIS.  Our error bounds also provide useful insights into the accuracy of these methods. 
In addition, we analyze the integration of LIS with sampling methods such as Markov Chain Monte Carlo (MCMC) and sequential Monte Carlo (SMC). 
We also demonstrate the applicability of our analysis on a linear inverse problem with Gaussian prior, which shows that all the estimates can be dimension-independent if the prior covariance is a trace-class operator. Finally, we demonstrate various aspects of our theoretical claims on two nonlinear inverse problems.
\end{abstract}


\end{frontmatter}

\section{Introduction}

Many applications in science and engineering must contend with expensive or intractable models that are typically driven by high-dimensional or even infinite-dimensional random variables. Some examples are seismic imaging \cite{bui2012extreme, martin2012stochastic}, subsurface energy \cite{cui2011bayesian}, glaciology \cite{petra2014computational}, groundwater \cite{dodwell2019multilevel,iglesias2014well}, electrical impedance tomography \cite{kaipio2000statistical}, and density estimation \cite{murray2008gaussian}.
Denoting the high-dimensional random variables of interest by $X \in \calX \subseteq \reals^d$, the associated target probability density often takes the form
\begin{equation}
\label{eqn:bayes}
\pi(x) = \frac{1}{Z} \, \mu(x) f(x), \quad Z = \int \mu(x)f(x) dx,
\end{equation}
where we refer to $Z$, $\mu(x)$, and $f(x)$ as the \emph{normalization constant}, the \emph{reference} density and the \emph{likelihood} function, respectively.
In the most common scenario, the target density is the posterior defined by Bayes' rule, the reference density is the prior, and the likelihood function is often denoted by $f(x; y)$ for some observed data $y$. Here we drop the dependency of $f$ on $y$ for brevity unless otherwise required. 

In most of the aforementioned applications, the reference density $\mu(x)$ takes a simple form, e.g. a Gaussian density or an elliptical density, so that the reference distribution, its marginal distributions, and its conditional distributions can be directly evaluated and sampled from. 
However, the likelihood function $f$, which often encodes some highly nonlinear parameter-to-observable map that represents the underlying model, may introduce complicated nonlinear interactions among parameters. When the parameter is also high-dimensional, generating samples from the target distribution using classical methods such as Markov chain Monte Carlo (MCMC) and sequential Monte Carlo (SMC) can be a computationally challenging task. The computational effort required for generating each independent sample from $\pi(x)$ may scale super-linearly with the ambient parameter dimension $d$.

In many high-dimensional problems, there often exists a low-dimensional ``effective" or ``intrinsic" dimension. 
Designing scalable sampling methods that can use this property has been a focus in the recent literature \cite{agapiou2018rates,agapiou2017importance,agapiou2018unbiased,beskos2014stability,beskos2017geometric,beskos2018multilevel,Cuietal16b,morzfeld2019localization,rudolf2018generalization,sanz2018importance,tong2020mala}.
One effective strategy involves finding a parameter subspace $\calX_r$ with dimensionality $d_r \ll d$, so that the original density with high ambient parameter dimensions can be approximated by some low-dimensional parametrization. 
The recently developed likelihood informed subspace (LIS) method \cite{CKB16,cui2014likelihood,zahm2018certified} offers a way to identify $\calX_r$ for high-dimensional target densities and approximates the target density via projections of the likelihood function onto $\calX_r$. 
For sampling related problems, such projections naturally lead to MCMC and SMC computations on the reduced subspace $\calX_r$. As a result, this may significantly lower the computation effort compared with implementations directly targeting the ambient space $\calX$. 
In this work, we focus on the analysis of the approximation accuracy of the LIS method and its related sampling algorithms. 

\subsection{Likelihood informed subspaces}
\label{sec:gram}

Dimension reduction techniques have been exploited to reduce the computational cost due to the parameter dimension. 
When the target density $\pi(x)$ has a known covariance matrix $\Sigma$,  a common approach is to use the principal component analysis or Karhunen--Lo\'{e}ve decomposition \cite{DimRedu:Karhunen_1947,DimRedu:Loeve_1978} that identifies the leading eigenvectors of $\Sigma$ to define the subspace $\calX_r$. Then, the parameters in the complement subspace of $\calX_r$ are ignored in the inference problem. 
Other than the computational difficulties of estimating the covariance matrix for high-dimensional non-Gaussian target densities, this approach is proven to be suboptimal even for problems with Gaussian reference densities and Gaussian likelihood functions \cite{spantini2015optimal}.

Without ignoring parameters from the inference procedure, LIS exploits an alternative way to approximate target densities. 
The intuition underpinning the development of LIS is that the likelihood function $f(x_r, x_\bot)$ is often effectively supported on a low-dimensional subspace $\calX_r$ with dimension $d_r \ll d$. In other words, $f$ can be approximated by a function that depends only on $x_r\in \calX_r$.
For a given subspace $\calX_r$ with dimension $d_r$, we denote its complement subspace by $\calX_\bot$ and define projection operators $P_r$ and $P_\bot$ such that ${\rm range}(P_r) = \calX_r$ and ${\rm range}(P_\bot) = \calX_\bot$. A parameter $x$ can be decomposed as
\begin{equation}
\label{eqn:decompose}
x=x_r+x_\bot,\quad x_r = P_r x \in \calX_r,\quad x_\bot = P_\bot x \in \calX_\bot.
\end{equation}
For a density $\nu(x_r, x_\bot)$ on $\reals^d$, we use $\bar \nu(x_r)$ to denote its marginal on $\calX_r$ and $\nu(x_\bot | x_r)$ to denote the conditional density.
This way, the target density can be decomposed as
\[
\pi(x)\equiv \pi(x_r, x_\bot) = \bar \pi(x_r)\pi(x_\bot|x_r),
\] 
where the marginal density and the conditional density take the form 
\begin{equation}
\bar\pi(x_r) = \frac1Z \, \bar\mu(x_r) \int f(x_r, x_\bot) \mu(x_\bot | x_r) d x_\bot \quad {\rm and} \quad \pi(x_\bot|x_r) = \frac{f(x_r, x_\bot) \mu(x_\bot | x_r)}{\int f(x_r, x_\bot) \mu(x_\bot | x_r) d x_\bot}, \label{eqn:pif}
\end{equation}
respectively. With the assumption that the likelihood function $f$ is effectively supported on $\calX_r$, the above decomposition suggests that $\mu(x_\bot|x_r)$ can be a good approximation of  $\pi(x_\bot|x_r)$.
Thus, one can identify the subspace $\calX_r$ and construct a suitable {\it lower-dimensional surrogate} density $\bar\apdf_s(x_r)$ to approximate the marginal target density $\bar \pi(x_r)$. This allows one to approximate the full-dimensional target density by
\begin{equation}
\label{eqn:approx_tar}
\apdf_s(x_r, x_\bot) \propto \bar\apdf_s(x_r) \mu(x_\bot|x_r),
\end{equation}
where the subscript $s$ in $\bar\apdf_s(x_r)$ and $\apdf_s(x_r, x_\bot)$ denotes the method for constructing the surrogate density, which will be specified in Section \ref{sec:LIS}. 
The approximate target density $\apdf_s(x_r, x_\bot)$ can be efficiently sampled using a two-step strategy---one can first apply MCMC or SMC to generate samples from the lower-dimensional surrogate density $\bar\apdf_s(x_r)$, and then draw independent samples from the conditional reference density $\mu(x_\bot|x_r)$.

The identification of the subspace $\calX_r$ is the key in constructing approximate densities in the form of \eqref{eqn:approx_tar}.
Several methods based on the derivative information of the likelihood function have been developed for this purpose. Some examples include the use of the Fisher information matrix \cite{cui2014likelihood,cui2020data}, the Hessian matrix of $\log f$ \cite{bui2013computational,martin2012stochastic}, and the gradient of $\log f$ \cite{CKB16,zahm2018certified}.
Here we focus on the analysis of the gradient-based techniques. 
Note that the gradient of the logarithm of the likelihood, $\nabla  \log f(x)$, indicates a local direction at $x$ in which the log-likelihood changes most rapidly, 
and the Gram matrix of $\nabla \log f(X)$ after averaging over all outcomes of $X$ can measure variations of the likelihood function. 
Depending on the choice of the distribution assigned to $X$, different Gram matrices have been considered:
\begin{equation}
\begin{gathered}
\label{eqn:H}
H_{0}:=\int \nabla  \log f(x) \nabla \log f(x)^\top \mu(x)dx,\\
H_{1}:=\int \nabla  \log f(x) \nabla \log f(x)^\top\pi (x)dx.
\end{gathered}
\end{equation}
When the gradient Gram matrix $H_k$, $k\in\{0, 1\}$, is presented, the subspace spanned by the eigenvectors of the largest eigenvalues of $H_k$ preserves most of the variations of $\nabla  \log f(x)$. Thus, the first $d_r$ eigenvectors (which we will refer to as the `leading eigenvectors') of the gradient Gram matrix can be used to construct the subspace $\calX_r$.

Both $H_0$ and $H_1$ can be numerically estimated using Monte Carlo integration. 
The matrix $H_0$ can be simply estimated using independent samples drawn from the reference density $\mu(x)$. 
In comparison, estimation of $H_1$ is more challenging, because samples drawn from the target density $\pi(x)$ are needed. One may apply importance sampling 
\[
H_1=\frac1Z\int \nabla  \log f(x) \nabla \log f(x)^\top f(x)\mu(x)dx,
\]
so that samples from $\mu(x)$ weighted by the likelihood function can be used to estimate $H_1$. However, the likelihood $f(x)$ may concentrate in a small region for problems with informative data, and thus the above importance sampling formula may suffer from a low effective sample size. In this case, adaptive MCMC sampling or SMC sampling can be used to estimate $H_1$.
At first glance, it appears that the matrix $H_1$ is not an effective way to identify the subspace $\calX_r$. However, our analysis explains why using $H_1$ instead of $H_0$ leads to a more accurate approximation of the subspace $\calX_r$.

\subsection{Posterior approximation via marginalization}
\label{sec:LIS}

Given a subspace $\calX_r$, here we discuss three methods for building the lower-dimensional surrogate density. 
A natural choice is to use the marginal density in \eqref{eqn:pif}. 
\begin{defn}[Marginal likelihood]\label{def:pif} 
By marginalizing the likelihood function over the complement subspace $\calX_\bot$, one has
\begin{equation}
\label{eqn:fr}
\bar f(x_r):= \int f(x_r, x_\bot) \mu(x_\bot|x_r)dx_\bot
 = \E_\mu (f(X)|P_r X=x_r).
\end{equation}
This yields the lower-dimensional surrogate density $\bar\apdf_f(x_r) = \frac1Z\,\bar f(x_r)\mu(x_r)$ and the approximate target density $\apdf_f(x_r, x_\bot) = \bar\apdf_f(x_r) \mu(x_\bot|x_r)$.
\end{defn}
Since the low-dimensional surrogate density $\bar\apdf_f(x_r)$ is equivalent to the marginal target density $\bar\pi(x_r)$, the approximate target density $\apdf_f$ shares the same normalizing constant $Z$ with the full-dimensional target $\pi$. 
Closely related to the marginal likelihood approximation, we also consider the following approximations based on marginalizing the square root of the likelihood and the logarithm of the likelihood. 
\begin{defn}[Radical likelihood]\label{def:pig} 
Defining the square root of the likelihood by $g(x) := \sqrt{f(x)}$, the marginal function $\bar g(x_r) = \E_\mu (g(X)|P_rX=x_r)$ defines the lower-dimensional surrogate density
\begin{equation}
\label{eqn:margsq}
\bar\apdf_g(x_r) = \frac{1}{Z_g} \bar g(x_r)^2 \mu(x_r) ,\quad  Z_g = \int \bar g(x_r)^2 \mu(x_r) dx_r,
\end{equation}
and the approximate target density $\apdf_g(x_r, x_\bot) = \bar\apdf_g(x_r) \mu(x_\bot|x_r)$.
\end{defn}

\begin{defn}[Log-likelihood]\label{def:pil} 
Defining the logarithm of the likelihood by $l(x) := \log f(x)$, the marginal function $\bar l(x_r) = \E_\mu (l(X)|P_rX=x_r)$ defines the lower-dimensional surrogate density
\begin{equation}
\label{eqn:marglog}
\bar\apdf_{l}(x_r) = \frac{1}{Z_l} \exp( \bar l(x_r) ) \mu(x_r),\quad  Z_l = \int  \exp( \bar l(x_r) ) \mu(x_r) dx_r,
\end{equation}
and the approximate target density $\apdf_l(x_r, x_\bot) = \bar\apdf_l(x_r) \mu(x_\bot|x_r)$.
\end{defn}

Note that the combination of $\apdf_l$ and the subspace defined by $H_0$ is also known as the \emph{active subspace} method \cite{CKB16} in the literature. To provide a unified discussion, here we view it as one specific scenario of the LIS.
While using $\apdf_g$ and $\apdf_l$ may seem less natural than using $\apdf_f$, we will show in Sections  \ref{sec:errorLIS}-\ref{sec:MCsubspace} that their theoretical and computational properties differ from those of $\apdf_f$. 
We use the shorthand notation $X\sim\mu(x)$ to indicate that a random variable $X$ follows a probability distribution with the density $\mu(x)$.
In practice, we can generate independent and identically distributed (i.i.d.) samples $X^i_\bot | X_r{=}x_r , i = 1, \ldots, M,$ from the conditional distribution $\mu(x_\bot|x_r)$ using a map $X^i_\bot=T(x_r, W^i)$, where $W^i$ are i.i.d. samples describing the randomness of $X_\bot$ conditioned on $X_r$. Then the marginalization in all of the approximate likelihood functions $\bar f(x_r)$, $\bar g(x_r)$, and $\bar l(x_r)$ can be respectively computed by Monte Carlo integration
\begin{equation}
\label{eqn:fhat}
\bar f^M(x_r):= \frac{1}{M} \sum_{i=1}^{M} f(x_r, X^{i}_\bot), \;\; \bar g^M(x_r):= \frac{1}{M} \sum_{i=1}^{M} g(x_r, X^{i}_\bot), \;\; \bar l^M(x_r):= \frac{1}{M} \sum_{i=1}^{M} l(x_r, X^{i}_\bot).
\end{equation}
Then, we denote the corresponding Monte Carlo version of the densities by 
\[
\bar\apdf_f^M(x_r) \propto \bar f^M(x_r) \mu(x_r), \quad \bar\apdf_g^M(x_r) \propto \bar g^M(x_r)^2 \mu(x_r), \quad \bar\apdf_l^M(x_r) \propto \exp(\bar l^M(x_r)) \mu(x_r),
\]
respectively, and the corresponding Monte Carlo version of the approximate target densities in a similar way.

\subsection{Related work and main contributions}
The use of the approximate target densities naturally introduces errors compared with solutions obtained from the full target densities. Several interconnected factors impact the approximation accuracy. 
Under mild assumptions, this paper aims to assess the following error sources and the performance of related sampling algorithms:
\begin{enumerate}[1.]
\item \textbf{Accuracy of \boldmath$\apdf_s(x), s\in\{f,g,l\}$.} In Section \ref{sec:errorLIS}, we derive error bounds on the difference between the approximate target densities $\apdf_s$ and the full-dimensional target $\pi$, quantified through either estimation error of some test function or various statistical divergences. The highlight is that all these errors can all be bounded by the spectrum of $H_0$ or $H_1$. So if we have the true values of $H_0$ or $H_1$, we can find the optimal projection subspace with performance guarantees.
From the results, we will also observe that the approximation error of the subspace estimated using $H_1$ tends to be smaller than that of the subspace estimated using $H_0$, and it is independent of the normalizing constant. In subspace estimation, this leads to a trade-off between $H_0$ and $H_1$: the former is easier to estimate while the latter tends to have better approximation accuracy.  
\item \textbf{\boldmath Monte Carlo errors of $\apdf_s^M(x), s\in\{f,g,l\}$.} In most practical cases, each of the approximate target densities $\apdf_s(x)$ need to be replaced by the Monte Carlo version $\apdf_s^M(x)$ using  samples drawn from the conditional reference density  $\mu(x_\bot|x_r)$. 
In Section \ref{sec:MCerror}, we show that Monte Carlo averaging incurs an additional error that is about $O(1/\sqrt{M})$ times as large as the error of $\apdf_s(x)$. Therefore, $M$ can be small when the approximation error of $\apdf_s(x)$ is moderate. 
\item\textbf{\boldmath Monte Carlo errors in estimating $H_0$ and $H_1$. }
The Gram matrices $H_0$ and $H_1$ must be approximated by their Monte Carlo estimates $\widehat{H}_0$ and $\widehat{H}_1$, respectively. The resulting sample-averaged subspace $\Xhat_r$ may lead to additional approximation errors. 
In Section \ref{sec:MCsubspace}, we establish bounds on the errors of the approximate target densities using $\Xhat_r$ instead of using the true subspace $\calX_r$. These bounds only depend on the dimension of $\Xhat_r$ and the  variances of $H_0$ and $H_1$. Importantly, our bounds do not rely on eigenvalue gaps, which is a typical assumption used in dimension reduction (e.g., \cite{CKB16}) but may have limited practical applicability. See Remark \ref{rem:subspace}, Figures \ref{fig:ex1_spec} and \ref{fig:ex2_spec_obs}, and \cite{drineas2019low} for further details.
\item \textbf{\boldmath Efficiency of LIS accelerated sampling. }
We can implement MCMC to draw samples from  the low-dimensional surrogate density $\bar\apdf_s(x_r)$, and then augment the low-dimensional samples by adding samples drawn from the conditional reference density $\mu(x_\bot|x_r)$ to obtain samples from the full-dimensional approximate target density $\apdf_s(x)$. 
In Section \ref{sec:MCMC}, we investigate the efficiency of this algorithm, in which Proposition \ref{prop:MCMC} shows the overall efficiency is mostly determined by the MCMC targeting the approximate target density $\apdf_s(x_r)$. 
In Section \ref{sec:SMC}, we further investigate the connection between SMC and LIS, in particular how to use SMC to simplify the estimation of $H_1$. 
\item \textbf{\boldmath Dimension independence.} LIS methods are mostly used in high-dimensional problems, and hence it is important for the error bounds to be dimension independent. In other words, various approximation error bounds should depend only on the effective dimension $d_r$ and some other statistics,  but not on the ambient dimension $d$. We illustrate this is indeed the case in Section \ref{sec:dimension} for a class of linear inverse problems. It also serves as a concrete example to demonstrate the efficacy of our analysis. 
\end{enumerate}

We provide some numerical examples on nonlinear inverse problems to further verify our results in Section \ref{sec:numerics}. 
We allocate most of the technical proofs to the\appname. 

We now discuss some related work that addresses the preceding issues. In \cite{CKB16}, Problems 1--3 are investigated in the context of the active subspace method, which employs $\apdf_l$ with $\calX_r$ estimated from $H_0$, using the Hellinger distance. An analysis similar to that of \cite{CKB16} has also been developed for function approximation problems with $H_0$ in \cite{parente2020generalized}. 
The work of \cite{zahm2018certified} investigated Problems 1--3 for $\apdf_f$ with $\calX_r$ estimated from $H_1$ using the Kullback--Leibler (KL) divergence. 
For Problem 1, our analysis establishes new error bounds of $\apdf_s, s\in\{f,g\}$ with $\calX_r$ estimated from both $H_1$ and $H_0$ based on the Hellinger distance. Using the bounds on the Hellinger errors, we can establish new sharp bounds on the expected Monte Carlo errors in $\apdf_s^M, s\in\{f,g\}$ for Problem 2. 
This analysis also sheds light on the trade-off between $H_0$ and $H_1$. 
For the sake of completeness, we also establish the error bound of $\apdf_l$ and $\apdf_l^M$ with $\calX_r$ estimated from $H_0$ based on the KL divergence.
Moreover, our analysis for Problem 3 does not require the eigenvalue gap condition, which is assumed in \cite{CKB16,zahm2018certified} and not easily fulfilled in applications (see Remark \ref{rem:subspace}, Figures \ref{fig:ex1_spec} and \ref{fig:ex2_spec_obs} and \cite{drineas2019low}).
Low-rank matrix approximation methods that do not require an eigenvalue gap have also been studied, e.g. in \cite{drineas2019low}, but not for sample-averaged subspace estimation. 
Beyond Problems 1--3, our analysis also enables us to investigate Problems 4 and 5, which have practical significance but have not been previously addressed.  

\section{Accuracy of approximate target densities}
\label{sec:errorLIS}

Our starting point is to establish bounds on the errors of approximate target densities $\apdf_s(x), s \in \{f,g,l\}$ in Section \ref{sec:LIS}. We consider two forms to quantify the approximation errors. 
The first way is through the estimation error. Suppose the goal is to estimate $\E_\pi [h]$ for some function of interest $h$. The approximate density $\apdf_s$ yields an approximate estimate $\E_{\apdf_s} [h]$ that has the estimation error 
\begin{equation}
\label{eqn:eh}
\Error_h(\pi,\apdf_s):=\left|\E_\pi [h]-\E_{\apdf_s} [h]\right|.
\end{equation}
The second way is via statistical divergences, which are also known as $f$-divergences. Some popular choices include the (squared) Hellinger distance
\[
\Hell(\pi,\nu)^2=\frac12\int \left(\sqrt{\frac{\pi(x)}{\lambda(x)}}-\sqrt{\frac{\nu(x)}{\lambda(x)}}\right)^2 \lambda(x)dx.
\]
where $\lambda$ is a reference density such as the Lebesgue density; and the KL divergence
\[
D_{KL}(\pi,\nu)=\int \log \frac{\pi(x)}{\nu(x)}\pi(x)dx.%
\]
We present in Lemma \ref{lem:hel2exp} a few results regarding the relationship between these divergences and their connections with the estimation error $\Error_h$. 
Various error forms can be useful for applying dimension reduction in different inference tasks, as each inference task often has its ``preferred'' way to quantify the error.
For example, the optimization problems in transport maps \cite{bigoni2019greedy,marzouk2016sampling,spantini2018inference} and Stein variational methods \cite{detommaso2018stein,liu2016stein} are formulated using the KL divergence, tensor train \cite{cui2020deep} and other approximation methods, e.g., \cite{lie2019error}, give bounds in terms of the Hellinger distance, and the min-max formulation in density estimation methods such as \cite{tabak2020conditional,tabak2013family,trigila2016data} relies on the estimation error in \eqref{eqn:eh}.
Unless otherwise specified, we only consider the estimation error and statistical divergences of the full-dimensional approximate target densities $\apdf_s(x), s \in \{f,g,l\}$ rather than their lower-dimensional counterparts $\bar\apdf_s(x_r)$. 

For different combinations of approximate target densities, $\apdf_s(x), s \in \{f,g,l\}$, and subspace construction methods, $H_k, k \in \{0,1\}$, our first result discusses the {\it a priori} estimate of either $\Error_h(\pi,\apdf_s)$ or $D_{(\,\cdot\,)}(\pi,\apdf_s)$ using the subspace $\calX_r$ and spectral information of $H_k$. 
Intuitively, the approximation error is related to the sum of the residual eigenvalues of $H_k$, which is denoted by 
\begin{equation}
\label{eqn:res}
\res(\calX_r, H_k):=\trace(P_\bot H_k P_\bot),
\end{equation}
where $P_\bot$ is the projector defined in \eqref{eqn:decompose}. Note that \eqref{eqn:res} is well defined for any linear subspace $\calX_r$ and computable for a given $H_k$, whereas many statistical divergences do not have closed-form formulas. 

To build a connection between approximation errors of $\apdf_s(x)$ and the residual function $\res(\calX_r, H_k)$, we assume the reference density $\mu(x)$ is compatible with the subspace $\calX_r$ in the following sense:
\begin{aspt}
\label{aspt:poincare}
The conditional reference density 
$\mu(x_\bot|x_r)$ satisfies a $\kappa$-Poincar\'{e} inequality: for all $x_r$ and any $\mathbb{C}^1$ function $h$:
\[
\var_{\mu(x_\bot|x_r)}(h)\leq \kappa \int \|\nabla h(x_r, x_\bot)\|^2 \mu(x_\bot|x_r) dx_\bot. 
\]
\end{aspt}

Assumption \ref{aspt:poincare} asserts a Poincar\'{e}-type inequality that is modified for our subspace approximations. In probability theory, it is well known that Poincar\'{e}-type inequalities hold for any strongly log-concave density $\mu(x)$. We refer the readers to \cite{bobkov1999isoperimetric} for a summary and its connection to other inequalities such as the Brascamp--Lieb inequality \cite{brascamp1976} and the logarithmic Sobolev inequality \cite{bobkov2000brunn,gross1975logarithmic,ledoux1994simple,otto2000generalization}. 
In the following proposition, we provide a concrete example of Assumption \ref{aspt:poincare} for the case that $\mu(x)$ is a slight perturbation from a strongly log-concave density. A similar result can be found in \cite[Corollary F.4]{zahm2018certified}. We provide it here for the sake of completeness.
\begin{prop}
\label{prop:logconcave}
Suppose  $\mu(x)\propto \exp(-V(x)-U(x))$ and there are constants $c,B>0$ such that 
\begin{itemize}
\item For any $x$, the minimal eigenvalue of the Hessian $\nabla^2 V(x)$ is larger than $c$;
\item The variation in $U$ is bounded in the sense that $\exp(\sup_xU(x)-\inf_x U(x))\leq B$;
\end{itemize}
Then  Assumption \ref{aspt:poincare} holds with $\kappa=B^2/c$.
\end{prop}
\begin{proof}
See\smartref{proof:logconcave}.
\end{proof}
\begin{table}[h]
\caption{A summary of approximation error bounds. The second column indicates the functions marginalized by the approximate target densities.}
\label{tab:lis}
\begin{center}
\begin{tabular}{ l | l | c | c }
\hline
approximation method & marginalization    & approximation errors & upper bounds   \\
\hline
$H_1$ and $\apdf_f$ & $f$  & $\Error_h,\Hell, (\sqrt{D_{KL}})$ & O($\res(\calX_r,H_1)^{\frac12}$) \\
\hline
$H_1$ and $\apdf_g$ & $g = \sqrt{f}$ & $\Error_h,\Hell$  & O($\res(\calX_r,H_1)^{\frac12}$)\\
\hline
$H_0$ and $\apdf_f$ & $f$  & $\Error_h,\Hell$ & O($\res(\calX_r,H_0)^{\frac12}$) \\
\hline
$H_0$ and $\apdf_g$ & $g = \sqrt{f}$ & $\Error_h,\Hell$  & O($\res(\calX_r,H_0)^{\frac12}$)\\
\hline
$H_0$ and $\apdf_l$ & $l = \log f$ & $\Error_h, \Hell,\sqrt{D_{KL}} $  & O($\res(\calX_r,H_0)^{\frac14}$) \\
\hline
\end{tabular}
\end{center}
\end{table}

Under Assumption \ref{aspt:poincare} we will show that  $\Error_h(\pi,\apdf_s)$ and $D_{(\,\cdot\,)}(\pi,\apdf_s)$ can be upper bounded by a fractional power of $\res(\calX_r, H_k)$. We summarize the results in Table \ref{tab:lis}. The first row indicates that if the Gram matrix $H_1$ and the approximate density $\apdf_f$ are used, then $\Error_h(\pi,\apdf_f)$ is bounded by $O(\sqrt{\res(\calX_r,H_1)})$. The same applies to other entries in the table. Note that we have written parentheses around $\sqrt{D_{KL}}$ for $H_1$ and $\apdf_f$, since this scenario has been analyzed in \cite{zahm2018certified} under a similar assumption. Therefore we do not discuss bounds for $D_{KL}(\pi,\apdf_f)$ and focus on other bounds that have yet to be analyzed. 
We first consider the approximations $(H_1, \apdf_f)$ and $(H_1, \apdf_g)$ as described in Definitions \ref{def:pif} and \ref{def:pig}, respectively. 

\begin{prop}
\label{prop:pivar}
For a given subspace $\calX_r$, the expected conditional variance of the radical likelihood function $g = \sqrt{f}$ provides the following upper bounds:
\begin{enumerate}[1)]
\item $\displaystyle \Hell(\pi, \apdf_f)^2 \leq  \frac{1}{Z}  \int  \var_{\mu(x_\bot|x_r)} \big[ g \big] \mu(x_r)dx_r$.
\item $\displaystyle \Hell(\pi, \apdf_g)^2 \leq  \frac{1}{Z}  \int  \var_{\mu(x_\bot|x_r)} \big[ g \big] \mu(x_r)dx_r$.
\end{enumerate}
In addition, the normalizing constants $Z$ and $Z_g$ satisfy $Z \geq Z_g$.
\end{prop}
\begin{proof}
See\smartref{proof:pivar}.
\end{proof}

\begin{thm}
\label{thm:pif}\label{thm:pig}
Suppose the approximate densities $\apdf_f$ and $\apdf_g$ are obtained using a subspace $\calX_r$ constructed from the matrix $H_1$. Under Assumption \ref{aspt:poincare}, we have the following:
\begin{enumerate}[1)]
\item The Hellinger distance between $\pi$ and $\apdf_f$ is bounded by 
\begin{equation}
\label{eqn:Hpif}
\Hell(\pi, \apdf_f) \leq \frac12 \sqrt{\kappa \res(\calX_r,H_1)}.
\end{equation}
\item The estimation error with any $L^2$ integrable function $h$ is given by 
\[
\Error_h(\pi,\apdf_f)\leq \sqrt{ \tfrac\kappa 2(\E_{\pi} [h^2]+\E_{\apdf_f}[h^2]) \res(\calX_r,H_1)}.
\]
\item The above two claims also hold for the approximation $\apdf_g$.
\end{enumerate}
\end{thm}
\begin{proof} 
See\smartref{proof:pig}.
\end{proof}

Although the result of Theorem \ref{thm:pif} claim 1) can also be obtained from Lemma \ref{lem:hel2exp} claim 2) and Corollary 1 of \cite{zahm2018certified} (which uses the logarithmic Sobolev inequality), our proof offers additional insights into the subspace construction. Proposition \ref{prop:pivar} connects the error of approximate target densities with the $\kappa$-Poincar\'{e} inequality via the expected conditional variance. This may also lead to new subspace construction techniques beyond the gradient-based methods. 

\begin{rem}
\label{rem:25}
Recalling the definitions of $H_k$, we have $H_1\leq \frac{1}{Z} \sup_x f \, H_0$. Thus, we have a direct corollary of Theorem \ref{thm:pif} for the case where the subspace $\calX_r$ is constructed from the matrix $H_0$:
\begin{equation}
\label{tmp:Hpif}
\Hell(\pi, \apdf_f) \leq \frac12 \sqrt{\frac{\kappa\,\sup_x f}{Z} \res(\calX_r,H_0)}, \quad \Hell(\pi, \apdf_g) \leq \frac12 \sqrt{\frac{\kappa\,\sup_x f}{Z} \res(\calX_r,H_0)}.
\end{equation}
Similar bounds for the $L^2$ distance between $\log \apdf_l$ and $\log \pi$ assuming $\sup_x f = 1$ can be found in Theorem 3.1 \cite{CKB16} with a more complicated pre-constant. 
In problems where the likelihood function $f$ concentrates in a small region, the associated normalizing constant $Z$ can be small. This way, the constant on the right-hand side of \eqref{tmp:Hpif} can have a large value. In contrast, the only constant in \eqref{eqn:Hpif} is $\kappa$, which is of value $0.5$ when the reference density is the standard Gaussian distribution. This partially explains why using $H_0$ can be suboptimal.
Following this observation, we predict that the reduced subspace from $H_1$ will perform better than the one from $H_0$, especially when the likelihood function has concentrated support. This will be verified in our numerical examples. %
\end{rem}

For the approximate target density $\apdf_l$, one can obtain bounds on the associated approximation errors only if the matrix $H_0$ is used to construct the subspace. In contrast, for the approximate target densities $\apdf_f$ and $\apdf_g$, error bounds can be obtained using both $H_1$ and $H_0$. See Table \ref{tab:lis}. The error bounds for $\apdf_l$ are in general weaker than those for $\apdf_l$ and $\apdf_g$---they depend on additional constants that can take large values and the exponent of $\res(\calX_r, H_0)$ in the error bounds is $1/4$.
\begin{thm}
\label{thm:pil}
Suppose the approximate density $\apdf_l$ is obtained using a subspace $\calX_r$ constructed from the matrix $H_0$. Under Assumption \ref{aspt:poincare}, we have the following:
\begin{enumerate}[1)]
\item The error in KL-divergence is bounded by
\[
D_{KL}(\pi,\apdf_l)\leq \frac{\sqrt{\kappa}\|f\|_{2,\mu}}{Z}\sqrt{\res(\calX_r, H_{0})},\quad  \|f\|_{2,\mu}:=\sqrt{\int f^2(x)\mu(x) dx}\geq Z. 
\]
This also leads to an upper bound in Hellinger distance, since $D_{H}(\pi,\apdf_l)\leq\sqrt{\frac12D_{KL}(\pi,\pi_l)}$. 
\item The estimation error is bounded by
\[
|\E_{\pi} [h]-\E_{\apdf_l} [h]|\leq (\E_{\pi} [h^2]+\E_{\apdf_l}[h^2])^{\frac12}\sqrt{\frac{\|f\|_{2,\mu}}{Z}} (\kappa\res(\calX_r, H_{0}))^{\frac14}.
\]
\end{enumerate}
\end{thm}
\begin{proof}
See\smartref{proof:pil}.
\end{proof}

\section{Monte Carlo error of approximate target densities}
\label{sec:MCerror}
To construct the approximate densities $\apdf_s, s \in \{f,g,l\}$, the marginalization in the lower-dimensional likelihood approximations (cf. Definitions \ref{def:pif}--\ref{def:pil}) often needs to be computed by Monte Carlo integration, where i.i.d. samples $X^i_\bot|X_r{=}x_r, i = 1,\ldots, M,$ drawn from the conditional reference density $\mu(x_\bot|x_r)$ are used.
To estimate the expected errors of the Monte Carlo version of the approximate densities, denoted by $\apdf^M_s, s \in \{f,g,l\}$, we consider the expectation of some function in the form of 
\[
h^M := \int h(x_r, X_\bot^1, \ldots, X_\bot^M) \bar\mu(x_r) d x_r.
\]
By generating conditional samples from $\mu(x_\bot|x_r)$ using a map $X_\bot=T(x_r,W)$, where $W \sim \nu(w)$, we can express the expectation of $h^M$ over all possible outcomes of $X^i_\bot|X_r{=}x_r, i = 1,\ldots, M,$ as
\begin{align*}
\E_M \left[h^M\right] & := \int\cdots \int h\big(x_r,T(x_r,w^1),\ldots, T(x_r,w^M)\big) \bar\mu(x_r) d x_r \bigg(\prod_{i=1}^M \nu(w^i) \bigg) d w^1 \cdots d w^M,
\end{align*}
in order to remove the conditional dependency of $X^i_\bot$ on $X_r$ in the expectation.
The following theorems reveal the accuracy of the sample-averaged approximate densities $\apdf_s^{M}$.
\begin{thm}
\label{thm:MCerrorfg}
Suppose the approximate densities $\apdf_f$ and $\apdf_g$ are obtained using a subspace $\calX_r$ constructed from the matrix $H_1$. Under Assumption \ref{aspt:poincare}, the following bounds hold:
\begin{enumerate}[1)]
\item The expected Hellinger distance between $\apdf_g^M$ and $\apdf_g$ satisfies
\[
\E_M \left[ \Hell(\apdf_g^M,\apdf_g) \right] \leq \frac{\sqrt{2\kappa Z}}{\sqrt{Z_g M}} \sqrt{\res(\calX_r, H_1)}.
\]
\item Given the conditional likelihood $f(x_\bot|x_r):= f(x_\bot,x_r) / \bar f(x_r)$ and 
\(
C_f = \sup_{x_r} \sup_{x_\bot}f(x_\bot|x_r), 
\)
then the expected Hellinger distance between $\apdf_f^M$ and $\apdf_f$ satisfies
\[
\E_M \left[ \Hell(\apdf_f^M,\apdf_f)\right] \leq \frac{\sqrt{2 \kappa C_f}}{\sqrt{M}}\sqrt{\res(\calX_r, H_1)}.
\]
\end{enumerate}
\end{thm}
\begin{proof}
See\smartref{sec:proof_thm3fg}.
\end{proof}

Note that claim 2) of Theorem \ref{thm:MCerrorfg} needs an additional assumption on the supremum of $f(x_\bot|x_r)$, while claim 1) does not, showing the analytical advantage of $\apdf_g$. The requirement that $f(x_\bot|x_r)$ is bounded is not restrictive in practice, since the conditional likelihood is expected to be flat in the complement subspace of $\mathcal{X}_r$.
Since the Hellinger distance enjoys the triangle inequality, we have 
\[
\E_M \left[ \Hell(\apdf_s^M,\pi) \right]\leq \E_M \left[ \Hell(\apdf_s^M,\apdf_s) \right] + \Hell(\apdf_s,\pi),\quad s\in\{f,g\}.
\]
This way, Theorem \ref{thm:MCerrorfg} and Theorem \ref{thm:pif} together reveal that the Monte Carlo averaging used in $\apdf_s^{M}(x), s \in \{f,g\}$ incurs an additional error that is about $O(1/\sqrt{M})$ as large as the error of $\apdf_s(x), s \in \{f,g\}$.
Since the KL-divergence does not satisfy the triangle inequality, we directly establish the bound on $\E_M\left[ D_{KL}(\pi,\apdf_l^M)\right]$ as follows. 

\begin{thm}
\label{thm:MCerrorl}
Suppose the approximate density $\apdf_l$ is obtained using a subspace $\calX_r$ constructed from the matrix $H_0$. Under Assumption \ref{aspt:poincare}, the expected $L^2$ error of the marginalized log-likelihood is bounded by 
\[
\E_M\left[ \left(\int (\bar l^M(x_r)-\bar l(x_r))^2\mu(x_r)dx_r \right)^\frac12\right]\leq \frac{\sqrt{\kappa}}{\sqrt{M}} \sqrt{\res(\calX_r,H_0)}. 
\]
The expected KL-divergence of $\pi$ from the approximation $\apdf_l^M$ is bounded by 
\[
\E_M\left[ D_{KL}(\pi, \apdf^M_{l}) \right] \leq \frac{\sqrt{\kappa} \|f\|_{2,\mu}}{Z} \left( 1 + \frac1{\sqrt{M}} \right) \sqrt{\res(\calX_r,H_0)}. 
\]
\end{thm}
\begin{proof}
See\smartref{sec:proof_thm3l}.
\end{proof}

%
Theorems \ref{thm:MCerrorfg} and \ref{thm:MCerrorl} reveal that the sample size $M$ does not need to be  large in practice, as the error $D_{(\cdot)}(\apdf^M_s,\pi)$ is dominated by the projection residual $\res(\calX_r,H_0)$, which is independent of $M$.

\section{Sample-based Gram matrix estimation}
\label{sec:MCsubspace}
Given a subspace $\calX_r$ constructed from the matrix $H_k, k \in \{0,1\}$, 
Sections \ref{sec:errorLIS} and \ref{sec:MCerror} show that the approximation errors are bounded by $\res(\calX_r,H_k)$. 
Since the gradient Gram matrix $H_k$ has to be estimated through Monte Carlo integration in practice, here we provide rigorous estimates of how the sampling error of $H_k$ affects the overall approximation error.

We start with a general importance sampling formulation for estimating the gradient Gram matrix. Suppose we can generate i.i.d. samples $X^i, i = 1, \ldots, \Mmu,$ from a density $\nu$, then the Monte Carlo estimators of $H_0$ and $H_1$ are given by 
\begin{equation}
\label{eqn:Hhat1}
\begin{gathered}
\Hhat_0=\frac1{\Mmu}\sum_{i=1}^{\Mmu} \nabla \log f(X^i)\nabla \log f(X^i)^\top \frac{\mu(X^i)}{\nu(X^i)}, \\
\Hhat_1=\frac1{\Mmu}\sum_{i=1}^{\Mmu} \nabla \log f(X^i)\nabla \log f(X^i)^\top \frac{\pi(X^i)}{\nu(X^i)}. 
\end{gathered}
\end{equation}
For some function $h^m(X^1, \ldots, X^m)$ where $X^i\sim\nu(x)$ are i.i.d. samples, we denote the expectation of $h^m$ over all sampling outcomes of $X^i, i = 1, \ldots, \Mmu$ by
\[
\E_\nu [h^m]=\int \cdots \int h(x^1, \ldots, x^m) \bigg(\prod_{i = 1}^m \nu(x^i)\bigg)dx^1 \cdots dx^m.
\]
For example, we have $\E_\nu [\Hhat_k]=H_k$. We also define the one-sample variance of the matrix estimators under the Frobenius norm $\|\,\cdot\,\|_F$ by
\begin{align}
V(H_0, \nu)& :=\sum_{i,j=1}^d\text{var}_{X\sim \nu}\left[\partial_i\log f(X)\partial_j \log f(X)\frac{\mu(X)}{\nu(X)}\right]=\Mmu \E_\nu \left[ \|\Hhat_0-H_0 \|_F^2\right] ,\\
V(H_1, \nu) &:=\sum_{i,j=1}^d\text{var}_{X\sim \nu}\left[\partial_i \log f(X)\partial_j \log f(X) \frac{\pi(X)}{\nu(X)}\right]=\Mmu \E_{\nu} \left[\|\Hhat_1-H_1 \|_F^2\right]. 
\end{align}

Recall that in the LIS procedure, the reduced subspace $\calX_r$ is obtained as the $d_r$ dimensional leading eigensubspace of $H_k$. The associated residual is given by $\res(\calX_r,H_k)=\sum_{i=d_r+1}^d \lambda_i(H_k)$. 
In practice, we can only obtain the leading eigensubspace $\Xhat_r$ generated by the sample-averaged matrix $\Hhat_k$. Thus, we must consider alternative residuals based on $\Xhat_r$ and $\Hhat_k$. We first consider the ``effective'' residual $\res(\Xhat_r,H_k)$, which provides upper bounds on the approximation errors induced by the estimated subspace $\Xhat_r$, as given in Table \ref{tab:lis}. Note that the true matrix $H_k$ must be used here.
We aim to compare the residual $\res(\Xhat_r,H_k)$ to the residual $\res(\calX_r,H_k)$ to understand the impact of the sample-based estimation of the subspace $\Xhat_r$.
Since we cannot compute the effective residual $\res(\Xhat_r,H_k)$ in practice, we must use the computable residual $\res(\Xhat_r,\Hhat_k)=\sum_{i=d_r+1}^d \lambda_i(\Hhat_k)$ to determine the truncation dimension $d_r$. Thus, we also aim to estimate the difference between $\res(\Xhat_r,H_k)$ and $\res(\Xhat_r,\Hhat_k)$ to understand the reliability of the computable residual $\res(\Xhat_r,\Hhat_k)$. 
The following variation of the Davis--Kahan Theorem \cite{yu2015useful} is useful for addressing these questions.

\begin{lem}
\label{lem:matrix}
Let $\Sigma \in \reals^{d\times d}$ and $\Sigmahat\in \reals^{d\times d}$ be two positive semidefinite matrices. Let $\Xhat_r$ be the $d_r$-dimensional leading eigensubspace of $\Sigmahat$ and $\widehat{P}_\bot$ be the orthogonal projection to its complementary subspace. Then the following hold:
\begin{enumerate}[1)]
\item $\res(\Xhat_r, \Sigma)=\widehat{P}_\bot\Sigma \widehat{P}_\bot \leq \sum_{i=d_r+1}^d \lambda_i(\Sigma)+ 2\sqrt{d_r} \|\Sigmahat-\Sigma\|_F$.
\item $\res(\Xhat_r, \Sigma)=\widehat{P}_\bot\Sigma \widehat{P}_\bot \leq \sum_{i=d_r+1}^d \lambda_i(\Sigmahat)+ \sqrt{d_r} \|\Sigmahat-\Sigma\|_F +\tr(\Sigma-\Sigmahat)$.
\end{enumerate}
\end{lem}

\begin{proof}
See\smartref{proof:matrix}. \end{proof}
A unique feature of these bounds is that they do not depend on eigenvalue gaps, which are usually necessary for finding the subspace correctly. Further implications will be discussed in Remark \ref{rem:subspace}.

\begin{thm}
\label{thm:sample}
Under Assumption \ref{aspt:poincare}, suppose $\Hhat_k, k\in\{0,1\}$ is computed by \eqref{eqn:Hhat1} and the computed subspace $\Xhat_r$ is spanned by the $d_r$ leading eigenvectors of $\Hhat_k$ and the true subspace $\calX_r$ is spanned by the $d_r$ leading eigenvectors of true $H_k$. Then the following bounds hold: 
\begin{enumerate}[1)]
\item The effective residual satisfies
\(
\displaystyle \E_\nu\left[\res(\Xhat_r, H_k)\right] - \res(\calX_r, H_k)\leq  \frac{2}{\sqrt{\Mmu}}\sqrt{d_r V(H_k,\nu)}.
\)
\item The computable residual satisfies:
\(
\displaystyle  \E_\nu \left[\res(\Xhat_r, H_k)- \res(\Xhat_r,\Hhat_k) \right]\leq \frac{1}{\sqrt{\Mmu}}\sqrt{d_r V(H_k,\nu)}.
\)
\end{enumerate}
\end{thm}

\begin{proof}%
Using claim 1) of Lemma \ref{lem:matrix} and the identity 
\[
\E_\nu [\|H_k-\Hhat_k\|_F]\leq \sqrt{\E_\nu [\|H_k-\Hhat_k\|_F^2]}=\frac{\sqrt{d_r V(H_k,\nu)}}{\sqrt{\Mmu}},
\]
claim 1) directly follows. 
Using the fact that $\E_\nu [\Hhat_k]=H_k$, $\E_\nu [\tr (\Hhat_k-H_k)]=0$, claim 2) follows from claim 2) of Lemma \ref{lem:matrix}.
\end{proof}

Claim 1) of Theorem \ref{thm:sample} shows that the difference between the expected effective approximation residual using the sample average defined in \eqref{eqn:Hhat1} and the true approximation residual is of order  $1/\sqrt{\Mmu}$,  where the prefactor is controlled by the variance $V(H_k,\nu)$ and the dimension $d_r$.
This reveals that, with increasing $\Mmu$, the approximation accuracy of the subspace given by the sample averages becomes closer to that of the true subspace. 
Claim 2) of Theorem \ref{thm:sample} shows that the computable residual $\res(\Xhat_r,\Hhat_k)=\sum_{i=d_r+1}^d \lambda_i(\Hhat_k)$ provides a reliable estimate of the approximation residual in expectation, where the reliability is controlled by the sample size $\Mmu$, the variance $V(H_k,\nu)$ and the subspace dimension $d_r$.

In the following corollary, we combine Theorem \ref{thm:sample} with the results in Section \ref{sec:errorLIS} to address a practical problem: given the estimated $\Hhat_k, k\in\{0,1\}$, quantify the associate LIS approximation error for estimating $\E_\pi[h]$. 
We use $\hat\apdf_s(x), s \in \{f,g,l\},$ to denote the approximate target densities defined by an estimated subspace $\Xhat_r$.
Similar upper bounds for the statistical divergences discussed in Section \ref{sec:errorLIS} can also be established. We do not present them for the sake of conciseness. 

\begin{cor}
\label{cor:sample}
For any bounded test function $h$, the estimation errors satisfy the following bounds:
\begin{enumerate}[1)]
\item Given $\Xhat_r$ obtained from $\Hhat_1$ the resulting approximate target densities $\hat{\apdf}_s, s \in \{f,g\},$ satisfy
\begin{align*}
\E_\nu\left[\Error_h(\pi,\hat{\apdf}_s)\right] & \leq \kappa^\frac12 \Bigg( \frac{\E_{\pi} [h^2]+\E_\nu \left[\E_{\hat{\apdf}_s}[h^2]\right] }2\Bigg)^\frac12 \Bigg( \res(\calX_r, H_1)+\frac{2\sqrt{d_r V(H_1,\nu)}}{\sqrt{\Mmu}}\Bigg)^\frac12,\\
\E_\nu\left[\Error_h(\pi,\hat{\apdf}_s)\right] & \leq \kappa^\frac12 \Bigg( \frac{\E_{\pi} [h^2]+\E_\nu \left[\E_{\hat{\apdf}_s}[h^2]\right] }2\Bigg)^\frac12 \Bigg( \E_\nu\bigg[\sum_{i=d_r+1}^d \lambda_i(\Hhat_1)\bigg]+\frac{\sqrt{d_r V(H_1,\nu)}}{\sqrt{\Mmu}}\Bigg)^\frac12.
\end{align*}
\item Given $\Xhat_r$ obtained from $\Hhat_0$, the resulting approximate target density $\hat{\apdf}_l$ satisfies
\begin{align*}
\E_{\nu}\left[\Error_h(\pi,\hat{\apdf}_l)\right] & \leq \kappa^\frac14 \Bigg(\frac{\|f\|_{2,\mu} \left(\E_{\pi} [h^2]+\E_{\nu}\left[ \E_{\hat{\apdf}_l}[h^2]\right]\right)}{Z}\Bigg)^\frac12  \Bigg( \res(\calX_r, H_0)+\frac{2\sqrt{d_r V(H_0,\nu)}}{\sqrt{\Mmu}}\Bigg)^{\frac14} ,\\
\E_{\nu}\left[\Error_h(\pi,\hat{\apdf}_l)\right] & \leq\kappa^\frac14 \Bigg(\frac{\|f\|_{2,\mu} \left(\E_{\pi} [h^2]+\E_{\nu}\left[ \E_{\hat{\apdf}_l}[h^2]\right]\right)}{Z}\Bigg)^\frac12 \Bigg(\E_\nu\Bigg[ \sum_{i=d_r+1}^d \lambda_i(\Hhat_0)\Bigg]+\frac{\sqrt{d_r V(H_0,\nu)}}{\sqrt{\Mmu}}\Bigg)^{\frac14}.
\end{align*}
\end{enumerate}
\end{cor}
\begin{proof}
For claim 1), recall that Theorem \ref{thm:pif} applies to any given subspace, including $\Xhat_r$, so we have 
\[
\E_\nu [\Error_h(\pi,\hat{\apdf}_s)]\leq \sqrt{\tfrac\kappa 2\big(\E_{\pi} [h^2]+\E_\nu\big[\E_{\hat{\apdf}_s}[h^2]\big]\big) \E_{\nu}[\res(\Xhat_r,H_1)]},
\]
by the Cauchy--Schwarz inequality. Then we apply the upper bound of $\E_{\nu}[\res(\Xhat_r,H_1)]$ in Theorem \ref{thm:sample}  to obtain the corollary. Claim 2) can be shown similarly.
\end{proof}

\begin{rem}
\label{rem:subspace}
It is worth pointing out that our results did not discuss the difference between the estimated subspace $\Xhat_r$ and the subspace $\mathcal{X}_r$ obtained from the true $H_k$. (For a mathematical definition of this difference, one can refer to Theorem 1 in \cite{yu2015useful}).
 While finding the difference is possible using tools like the Davis--Kahan theorem, this difference is usually inversely proportional to the eigenvalue gap, i.e., $O(\lambda_1(H_k) / (\lambda_{d_r}(H_k)-\lambda_{d_{r+1}}(H_{k})) )$. See Theorem 1 in \cite{yu2015useful} for example. 
This quantity can be very large if the matrix $H_k$ does not have a significant eigenvalue gap near the truncation dimension $d_r$. This is often observed in various applications, e.g., \cite{bui2013computational, CH08, FlathEtAl11, HH07}, where eigenvalues of $H_k$ decay rapidly.
For example, if $\lambda_m(H_k) = m^{-2}$, then  $(\lambda_{d_r}(H_k)-\lambda_{d_{r+1}}(H_{k}))^{-1}= O(d_r^3)$. 
For the numerical examples in Section \ref{sec:numerics}, we observe that the eigenvalue gap is in the order of $10^{-4}$ to $10^{-5}$ for a moderate $d_r$. In other words, it is impractical to recover the subspace exactly. 

Fortunately, different eigenvectors of $\Xhat_r$ have very different impact on the resulting approximate target density $\hat{\apdf}_s$.  Intuitively, the accuracy of $\hat{\apdf}_s$ has little dependence on eigenvectors of $H_k$ with close-to-zero eigenvalues, because they contribute little to $H_k$. But having accurate estimations for these eigenvectors is the most difficult, since their eigenvalues are close to each other. Our analysis avoids considering the difference between  $\Xhat_r$ and $\calX_r$ and focuses on the difference between $\pi$ and $\hat\apdf_s$, since the latter does not need the eigenvalue gap and is the purpose of identifying the subspace.
\end{rem}

In importance sampling, the proposal density $\nu$ plays an important role in the sampling accuracy. In particular, the one-sample importance sampling variance of the Gram matrix can be bounded by the likelihood ratio between $\nu$ and $\mu$, or $\nu$ and $\pi$, as follows. 
\begin{prop}
\label{prop:Vbound}
We have the following upper bounds for the sampling variance of the Gram matrix
\[
V(H_0, \nu)\leq \E_{X\sim \nu}\left[ \|\nabla \log f(X)\|^4\frac{\mu(X)^2}{\nu(X)^2}\right],\quad
V(H_1, \nu)\leq \E_{X\sim \nu}\left[ \|\nabla \log f(X)\|^4\frac{\pi(X)^2}{\nu(X)^2}\right].
\]
\end{prop}
\begin{proof}
See\smartref{proof:Vbound}.
\end{proof}

Proposition \ref{prop:Vbound} shows that one should make the ratios, $\tfrac{\mu}{\nu}$ and  $\tfrac{\pi}{\nu}$, close to one in order to minimize the sampling variance for $H_0$ and $H_1$, respectively. 
For estimating $H_0$, we can naturally use the reference distribution, which is easy to sample from, as the biasing distribution, i.e., $\nu=\mu$. %
For estimating $H_1$, the second inequality in Proposition \ref{prop:Vbound} suggests that using the reference distribution may not be a feasible strategy. 
Consider a scenario where the likelihood function is bounded as $\sup_x f = 1$ and the gradient of the log-likelihood is bounded as $\sup_x \|\nabla \log f(x)\| = M_f$. Using $\nu=\mu$, the variance $V(H_1, \mu)$ is inversely quadratic in the normalizing constant $Z$, i.e., 
\(
V(H_1, \mu)\leq M_f^4 \big/ Z^2.
\)
For a target density concentrating in a small region of the parameter space, the normalizing constant $Z$  can take a small value, and thus the variance $V(H_1, \mu)$ can take a rather large value. 
This way, alternative strategies such as MCMC and SMC must be used to adaptively collect samples from the target distribution for estimating $H_1$, while the intermediate estimation of $H_1$ provides approximate target densities that can be used to accelerate MCMC and SMC. Further details are presented in the next section.

\section{Integration with MCMC and SMC}
\label{sec:MCs}
In this section, we discuss the integration of MCMC and SMC with the approximate target densities defined by LIS for estimating the Gram matrix $H_1$. %
\subsection{MCMC with LIS}
\label{sec:MCMC}
For a given target density $\pi(x)$, the Metropolis--Hastings (MH) method employs a proposal density $p(x,x')$ and an acceptance/rejection step with the acceptance probability 
\[
\beta(x,x')=1\wedge \frac{\pi(x')p(x',x)}{\pi(x)p(x,x')}
\]
to construct a Markov chain of random variables with $\pi(x)$ as the invariant density. 
With the subspace identified by the LIS approach, we can apply different strategies to different subspaces to accelerate the convergence of MCMC. 
For a given subspace $\calX_r$, we can formulate an MCMC transition kernel on $\calX_r$ that has one of the lower-dimensional surrogate densities  $\bar\apdf_s(x_r), s \in\{f,g,l\},$ as the invariant density. Then, combining the transition kernel on $\calX_r$ and the conditional reference density $\mu(x_\bot|x_r)$, we can define a Markov chain transition kernel that has the full target density $\pi(x)$ as the invariant density. 
This procedure is summarized in Algorithm \ref{alg:MCMC}.

\begin{algorithm}[H]
 \KwIn{target density $\pi(x)$, an initial state $X^0 = x^0$, a LIS subspace $\calX_r$, proposal density $p$, conditional reference density $\mu(\,\cdot\,|\,\cdot\,)$, likelihood function $f$, and lower-dimensional surrogate density $\bar\apdf_s, s \in \{f,g,l\}$, iteration count $t$}
 \KwOut{a Markov chain $X^1,\ldots, X^t$}
 \For{$ j=1, \ldots, t$}{
 Given the previous state $X^{(j-1)} = x$, decompose it as $x = x_r + x_\bot$ based on the subspace decomposition $\reals^d=\calX_r\oplus \calX_\bot$\;
 Generate a MCMC proposal $x'_r\sim p(x_r, \,\cdot\,)$\;
 Let $x'_r=x_r$ with rejection probability $1-\beta(x_r,x_r')$,  $\beta(x_r,x_r')=1 \wedge \frac{\bar\apdf_s(x'_r)p(x'_r, x_r)}{\bar\apdf_s(x_r)p(x_r, x'_r)}$\; 
 Generate a proposal $x'_\bot\sim \mu(x_\bot| x_r)$ and set $x'=x'_r + x'_\bot$\;
 Compute the acceptance probability
 $\displaystyle \alpha(x,x') = 1 \wedge \frac{f(x') \bar\apdf_s(x_r) \bar\mu(x'_r)}{f(x)\bar\apdf_s(x'_r) \bar\mu(x_r)}$\;
 With probability $\alpha(x,x')$, accept the complement proposal and set $X^{k}=x'$, otherwise reject $x'$ and set $X^{j}=x$. 
 }
\caption{MCMC with LIS proposal}
\label{alg:MCMC}
\end{algorithm}

The acceptance and rejection steps used in lines 4 and 7 of Algorithm \ref{alg:MCMC} are consistent with the approximation of the target density. 
Since the lower-dimensional surrogate density $\bar\apdf_s(x_r)$ carries most of the information provided by the likelihood function, it may have a complicated structure to explore. However, the rather low dimensionality of  $\bar\apdf_s(x_r)$ makes it possible to design efficient MCMC transition kernels. 
Note that the product of the lower-dimensional surrogate density and the conditional reference density, $\bar\apdf_s(x_r)\mu(x_\bot| x_r)$, defines an approximation of the full-dimensional target density, in which the approximation accuracy has been extensively analyzed in previous sections. 
This way, in the complement space $\calX_\bot$, we embed $\mu(x_\bot| x_r)$, which is an approximation of $\pi(x_\bot|x_r)$, into another MCMC transition kernel to explore the full-dimensional target density $\pi(x)$. 
Thus, the efficiency of the complement subspace MCMC transition in lines 5--7 of Algorithm \ref{alg:MCMC} should strongly depend on its acceptance rate. 
In the following proposition, we show that $\pi(x)$ is indeed the invariant measure of Algorithm \ref{alg:MCMC}. In addition, we also provide a lower bound on the complement transition acceptance rate in line 6 of Algorithm \ref{alg:MCMC}.

\begin{prop}\label{prop:MCMC}
The full-dimensional target density $\pi(x)$ is the invariant density for Algorithm \ref{alg:MCMC}. Moreover, the expected acceptance rate for the $x_\bot$ part is lower bounded by 
\[
\E \left[\alpha(X,X')\right]\geq 1-4\sqrt{2}\Hell(\pi,\apdf_s). 
\]
Here $X$ is a random sample from $\pi$ and $X'$ is a proposal generated by Algorithm \ref{alg:MCMC}. 
\end{prop}
\begin{proof}
See\smartref{proof:MCMC}.
\end{proof}

Proposition \ref{prop:MCMC} indicates that when running Algorithm \ref{alg:MCMC}, the acceptance rate of the MCMC transition in the complement subspace $\calX_\bot$ is controlled by the accuracy of the approximate target density. One anticipates that the acceptance rate in line 6 approaches $1$ if the approximation error approaches $0$. In other words, the efficiency of Algorithm \ref{alg:MCMC} depends largely on the efficiency of the MCMC on the low dimensional $\calX_r$. %
To implement Algorithm \ref{alg:MCMC},  we need the lower-dimensional surrogate density $\bar\apdf_s, s \in \{f,g,l\}$. This in practice can be replaced by the Monte Carlo version $\bar\apdf^M_s$ (cf. Definitions \ref{def:pif}--\ref{def:pil}), with accuracy guaranteed by Theorems \ref{thm:MCerrorfg} and \ref{thm:MCerrorl}. 
Note that the surrogate density $\bar\apdf^M_f(x_r)$ provides an unbiased estimate of the marginal target density $\bar\pi(x_r)$. In  \cite{cui2020data}, this is used together with the pseudo-marginal technique \cite{andrieu2009pseudo,AV15} to design alternative sampling methods.

Another key ingredient in Algorithm \ref{alg:MCMC} is the LIS subspace $\calX_r$, which is obtained by estimating either the matrix $H_0$ or the matrix $H_1$. 
While $H_0$ is easy to compute, the resulting subspace may have inferior approximation accuracy compared with that obtained by $H_1$. 
However the estimation of $H_1$ often requires samples from the target $\pi(x)$ (cf. Section \ref{sec:MCsubspace}). To resolve this dilemma, we consider to adaptively estimate $H_1$ and the LIS $\calX_r$ within MCMC. 
The procedure is summarized in Algorithm \ref{alg:adaptMCMC}.

\begin{algorithm}[H]
 \KwIn{target density $\pi(x)$, reference density $\mu$, likelihood function $f$, number of epochs $K$, iteration count $t$, and a truncation index $K_\ast$}
 \KwOut{a Markov chain $X^1,\ldots, X^{(K+1)t}$ and a subspace $\calX_r$} 
 Generate $X^1,\ldots, X^{t}$ from $\mu$ and compute $H^{(0)}=\frac1{t} \sum_{i=1}^{t}\nabla \log f(X^{i}) \nabla \log f(X^{i})^\top$\;
 Find the leading eigenvectors $\{v_1,\ldots, v_{d_r}\}$ of $H^{(0)}$ to define $\calX_r=$span$\{v_1,\ldots, v_{d_r}\}$\;
 \For {$j=1,\ldots, K$}{
Run Algorithm \ref{alg:MCMC} with target density $\pi$, initial state $X^{jt}$, and LIS subspace $\calX_r$ for $t$ iterations to generate the Markov chain $X^{jt+1},\ldots X^{jt+t}$\;
Compute $H^{(j)} = \frac1{t} \sum_{i=1}^{t}\nabla \log f(X^{jt+i}) \nabla \log f(X^{jt+i})^\top$\;
Compute $\bar H = \frac{1}{j+1- {\rm min}(j,K_\ast)} \sum_{i = {\rm min}(j,K_\ast)}^{j} H^{(i)} $\;
Find the leading eigenvectors $\{v_1,\ldots, v_{d_r}\}$ of $\bar H$ to define $\calX_r=$span$\{v_1,\ldots, v_{d_r}\}$\;
 }
\caption{MCMC with adaptive LIS}
\label{alg:adaptMCMC}
\end{algorithm}

Our starting point is an initial LIS $\calX_0$ that can be estimated using $H_0$. Then, we run Algorithm \ref{alg:MCMC} using $\calX_0$ to generate samples from $\pi$. 
We call this the first epoch. Algorithm \ref{alg:MCMC} in this epoch might not be efficient, since $\calX_0$ may not be a good subspace. However, we can re-estimate the matrix $H_1$ and an improved LIS $\calX_1$ using the samples in the first epoch. 
Then, the updated $\calX_1$ is used in the next epoch to run Algorithm \ref{alg:MCMC}. This procedure can be carried out iteratively, where each epoch creates better estimates of the Gram matrix $H_1$ and the corresponding LIS. 
The truncation index $K_\ast$ is introduced to discard burn-in samples from initial epochs in estimating $H_1$.
Since the subspace estimation error follows a Monte Carlo convergence rate (cf. Theorem \ref{thm:sample}), we often only need to implement Algorithm \ref{alg:adaptMCMC} for $O(10^3)$ iterations in each epoch to estimate the LIS $\calX_r$ in many practical scenarios. In practice, the LIS usually stabilizes after a few epochs (e.g., about $10$ epochs) of training, so in total $O(10^4)$ iterations are needed to build the LIS.
%
%
Then, the estimated subspace $\calX_r$  can be used in the non-adaptive Algorithm \ref{alg:MCMC} to explore the target density.

\subsection{SMC with LIS}
\label{sec:SMC}

For target densities with complicated and multi-modal structures, SMC offers an efficient alternative to MCMC. 
Here we present the integration of SMC with LIS. This integration also offers a layered subspace construction procedure that is naturally embedded within SMC.
In our context, SMC uses a sequence of densities $\pi_k$, $k\in0,\ldots, K$, such that $\pi_0 = \mu$, $\pi_K = \pi$, and each ratio $\pi_{k+1}/\pi_{k}$ has a small variance. For example, one can obtain such a sequence using the tempering formula
\[
\pi_k(x)= \frac{1}{Z_k} \mu (x) f(x)^{\beta_k}, \quad Z_k = \int \mu (x) f(x)^{\beta_k} dx, \quad k\in0,1,\ldots, K,
\]
where $\beta_k \geq 0$ is an increasing sequence with $\beta_0=0$ and $\beta_K=1$.
This way, given samples from $\pi_k$, one can apply importance sampling to obtain weighted samples from $\pi_{k+1}$ and estimate associated statistics. Then, these statistics can be used to formulate MCMC transition kernels with the invariant density $\pi_{k+1}$ to update the weighted samples.
This procedure is summarized in Algorithm \ref{alg:SMC}.

\begin{algorithm}[H]
 \KwIn{likelihood function $f$, reference density $\mu$, tempering coefficients $\beta_k$ and iteration count $t_k$ for each level $k\in1,\ldots,K$}
 \KwOut{Samples $X^1,\ldots, X^T$ from $\pi\propto f \mu$}
  Generate samples $X^1,\ldots, X^{T}$ from $\pi_0 = \mu$\;
 \For{$ k\in0, \ldots, K-1$}{
Compute the weights $W^j = f^{\beta_{k+1}-\beta_k}(X^j)$ for $j=1,\ldots,T$\;
Let $H'_1= (\sum_{j=1}^{T} W^j)^{-1} \sum_{j=1}^{T}\nabla \log f(X^{j}) \nabla \log f(X^{j})^\top W^j $\;
Find the leading eigenvectors $\{v_1,\ldots, v_{d_r}\}$ of $H'_1$ to define $\calX_r=$span$\{v_1,\ldots, v_{d_r}\}$\;
\For {$t=1,\ldots, T$}{
Draw a resampling index $J$  from the categorical distribution with the probability mass function $\mathbb{P}[J = j] \propto W^j, j = 1,\ldots, T$, and set $Y^0 = X^J$\;
Run MCMC Algorithm \ref{alg:MCMC} with the invariant density $\pi_{k+1}\propto \mu f^{\beta_{k+1}}$,
initial state $Y^0$, LIS subspace $\calX_r$, iteration count $t_k$, let $Y^1,\ldots Y^{t_k}$ be the output \;
Update $X^t=Y^{t_k}$\;}
 }
\caption{SMC for LIS proposal}
\label{alg:SMC}
\end{algorithm}

In Algorithm \ref{alg:SMC}, a resampling step is used to transform a weighted particle representation of $\pi_{k+1}$ into an equally weighted particle representation, followed by MCMC updates. 
In practice, the tempering coefficients can be chosen adaptively. Given the weighting function $f^{\beta_{k+1} - \beta_k}(x)$, one can choose the next tempering coefficient $\beta_{k+1}$ such that
\[
\frac{\mathbb{E}_{\pi_k}\left[f^{\beta_{k+1} - \beta_k}(X) \right]^2}{\mathbb{E}_{\pi_k}\left[f^{\beta_{k+1} - \beta_k}(X)^2\right]} \approx  \frac{ESS}{n} < \tau,\quad \text{where }\,ESS=\frac{\left(\sum_{i = 1}^{n}f^{\beta_{k+1} - \beta_k}(X^i)\right)^2}{\sum_{i = 1}^{n} f^{\beta_{k+1} - \beta_k}(X^i)^2}
\]
is the effective sample size and $\tau \in (0, 1)$ is a predetermined threshold. 

For each tempering coefficient $\beta_{k+1}$, we construct the matrix $H_1$ for the corresponding target density $\pi_{k+1}$ using importance sampling with $\pi_k$ as the importance density, and then build MCMC transition kernels as described in Algorithm \ref{alg:MCMC}. Denoting the matrix $H_1$ for the density $\pi_{k+1}$ by $H_1^{(k+1)}$, its importance sampling estimate takes the form $\E_{\pi_k} [H^{(k,k+1)}(X)]$, 
where 
\begin{align*}
 H^{(k,k+1)}(x)= \frac{Z_k}{Z_{k+1}} \nabla \log f(x) \nabla \log f(x)^\top f(x)^{(\beta_{k+1}-\beta_k)}.
\end{align*}
The introduction of the tempering sequence reduces the variance of each $H_1$ estimation compared to the direct estimation of $H_1$ from the reference $\mu$. 
This can be characterized in the following proposition. 
\begin{prop}
\label{prop:SMC}
Suppose $\|\nabla \log f(x)\|^4f(x)^{\beta_{k+1}+\delta}$ with $\delta=\beta_{k+1}-\beta_k$ is integrable under the reference $\mu$. Then the variance upper bound
\[
V_{k+1}(H_1, \pi_{k}) := \sum_{i,j=1}^d {\rm var}_{X\sim\pi_k} \left[\big[H^{(k,k+1)}(X)\big]_{ij}\right] \leq \frac{Z_k}{Z_{k+1}}\E_{\pi_{k+1}} \left[ \|\nabla \log f(X)\|^4 f(X)^\delta \right]
\]%
for the subspace estimation in SMC is finite.
\end{prop}
\begin{proof}
See\smartref{proof:SMC}.
\end{proof}

By setting $\nu=\mu$ in the inequality for $V(H_1, \nu)$ in Proposition \ref{prop:Vbound}, and by using
that 
%
%
$\pi(x)/\mu(x)=\tfrac{1}{Z}f(x)$ for a.e. $x$, it follows that  $V(H_1,\mu)$ is finite if
$\|\nabla \log f(x)\|^4f(x)^{2}$ is integrable with respect to $\mu$. This is strictly stronger than the requirement in Proposition \ref{prop:SMC} if and only if $\beta_{k+1}+\delta<2$. The latter condition is valid as soon as $K\geq 2$ tempering coefficients are used.
Consider the same example used in Section \ref{sec:MCsubspace} where the likelihood function is bounded as $\sup_x f = 1$ and the gradient of the log-likelihood is bounded as $\sup_x \|\nabla \log f(x)\| =  M_f$. The variance $V_{k+1}(H_1, \pi_{k})$ in SMC satisfies
\[
V_{k+1}(H_1, \pi_{k})\leq \frac{M_f^4 Z_k}{Z_{k+1}},
\]
which can be much smaller than the upper bound $V(H_1, \mu)\leq M_f^4 \big/ Z^2$ of the direct importance sampling formula in Proposition \ref{prop:Vbound}. On the other hand, one needs to implement the SMC scheme which is in general more involved and computationally more expensive than using importance sampling.

\section{Dimension independent errors for linear inverse problems}
\label{sec:dimension}
LIS mainly targets high-dimensional problems with intrinsic low-dimensional structures. For infinite-dimensional problems, it is highly desired that the subspace approximation error (e.g., the result in Corollary \ref{cor:sample}) is independent of the ambient parameter dimension $d$. 
While the dimension independence of sampling methods has been extensively investigated in the literature, see \cite{cotter2013mcmc} and references therein, there has been little investigation on the dimension independence of LIS. 
Intuitively, for the approximation error to be dimension independent,  $H_k, k \in \{0,1\}$ must be trace-class in the limit as $d \rightarrow \infty$ and  the variance $V(H_k, \nu)$ must be bounded independently of $d$. 
It is an open question to establish conditions under which these properties hold for general likelihood functions. We will show that the two conditions above are satisfied for linear Gaussian Bayesian inverse problems, i.e. where the prior and likelihood are Gaussian and the parameter-to-observable map is linear.

We consider a Bayesian problem with unknown parameter $z \in \mathbb{R}^d$ and the prior $p_0(z)$ being $\mathcal{N}(0,\Gamma)$. Given a linear parameter-to-observable map $G$, the data are given by
\[
Y = G Z+\xi,\quad \xi\sim \mathcal{N}(0, I_{d_y}). 
\]
Applying a whitening transformation $X=\Gamma^{-1/2}Z$, we obtain $X\sim\mu(x)$ where $\mu(x) = \mathcal{N}(0,I_d)$ and 
\[
Y=AX+\xi,\quad A=G\,\Gamma^{1/2}. 
\] 
This defines the likelihood function 
\(
f(x) = \exp(-\tfrac12 \|Ax - Y\|^2).
\)
First we will establish a series of estimates for the quantities we derived in previous sections. 
\begin{prop}
\label{prop:linear}
Denote $C_A=A^\top A=G^\top\Gamma G$. The following hold:
\begin{enumerate}[1)]
\item The eigenvalues of $H_0$ are controlled by 
\(\displaystyle
\lambda_{i+1} (H_0)\leq\lambda_{i}(C_A)^2. \vphantom{\tfrac{f}{Z}}
\)
\item The eigenvalues of $H_1$ are controlled by 
\(\displaystyle
\lambda_{i+1} (H_1)\leq \lambda_i(C_A)^2 \big/ 1+\lambda_i(C_A).\vphantom{\tfrac{f}{Z}}
\)
\item The normalizing constant is bounded by 
\(\displaystyle
\sqrt{ \det(C_A+I)} \leq \tfrac{1}{Z}\leq \sqrt{\det(C_A+I)}\exp\big(\tfrac12 \|y\|^2 \big).
\)

\item The constant $\tfrac{\|f\|_{2,\mu}}{Z}$ is bounded by 
\(
\frac{\|f\|_{2,\mu}}{Z} \leq  \det(I + C_A^2)^{1/4} \exp\big( \tfrac12 (\sqrt{2}-1)^2 \|y\|^2 \big) \vphantom{\tfrac{f}{Z}}
\)

\item $\vphantom{\tfrac{f}{Z}}$When the reference density $\mu$ is used for estimating $H_k, k \in \{0,1\},$ the variances are bounded by 
\begin{align*}
V(H_0,\mu) & \leq  6 \Bigg( \left(\sum_{i = 1}^d \lambda_i(C_A)^2 \right)^2 + \|A^\top y\|^4 \Bigg), \\
V(H_1,\mu) & \leq  6 \sqrt{\det(I + C_A^2)}  \exp\left( (\sqrt{2}-1)^2 \|y\|^2 \right) \Bigg( \left( \sum_{i = 1}^d \frac{\lambda_i(C_A)^2}{1 + 2\lambda_i(C_A)}\right)^2 + \| A^\top y\|^4 \Bigg). 
\end{align*}
\item The SMC sampling variance in Proposition \ref{prop:SMC} is bounded by 
\[
V_{k+1}(H_1,\pi_k) \leq 6 \sqrt{\det(I + \delta^2 C_A^2)} \exp\left(\delta^2 \|A^\top y\|^2\right) \Bigg( \left(\sum_{i = 1}^d \frac{\lambda_i(C_A)^2}{1 + \tau \lambda_i(C_A)}\right)^2 + \| A^\top y\|^4 \Bigg).
\]
where $\delta = \beta_{k+1}-\beta_k$ and $\tau = \beta_{k+1} + \delta$.
\end{enumerate}
\end{prop}
\begin{proof}
See\smartref {proof:linear}.
\end{proof}
Claim 1) implies that the spectrum of $H_0$ is bounded by the spectrum of $C_A^2$.
Claim 2) implies that the spectrum of $H_1$ is bounded by the spectrum of $C_A^2(I+C_A)^{-1}$.
Note that 
\[
\frac{\lambda_i(C_A)^2}{1+\lambda_i(C_A)}\leq \lambda_i(C_A)^2,
\]
and the ratio between the two is large if $\lambda_i(C_A)$ is large. Also recall that in Remark \ref{rem:25}, we showed the approximation errors with subspace obtained from $H_0$ involve the pre-constant $1/Z$, which is estimated in claim 3), but not the subspace obtained from $H_1$. Thus, in the Gaussian linear setting, the error estimates obtained from $H_1$ will be much tighter than those obtained from $H_0$ when the dominating eigenvalues of $C_A$ are large, which is often the case in practice.

On the other hand, $H_0$ is easier for  Monte Carlo based estimations than $H_1$. This can be seen from the bounds on the  variance of the Gram matrices in claim 5). Comparing with  $V(H_0,\mu)$, $V(H_1,\mu)$ has an additional dependence on $\det(I + C_A^2)$ and $\exp\big( \tfrac12 (\sqrt{2}-1)^2 \|y\|^2 \big)$. 
Claim 6) shows that this estimation difficulty can be remedied by SMC, because the upper bound of the variance $V_{k+1}(H_1,\pi_k)$ is smaller than that of $V(H_1,\mu)$.

In many applications, the spectrum of the prior covariance $\Gamma$ is assumed to exhibit polynomial decay, i.e. $\lambda_j(\Gamma)\leq C_\Gamma j^{-\alpha}, \alpha>0$. This kind of assumption is common for functional data analysis \cite{CH08,HH07,RS05} and inverse problems \cite{St10}.
With $\alpha > 1/2$, the prior covariance is trace-class, and thus the prior has measure 1 on some suitably constructed Banach space. 
In the following corollary, we replace the bounds in Proposition \ref{prop:linear} with estimates obtained using this trace-class constraint to demonstrate the dimension scalability.

\begin{cor}
\label{cor:lin}
Suppose the eigenvalues of $\Gamma$ exhibit polynomial decay, $\lambda_j(\Gamma)\leq C_\Gamma j^{-\alpha}$ with $\alpha > 1/2$, the observation matrix $G$  has bounded $\ell^2$ operator norm, and the observed data have bounded $\ell^2$ norm. Then the following estimates hold independently of the ambient parameter dimension $d$. Consequently, the estimation error $\Error_h$ for any bounded $h$ is independent of $d$ by Corollary \ref{cor:sample}. 
\begin{enumerate}[1)]
\item When using $\calX_r$ as the subspace spanned by the first $d_r$ eigenvectors of $H_k$, $k\in\{0,1\}$,
\[
\res(\calX_r, H_k)\leq \frac{1}{2\alpha-1}\|G\|^4 \, C_\Gamma^2 (d_r-1) ^{1 - 2\alpha}.
\]
\item The constant $\tfrac{\|f\|_{2,\mu}}{Z}$ is bounded by
\(
\frac{\|f\|_{2,\mu}}{Z} \leq  \exp\big( \frac12 (\sqrt{2}-1)^2 \|y\|^2  + \frac{\alpha}{4\alpha -2} \|G\|^4 \, C_\Gamma^2 \big).
\)
\item When the reference distribution $\mu$ is used for estimating the matrices $H_k, k \in \{0,1\},$ the variances are bounded by 
\begin{align*}
V(H_0,\mu) & \leq 6 \left( \frac{4\alpha^2}{(2\alpha -1)^2} \|G\|^8 \, C_\Gamma^4 + \|A^\top y\|^4 \right), \\
V(H_1,\mu) & \leq 6 \exp\left( (\sqrt{2}-1)^2 \|y\|^2 + \frac{\alpha}{ 2 \alpha - 1} \|G\|^4 C_\Gamma^2  \right) \left( \frac{4\alpha^2}{(2\alpha -1)^2} \|G\|^8 \, C_\Gamma^4 + \|A^\top y\|^4 \right).
\end{align*}
\item The SMC sampling variance in Proposition \ref{prop:SMC} is bounded by 
\[
V_{k+1}(H_1,\pi_k) \leq 6 \exp\left(\delta^2 \|A^\top y\|^2 + \frac{\delta^2 \alpha}{ 2 \alpha - 1} \|G\|^4 C_\Gamma^2  \right) \left( \frac{4\alpha^2}{(2\alpha -1)^2} \|G\|^8 \, C_\Gamma^4 + \| A^\top y\|^4 \right).
\]
where $\delta = \beta_{k+1}-\beta_k$. 
\end{enumerate}
\end{cor}

\begin{proof}
See\smartref {proof:lin}.
\end{proof}

\section{Numerical examples}\label{sec:numerics}

Now we provide several numerical examples to illustrate the theoretical results developed in the preceding sections. We start with a synthetic linear inverse problem to demonstrate various likelihood approximation methods and continue with a more practical nonlinear Bayesian inference problem governed by a partial differential equation (PDE). 

\subsection{Example 1: synthetic example}
In the first example, we consider a Bayesian inverse problem with linear observations and log-normal prior. Problems of this type are applied in X-ray tomography and atmospheric remote sensing, see \cite{haario2004markov} and references therein. The parameter to be inferred can be modeled by a random Gaussian vector $X \sim \mu=\mathcal{N}(0, \Gamma)$, where $\Gamma \in \mathbb{R}^{d \times d}$ is the prior covariance matrix. The observation data are modeled through
\[
Y = G \exp(X) + \xi, \quad \xi \sim \mathcal{N}(0, \sigma^2 I_{d_y}),
\]
where $G \in \mathbb{R}^{d_y \times d}$ is a matrix. 
The observation likelihood is then given by 
\[
f(x;y) \propto \exp\left( -\frac1{2 \sigma^2} \left\| y - G \exp(x) \right\|^2 \right).
\]
To compare different approximations by exploring regimes where the data have differing impacts on different parameter directions, we generate random observation matrices and prescribe the spectra of the observation matrix and the prior covariance matrix. 
We specify the prior covariance by setting $\Gamma = {\rm diag}(\gamma_1, \gamma_2, \ldots, \gamma_d)$ with $\gamma_j = \gamma_0\,j^{-\beta_\gamma}$.
To create a random observation matrix $G$, we use the reduced singular value decomposition $G = U \Lambda V^\top$, where the matrices $U$ and $V$ are randomly and independently generated from the orthogonal group \cite{stewart1980efficient} and $\Lambda = {\rm diag}(\lambda_1, \lambda_2, \ldots, \lambda_{d_y})$ with $\lambda_j = \lambda_0 \, j^{-\beta_\lambda}$. 
In particular, $U$ and $V$ are computed using a QR decomposition of a matrix of independent standard Gaussian entries.
By using randomly generated $G$, we can confirm the observed phenomena are not restricted to a specific choice of $G$. In this example, we present 3 independently generated $G$, and we will see the numerical results have little differences among them. 
The problem dimensions are set to $d = 500$ and $d_y = 50$. The variables that determine $\Gamma$ and $G$ are given by $\gamma_0 = 4$,  $\beta_\gamma = -2$, $\lambda_0 = 100$, and $\beta_\lambda = -1$. The standard deviation of the observation noise is given by $\sigma = 1$. 

\begin{figure}[htb]
\centering
\includegraphics[width=0.3\linewidth]{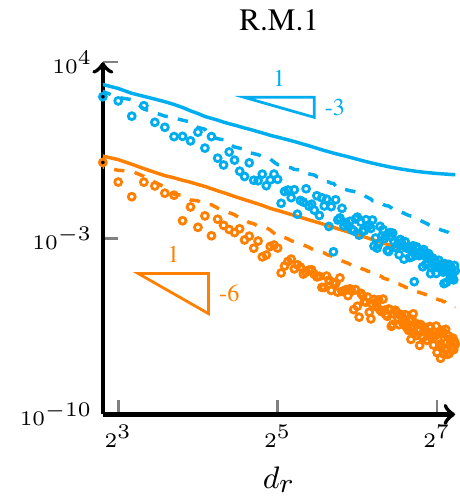}
\includegraphics[width=0.3\linewidth]{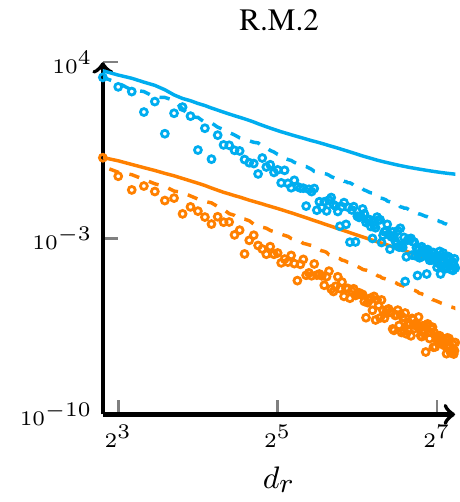}
\includegraphics[width=0.3\linewidth]{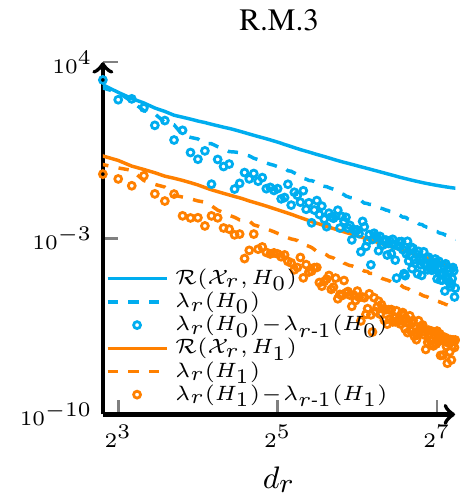}
\caption{Synthetic example. Eigenvalues and the gaps of eigenvalues of $H_k$ matrices for three randomly generated observation maps (R.M.) and the sums of the residual eigenvalues versus the projection dimensions.}
\label{fig:ex1_spec}
\end{figure}

\begin{figure}[htb]
\centering
\includegraphics[width=0.3\linewidth]{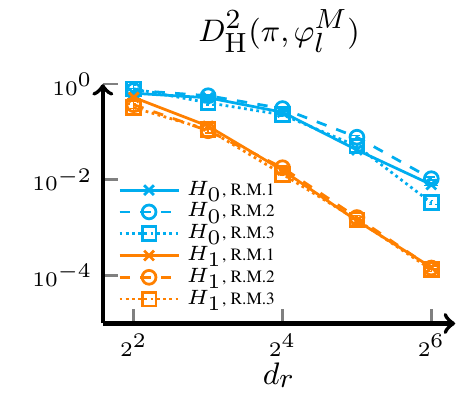}
\includegraphics[width=0.3\linewidth]{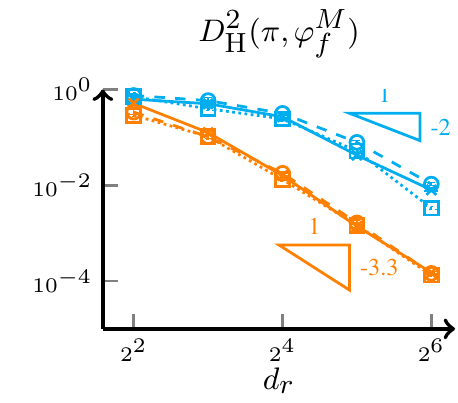}
\includegraphics[width=0.3\linewidth]{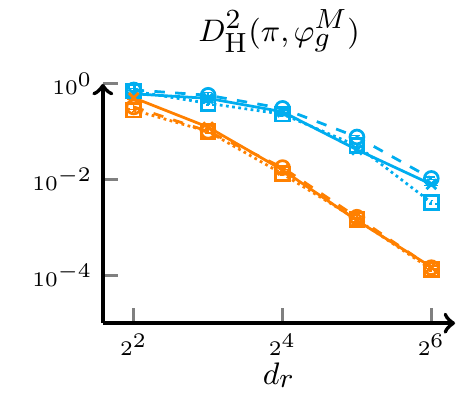}
\caption{Synthetic example.  Approximation errors of the approximate posterior densities versus projection dimensions for different $H_k$ matrices and different approximation methods. From left to right, approximation methods are $\apdf_l^M$, $\apdf_f^M$, and $\apdf_g^M$, respectively. Sample size $M = 4$ is used in computing the conditional expectation.}
\label{fig:ex1_errors}
\end{figure}

\begin{figure}[htb]
\centering
\includegraphics[width=0.3\linewidth]{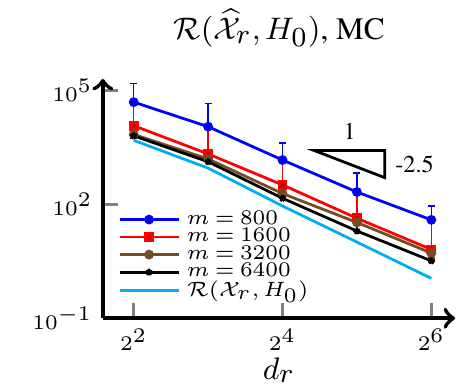}
\includegraphics[width=0.3\linewidth]{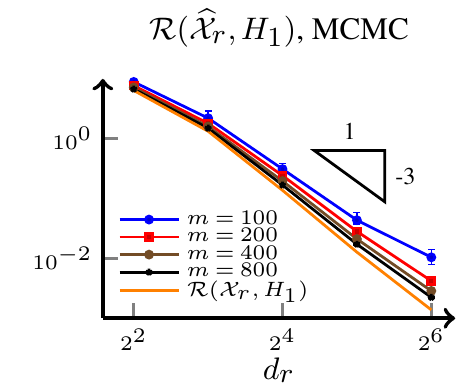}
\includegraphics[width=0.3\linewidth]{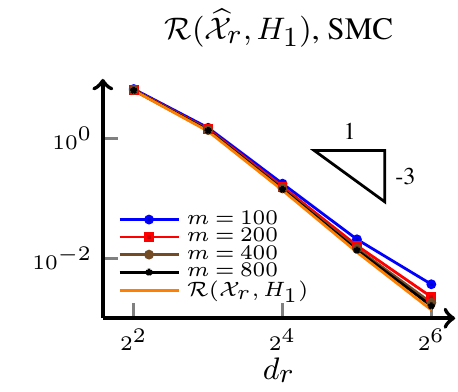}
\caption{Synthetic example with the first random observation matrix. The sums of the residual eigenvalues versus projection dimensions for different $H_k$ matrices and subspaces $\widehat{\mathcal{X}}_r$ computed using different samples sizes $\Mmu$ and different methods. From left to right, we have Monte Carlo estimation of $\widehat{H}_0$, MCMC estimation of $\widehat{H}_1$, and SMC estimation of $\widehat{H}_1$, respectively. The error bars represent $[10\%, 90\%]$ quantiles estimated using $100$ runs.}
\label{fig:ex1_H_data1}
\end{figure}
\begin{figure}[htb]
\centering
\includegraphics[width=0.3\linewidth]{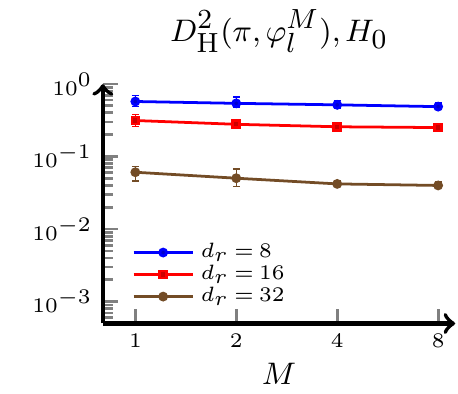}
\includegraphics[width=0.3\linewidth]{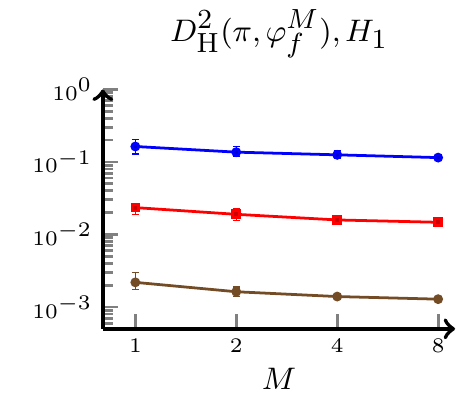}
\includegraphics[width=0.3\linewidth]{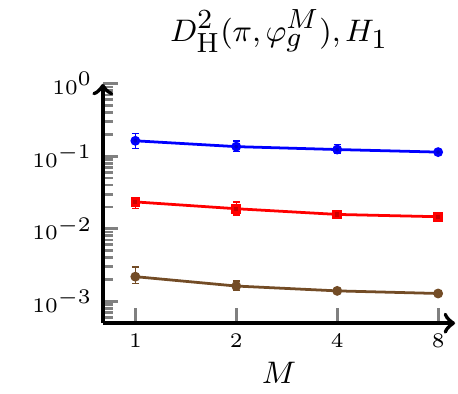}
\caption{Synthetic example with the first random observation matrix.  Approximation errors of the approximate posterior densities (with projection dimensions $d_r = \{8, 16, 32\}$) versus sample sizes $M = \{1, 2, 4, 8\}$. From left to right, approximation methods are $\apdf_l^M$ with the matrix $H_0$, $\apdf_f^M$ with the matrix $H_1$, and $\apdf_g^M$ with the matrix $H_1$, respectively. The error bars represent $[10\%, 90\%]$ quantiles estimated using $100$ runs.}
\label{fig:ex1_errors_data1}
\end{figure}

We present the numerical results based on three realizations of the randomly generated $G$ matrices. %
For each test case, we construct the ``true'' matrix $H_0$ using $5 \times 10^5$ Monte Carlo samples and construct the ``true'' matrix $H_1$ using $5 \times 10^6$ MCMC samples. 
As shown in Figure \ref{fig:ex1_spec}, for all three test cases, the eigenvalues of $H_1$ matrices are several orders smaller than those of $H_0$ matrices, and thus the sums of the residual eigenvalues, $\mathcal{R}(\mathcal{X}_r, H_1)$, are significantly lower than $\mathcal{R}(\mathcal{X}_r, H_0)$. 
This suggests that the approximate posterior density induced by the $H_1$ matrix should have better accuracy compared with that defined by the $H_0$ matrix in this example. This is confirmed in Figure \ref{fig:ex1_errors}, which shows the squared Hellinger distances between the true posterior and three approximations $\apdf_s^M, s \in \{f, g, l\}$ introduced in Section \ref{sec:MCerror}.
For all posterior approximation methods, the approximation subspaces defined by $H_1$ matrices yield significantly smaller squared Hellinger distances than those of $H_0$ matrices.

The eigenvalue gaps, $\lambda_r-\lambda_{r-1}$, are also plotted in Figure \ref{fig:ex1_spec}. Note that we use the same x-axis by assuming $r=d_r$. The eigenvalue gaps decay to zero quickly. In particular, for $H_1$,  the gap is around $10^{-3}$ for a moderate reduction dimension $r=32$, and $10^{-7}$ for a larger dimension $r=128$. For $H_0$, the effective eigenvalue gaps are of similar values, since we need to divide it by $\lambda_1$, which is around $10^4$ for $H_0$.  This illustrates that it is important for the theoretical results to be independent of eigenvalue gaps, as explained in Remark \ref{rem:subspace}. 

Since all the numerical results are similar among all three randomly generated observation matrices, we will focus on the first realization for subsequent discussion. 
Next, we investigate the impact of sample size $\Mmu$ for estimating the lower-dimensional subspace $\widehat{\mathcal{X}}_r$. 
Figure \ref{fig:ex1_H_data1} presents the quantiles of the sums of the residual errors for various sample-based subspace estimations. Although both $H_0$ and $H_1$ matrices have diminishing eigenvalue gaps in this example (c.f. Figure \ref{fig:ex1_spec}), the sample-based estimations still exhibit sufficient accuracy in probing the dimension reduced subspaces. For example, with a rather small sample size $\Mmu = 100$, using either MCMC or SMC can lead to accurate subspace estimations for the approach based on the $H_1$ matrix. This confirms the findings of our analysis in Section \ref{sec:MCsubspace}. In addition, the SMC-based estimation has better accuracy compared to that of MCMC-based estimation, which is also anticipated by our analysis in Section \ref{sec:MCs}.

Finally, we investigate the Monte Carlo sample size $M$ for computing the conditional expectations in various approximate posteriors, as discussed in  Section \ref{sec:MCerror}.
The results presented in Figure \ref{fig:ex1_errors_data1} confirm the results of our analysis. For problems with rather small residual eigenvalues $\mathcal{R}(\mathcal{X}_r, H_k)$, a small sample size $M$ is sufficient for accurate estimation of the conditional expectations. Increasing $M$ only leads to a marginal improvement in the approximation accuracy in this example. The reasons behind this were explained at the end of Section \ref{sec:MCerror}.

\subsection{Example 2: PDE problem}
We consider a classical Bayesian inverse problem governed by an elliptic PDE \cite{dashti2011uncertainty,dodwell2019multilevel}. Such problems arise in subsurface flows and oil reservoir management. %
Fix a domain of interest $D$ with boundary $\partial D$. The potential function $t \mapsto u(t)$ where $t \in D \subset \mathbb{R}^2$ is modeled by the PDE
\begin{equation}
-\nabla \cdot \big(\kappa(t) \nabla u(t)\big) = 0, \quad t \in D:= (0,1)^2,
\label{eq:pdeproblem}
\end{equation}
with Dirichlet boundary conditions $u|_{t_1=0}=1$ and $u|_{t_1=1}=0$ on the left and right boundaries, and homogeneous Neumann conditions on other boundaries.
The diffusion coefficient $\kappa(t)$ should be positive, and thus it is often parametrized by its logarithm, i.e., $\kappa(t) = \exp(x(t))$. The goal is to infer the unknown parameter function $x(t)$ from $d_y$ incomplete observations of the potential function $u(t)$. 
Following the setup of \cite{Cuietal16b}, a zero-mean Gaussian process prior with the exponential kernel
\[
K(t, t') = \exp\left( -\frac1{\ell} \|t-t'\| \right)
\]
is prescribed to the unknown parameter $x(t)$.

Given an arbitrary function $x(t)$, the PDE \eqref{eq:pdeproblem} cannot be solved analytically. This way, the functions $x(t)$ and $u(t)$ need to be discretized to numerically solve \eqref{eq:pdeproblem}. We tessellate the spatial domain $D$ with a uniform triangular grid with mesh size $h$, and then define continuous, piecewise quadratic finite element (FE) basis functions $\{\phi_1(t), \ldots, \phi_d(t)\}$ with cardinality $d$. Then, the infinite dimensional functions $x(t)$ and $u(t)$ can be approximated by $x(t) \approx x_h(t) := \sum_{i = 1}^{d} \phi_i(t) x_i$ and $u(t) \approx u_h(t) := \sum_{i = 1}^{d} \phi_i(t) u_i$, respectively. 
After discretization, the unknown function $x_h(t)$ can be effectively represented by a coefficient vector $x = (x_1, \ldots, x_d)$, which yields a multivariate Gaussian prior $\mu(x):=\mathcal{N}(0, \Gamma)$ where $\Gamma_{ij} = \int \int K(t, t') \phi_i(t) \phi_j(t') dt\, dt'$.

\begin{figure}[htb]
\centering
\includegraphics[width=0.8\linewidth]{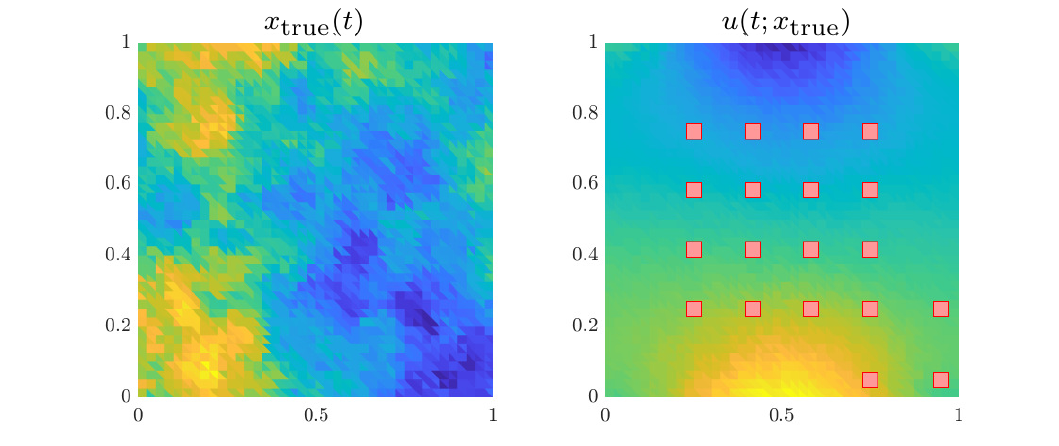}
\caption{Setup of the PDE inverse problem. Left: the true parameter $x_{\rm true}(t)$ used for generating synthetic data. Right: the corresponding potential function $u(t; x_{\rm true})$ and observation locations.}
\label{fig:ex2_setup}
\end{figure}

For any given parameter coefficients $x$, the corresponding discretized potential function $u_h(t; x)$ is obtained by solving the Galerkin projection of the PDE \eqref{eq:pdeproblem}. 
Observations are collected as $d_y = 19$ local averages of the potential function $u(t)$ over sub-domains $D_{i} \subset D$, $i=1,\ldots,d_y$. The subdomains are shown by the squares in Figure \ref{fig:ex2_setup}. To simulate the observable model outputs, we define the forward model $G: \mathbb{R}^d \mapsto \mathbb{R}^{d_y}$ with 
$$
G_i(x) = \frac{1}{|D_{i}|}\int_{D_i} u_h(t;x) dt, \quad i=1,\ldots,d_y\,.
$$
Synthetic data for these $d_y$ local averages are produced by $y = G(x_{\rm true}) + \xi$, where  $\xi \sim \mathcal{N}(0, \sigma^2 I_{d_y})$ and $x_{\rm true}$ is a realization of the prior random variable. 
To investigate the impact of the observation noise in practical applications, we present three test cases with observational standard deviations $\sigma = \{0.034, 0.017, 0.0085\}$ that correspond to signal-to-noise ratios, 10, 20, and 40, respectively. The resulting posterior distribution concentrates with reducing $\sigma$.

\begin{figure}[htb]
\centering
\includegraphics[width=0.3\linewidth]{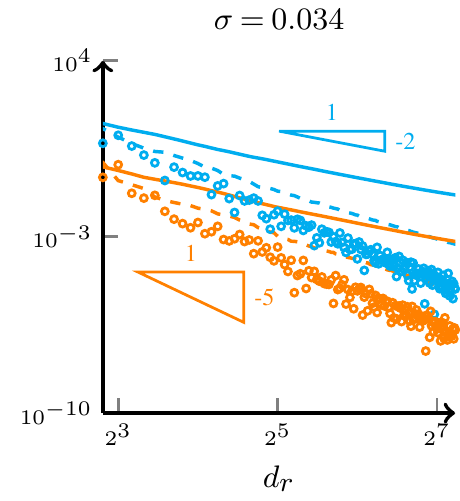}
\includegraphics[width=0.3\linewidth]{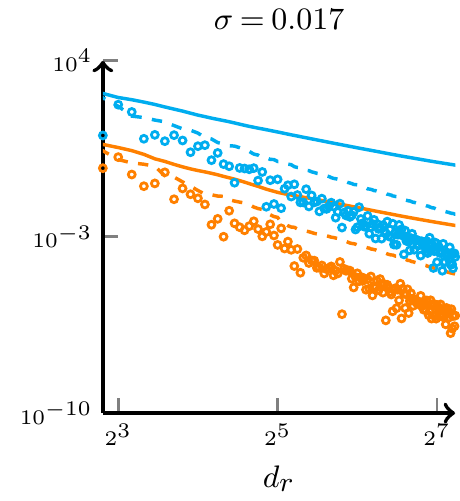}
\includegraphics[width=0.3\linewidth]{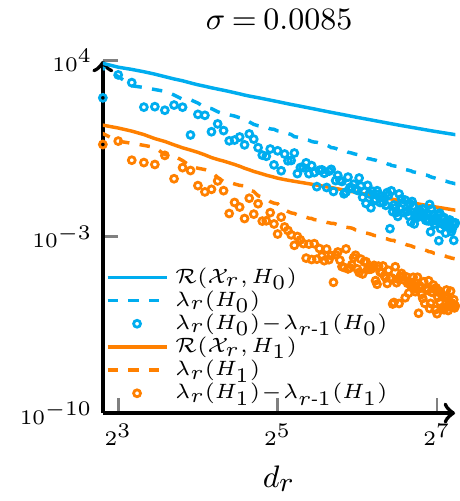}
\caption{PDE example. Eigenvalues and their gaps of $H_k$ matrices for three test cases with $\sigma = \{0.034, 0.017, 0.0085\}$ and the sums of the residual eigenvalues versus projector dimensions.}
\label{fig:ex2_spec_obs}
\end{figure}

\begin{figure}[htb]
\centering
\includegraphics[width=0.3\linewidth]{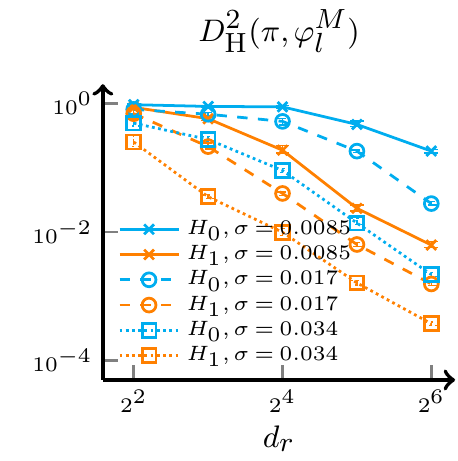}
\includegraphics[width=0.3\linewidth]{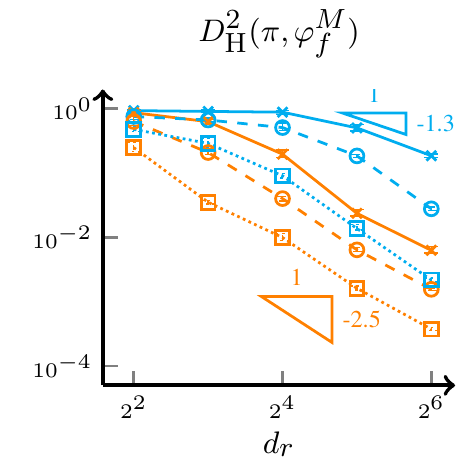}
\includegraphics[width=0.3\linewidth]{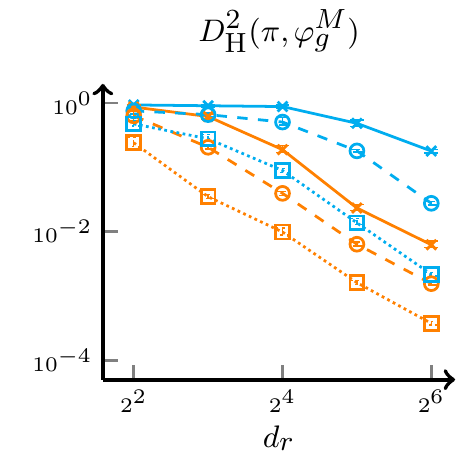}
\caption{PDE example. Approximation errors of the approximate posterior densities versus projection dimensions for different $H_k$ matrices and different approximation methods. From left to right, approximation methods are $\apdf_l^M$, $\apdf_f^M$, and $\apdf_g^M$, respectively. Sample size $M = 4$ is used in computing the conditional expectation.}
\label{fig:ex2_errors_obs}
\end{figure}

\begin{figure}[htb]
\centering
\includegraphics[width=0.3\linewidth]{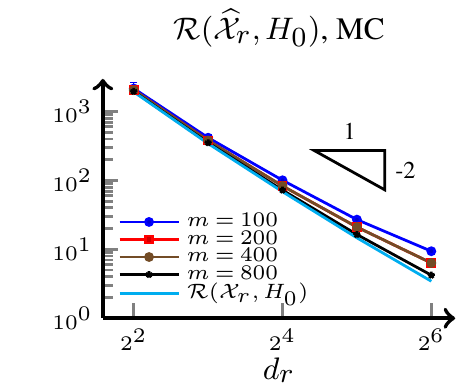}
\includegraphics[width=0.3\linewidth]{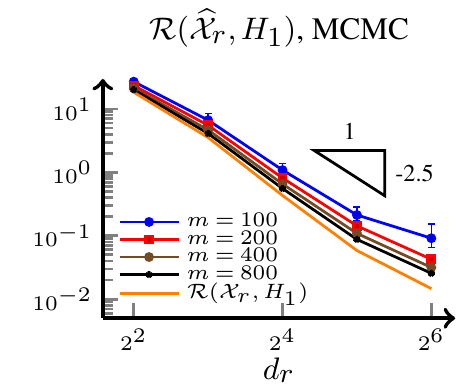}
\includegraphics[width=0.3\linewidth]{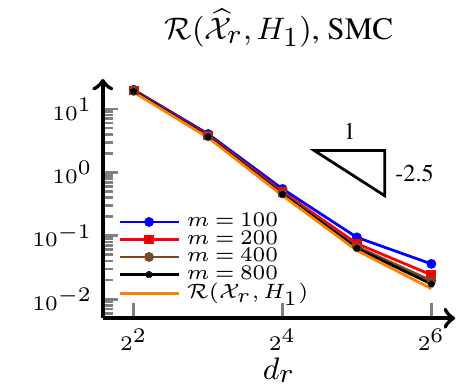}
\caption{PDE example with $\sigma = 0.017$. The sums of the residual eigenvalues versus projection dimensions for different $H_k$ matrices and subspaces $\widehat{\mathcal{X}}_r$ computed using different samples sizes $m$ and different methods. From left to right, we have Monte Carlo estimation of $\widehat{H}_0$, MCMC estimation of $\widehat{H}_1$, and SMC estimation of $\widehat{H}_1$, respectively. The error bars represent $[10\%, 90\%]$ quantiles estimated using $100$ runs.}
\label{fig:ex2_H_obs2}
\end{figure}

As shown in Figure \ref{fig:ex2_spec_obs}, for all three test cases, the eigenvalues of the $H_k$ matrices and their gaps decay rapidly, which are similar to the first numerical example. 
The eigenvalues of $H_1$ matrices and the associated sums of the residual eigenvalues, $\mathcal{R}(\mathcal{X}_r, H_1)$, are several orders smaller than the eigenvalues of $H_0$ matrices and $\mathcal{R}(\mathcal{X}_r, H_0)$. In addition, the gap between $\mathcal{R}(\mathcal{X}_r, H_1)$ and $\mathcal{R}(\mathcal{X}_r, H_0)$ increases with decreasing $\sigma$. 
This suggests that the accuracy improvement of the approximate posterior densities induced by the $H_1$ matrix over those of the $H_0$ matrix can be further enhanced for more concentrated posterior distributions in this example. Moreover, with a smaller $\sigma$, the signal-to-noise ratio is larger, and the posterior density is more different from the prior. As a result, the sampling problem becomes more challenging. This can be observed from the eigenvalue values for $\sigma=8.5\times 10^{-3}$, which are several magnitudes larger than the ones of  $\sigma=3.4\times 10^{-2}$.

In Figure \ref{fig:ex2_errors_obs}, we compare the performance of using $H_0$ and $H_1$ with three projection methods and noise scales. The performance is measured in Hellinger distance from the true posterior. We can see that all three projection methods yield similar results. Approximations using $H_1$ consistently outperform  the ones using $H_0$, especially when $d_r$ increases. By reducing $\sigma$, the sampling problem becomes harder, so the approximation becomes less accurate. But this is more severe for the approximations with $H_0$, since the approximation error is of order $1$ when  $d_r=2^6$ while the approximation error using  $H_1$ is of $10^{-2}$ when $d_r=2^6$.

Using the test case with $\sigma = 0.017$, we investigate the impact of sample size $\Mmu$ for estimating the lower-dimensional subspace $\widehat{\mathcal{X}}_r$.
Figure \ref{fig:ex2_H_obs2} presents the quantiles of the sums of the residual errors for various sample-based subspace estimations. Similar to the first numerical example, the diminishing eigenvalue gaps of $H_k$ matrices (c.f. Figure \ref{fig:ex2_spec_obs}) do not impact the accuracy of subspace estimations. 
With a rather small sample size $\Mmu = 100$, all methods (MC for $H_0$ and MCMC and SMC for $H_1$) can lead to accurate subspace estimations. %
\section{Conclusion}

This paper has provided a step-by-step analysis of the accuracy of the LIS method for approximating high-dimensional intractable target probability densities. %
We have shown that information about the spectrum of the Gram matrices $H_k, k \in \{0,1\},$ leads to upper bounds on the errors of various true approximations. We have also generalized these upper bounds to the numerical implementation of the approximate probability densities, in which Monte Carlo averaging is applied to both the estimation of $H_k, k \in \{0,1\},$ and the marginalization used during the construction of the approximate likelihood functions. 
Our analysis provides insights into the trade-off between the usage of $H_0$ and $H_1$ for constructing LIS: while the approximations based on $H_1$ can have smaller approximation errors compared with those obtained from $H_0$, the matrix $H_1$ is often more difficult to estimate. 
Fortunately, this difficulty can be addressed by integrating the LIS estimation process into sampling tools such as MCMC and SMC. We have also discussed the performance of the integration of MCMC and SMC with LIS. 
We have demonstrated our analysis on a linear Bayesian inverse problem, where all the error bounds presented in this paper are independent of the ambient parameter dimension, under suitable technical assumptions that are commonly used in high- or infinite-dimensional inverse problems. 
Finally, we have provided numerical examples to further demonstrate the efficacy of our analysis on nonlinear problems.

This work leads to some future research directions for dimension reduction techniques. Firstly, our analysis of the linear Bayesian inverse problem shows that various approximation errors are dimension independent. We conjecture this property will also hold for general nonlinear Bayesian inverse problems. Finding the conditions that guarantee this property remains an open problem. Secondly, our analysis indicates that the expected conditional variance of the square root of the likelihood controls the approximation error. This may lead to new dimension reduction techniques that bypass the usage of the Poincar\'{e} inequality and the gradient. 
Moreover, the analysis presented in this work can be further generalized to other types of log-concave reference distributions, for example, the Laplace distribution that is commonly used in sparsity-promoting learning. This may require further investigations on using weighted Poincar\'{e}-type inequalities \cite{bobkov1997poincare,bobkov2009weighted} for building alternative $H_{k}$ matrices and subspace approximations.

\appendix

\newcommand{\calE}{\mathcal{E}}

\newcommand{\R}{\mathbb{R}}
\newcommand{\gap}{\text{Gap}}

\section{Useful lemmas}
We begin with several useful lemmas for our discussion. Although some of them are not new, we provide proofs for all lemmas for the sake of completeness.
\begin{lem}
\label{lem:hel2exp}
The following holds
\begin{enumerate}[1)]
\item The estimation error of a $L^2$ function $h$ can be bounded by Hellinger distance
\[
|\E_{\pi} [h]-\E_\nu [h]|\leq \sqrt{2 \E_{\pi} [h^2]+2\E_{\nu}[h^2]}\Hell(\pi,\nu).
\]
This result can also be found in \cite[Proposition 5.12]{sullivan2015introduction}.
\item The Hellinger distance can be bounded by the square root of KL divergence
\[
\Hell(\pi,\nu)\leq \sqrt{\frac12 D_{KL}(\pi,\nu)}. 
\]
This is often referred as the Csisz{\'a}r-Kullback-Pinsker inequality \cite{bolley2005weighted, tsybakov2008introduction}. 
\item The total variation distance can be bounded by the Hellinger distance 
\[
D_{TV}(\pi,\nu)\leq \sqrt{2}\Hell(\pi,\nu). 
\]
This is often referred as the Kraft's inequality \cite{steerneman1983total}. 
\end{enumerate}
\end{lem}
\begin{proof}
\textbf{Proof of claim 1).}  Let $\lambda$ be a reference density, e.g. the Lebesgue density,  for the Hellinger distance, so
\[
\Hell(\pi,\nu)^2=\frac12\int \left(\sqrt{\frac{\pi(x)}{\lambda(x)}}-\sqrt{\frac{\nu(x)}{\lambda(x)}}\right)^2 \lambda(x)dx.
\]
Note that
\begin{align*}
|\E_{\pi} [h]-\E_\nu [h]|^2&=\left(\int \left(\frac{\pi(x)}{\lambda(x)}-\frac{\nu(x)}{\lambda(x)}\right) h(x)\lambda(x)dx\right)^2\\
&=\left(\int \left(\sqrt{\frac{\pi(x)}{\lambda(x)}}-\sqrt{\frac{\nu(x)}{\lambda(x)}}\right)\left(\sqrt{\frac{\pi(x)}{\lambda(x)}}+\sqrt{\frac{\nu(x)}{\lambda(x)}}\right) h(x)\lambda(x)dx\right)^2\\
(\mbox{by Cauchy--Schwarz})&\leq\Hell(\pi,\nu)^2\left(\int \left(\sqrt{\frac{\pi(x)}{\lambda(x)}}+\sqrt{\frac{\nu(x)}{\lambda(x)}}\right)^2 h^2(x)\lambda(x)dx\right)\\
(\mbox{by Young's ineq.})&\leq\Hell(\pi,\nu)^2\left(\int 2\left(\frac{\pi(x)}{\lambda(x)}+\frac{\nu(x)}{\lambda(x)}\right)h^2(x)\lambda(x)dx\right)\\
&=2 (\E_{\pi} [h^2]+\E_{\nu}[h^2])\Hell(\pi,\nu)^2. 
\end{align*}
\textbf{Proof of claim 2).} The result comes from the following
\begin{align*}
D_{KL}(\pi,\nu)&=\int \log \frac{\pi(x)}{\nu(x)}\pi(x)dx\\
&= -2 \int \log \frac{\sqrt{\nu(x)}}{\sqrt{\pi(x)}}\pi(x)dx\\
\mbox{ \quad(by $-\log(1+x)\geq -x$)}\quad &\geq -2 \int (\frac{\sqrt{\nu(x)}}{\sqrt{\pi(x)}}-1) \pi(x)dx\\
&=\int (\pi(x)+\pi(x)-2\sqrt{\pi(x)\nu(x)}) dx\\
&=\int (\pi(x)+\nu(x)-2\sqrt{\pi(x)\nu(x)}) dx=2 \Hell(\pi,\nu)^2. 
\end{align*}
\textbf{Proof of claim 3).} The result comes from the following
\begin{align*}
D_{TV}(\pi,\nu)^2&= \left(\frac12 \int\left|\frac{\pi(x)}{\lambda(x)}-\frac{\nu(x)}{\lambda(x)}\right| \lambda(x)dx\right)^2\\
&\leq \left(\frac12\int\left(\sqrt{\frac{\pi(x)}{\lambda(x)}}-\sqrt{\frac{\nu(x)}{\lambda(x)}}\right)^2 \lambda(x)dx\right)\; \left(\frac12
\int\left(\sqrt{\frac{\pi(x)}{\lambda(x)}}+\sqrt{\frac{\nu(x)}{\lambda(x)}}\right)^2 \lambda(x)dx\right)
\\
&\leq \Hell(\pi,\nu)^2\int\left(\frac{\pi(x)}{\lambda(x)}+\frac{\nu(x)}{\lambda(x)}\right) \lambda(x)dx
= 2 \Hell(\pi,\nu)^2. 
\end{align*}
\end{proof}

\begin{lem}
\label{lem:l2hell}
Consider two probability densities $\pi(x) = \frac{1}{Z_f} f(x) \mu(x)$ and $p(x) = \frac{1}{Z_h} h(x) \mu(x)$ where 
\(
Z_f = \int f(x) \mu(x) dx
\)
and
\(
Z_h = \int h(x) \mu(x) dx.
\)
Given the $L^2$ distance between $\sqrt{f}$ and $\sqrt{h}$
\[
\|\sqrt{f}-\sqrt{h}\|_{2,\mu} = \left(\int \big( \sqrt{f(x)} - \sqrt{h(x)}\big)^2 \mu(x) dx\right)^\frac12.
\]
Then we have the following:
\begin{enumerate}[1)]
\item The normalizing constant difference is bounded as $\left| \sqrt{Z_f}  - \sqrt{Z_h}\right| \leq \|\sqrt{f}-\sqrt{h}\|_{2,\mu}$.
\item The squared Hellinger distance is bounded as $\Hell( \pi, p )^2 \leq \frac{2}{Z_f}\|\sqrt{f}-\sqrt{h}\|_{2,\mu}^2 .$
\end{enumerate}
\end{lem}
\begin{proof}
\textbf{Proof of claim 1).} Note that 
\begin{align*}
\left| Z_f  - Z_h\right| & = \left| \int \big( f(x) - h(x) \big) \mu(x) dx \right| \\
& = \left| \int (\sqrt{f(x)} - \sqrt{h(x)}) (\sqrt{f(x)} + \sqrt{h(x)}) \mu(x) dx \right| \\
& \leq \left| \int (\sqrt{f(x)} - \sqrt{h(x)}) ^2 \mu(x) dx \right|^{1/2} \left| \int (\sqrt{f(x)} + \sqrt{h(x)})^2 \mu(x) dx \right|^{1/2} \\
& \leq \left| \int (\sqrt{f(x)} - \sqrt{h(x)}) ^2 \mu(x) dx \right|^{1/2} \left( \left| \int f(x) \mu(x) dx \right|^{1/2} + \left| \int h(x) \mu(x) dx \right|^{1/2} \right)\\
&=(\sqrt{Z_f}+\sqrt{Z_h})\left| \int (\sqrt{f(x)} - \sqrt{h(x)}) ^2 \mu(x) dx \right|^{1/2} 
\end{align*}
Dividing both sides by $(\sqrt{Z_f}+\sqrt{Z_h})>0$ we have the result.

\noindent\textbf{Proof of claim 2).} The squared Hellinger distance of $\pi$ from $p$ satisfies
\begin{align*}
D_H^2( \pi, p ) & = \frac12 \int \left(  \sqrt{\pi(x)} - \sqrt{p(x)}  \right)^2 dx \\
& = \frac12 \int \left(  \sqrt{\frac{f(x)}{Z_f}} - \sqrt{\frac{h(x)}{Z_h}}  \right)^2 \mu(x) dx \\
& = \frac1{2 Z_f} \int \left(  \sqrt{f(x)} - \sqrt{h(x)} + \sqrt{h(x)} - \sqrt{h(x)} \sqrt{\frac{Z_f}{Z_h}}  \right)^2 \mu(x) dx \\
& = \frac1{2 Z_f} \int \left(  \sqrt{f(x)} - \sqrt{h(x)} + \sqrt{h(x)} \left( 1 -  \sqrt{\frac{Z_f}{Z_h}} \right) \right)^2 \mu(x) dx \\
\mbox{(by Young's ineq.)}& \leq \frac1{ Z_f}\left( \int \left(  \sqrt{f(x)} - \sqrt{h(x)} \right)^2 \mu(x) dx + \left| 1 -  \sqrt{\frac{Z_f}{Z_h}} \right|^2\int h(x) \mu(x) dx \right)  \\
& = \frac1{Z_f} \left( \|\sqrt{f}-\sqrt{h}\|_{\mu}^2 + \left| \sqrt{Z_h} -  \sqrt{Z_f} \right|^2 \right)\leq \frac{2}{Z_f} \|\sqrt{f}-\sqrt{h}\|_{2,\mu}^2. 
\end{align*}
Thus, the result follows.
\end{proof}

\begin{lem}
\label{lem:eigpert}
Let $C\in\reals^{d\times d}$ be symmetric and  positive semidefinite,  $U\in \reals^{d\times d}$ be a rank $p$ symmetric matrix, for $p\geq 1$. Then for any $k$
\[
\lambda_{k+p}(C+U)\leq \lambda_{k} (C).
\]
\end{lem}
\begin{proof}
By the Courant--Fischer--Weyl min-max principle, we note that for any symmetric matrix $C$
\[
\lambda_{k+p}(C+U)=\min_V\{\max\{x^\top (C+U)x, \|x\|=1, x\in V\}, \text{dim}(V)=d-k-p+1\}.
\]
Let the eigenvectors of $C$ be $v_1,\ldots, v_d$ and the eigenvectors of $U$ with nonzero eigenvalues be $u_1,\ldots, u_p$. Now we pick 
\[
V_\bot=\text{span}\{v_{1},\ldots,v_{k-1},u_1,\ldots, u_p\},
\]
and its orthogonal complement $V'$ as a subspace of dimension at least $d-k-p+1$. Select any subspace $V$ of dimension $d-k-p+1$ from $V'$, then 
\[
\lambda_{k+p}(C+U)\leq \max_{x\in V,\|x\|=1}x^\top (C+U)x=\max_{x\in V,\|x\|=1}x^\top Cx\leq \lambda_{k}(C). 
\]
\end{proof}

\section{Proofs in Section 2}
\label{proof:sec2}

\subsection{Proof of Proposition \ref{prop:logconcave}}
\label{proof:logconcave}
Denote the density $\nu(x)\propto \exp(-V(x))$ and the associated conditional density as $\nu(x_\bot|x_r)$. 
Note that 
\begin{align*}
\mu(x_\bot|x_r)=\frac{\mu(x_r,x_\bot)}{\int \mu(x_r,x_\bot) dx_\bot}=\frac{\exp(-V(x_r,x_\bot))\exp(-U(x_r,x_\bot))}{\int \exp(-V(x_r,x_\bot))\exp(-U(x_r,x_\bot))dx_\bot}.
\end{align*}
Let $c_0=\inf_x \exp(-U(x))$. Then $\exp(-U(x_r,x_\bot))\leq B c_0$, so 
\[
\mu(x_\bot|x_r)\leq \frac{\exp(-V(x_r,x_\bot))Bc_0}{\int \exp(-V(x_r,x_\bot))c_0 dx_\bot}=B \nu(x_\bot|x_r).
\]
Likewise
\[
\mu(x_\bot|x_r)\geq \frac{\exp(-V(x_r,x_\bot))c_0}{\int \exp(-V(x_r,x_\bot))c_0 B dx_\bot}=B^{-1}\nu(x_\bot|x_r).
\]
Finally, note that  
\[
-\nabla^2_{x_\bot}  \log \nu(x_\bot|x_r)=-\nabla^2_{x_\bot} \log \nu(x_\bot,x_r),
\]
which is a sub-matrix of $-\nabla^2 \log \nu(x)$, so its minimal eigenvalue is greater than $c$ by the assumption of strong log-concavity.
Then the Bakry--Emery principle (see, e.g. Theorem 3.1 of \cite{menz2014poincare}) indicates that $\nu(x_\bot|x_r)$ satisfies the Poincar\'{e} inequality with coefficient $c$, i.e. for any $h\in C^1$
\[
\var_{\nu(x_\bot|x_r)} [h]\leq  \frac1{c}\int \|\nabla h(x_r,x_\bot)\|^2 \nu(x_\bot|x_r) dx_\bot.
\]
Finally, we have the Poincar\'{e} inequality for $\mu(x_\bot|x_r)$:
\begin{align*}
\var_{\mu(x_\bot|x_r)} [h]&\leq \int [h(x)-\E_{\nu(x_\bot|x_r)} [h (x)]]^2 \mu(x_\bot|x_r)dx_\bot\\
&\leq B\int [h(x)-\E_{\nu(x_\bot|x_r)} [h(x)]]^2 \nu(x_\bot|x_r)dx_\bot\\
&=B \var_{\nu(x_\bot|x_r)} (h)\\
&\leq  \frac B{c}\int \|\nabla_{x_\bot} h(x_r,x_\bot)\|^2 \nu(x_\bot|x_r)dx_\bot\\
&\leq \frac {B^2}{c}\int \|\nabla_{x_\bot} h(x_r,x_\bot)\|^2 \mu(x_\bot|x_r) dx_\bot.
\end{align*}
\qed

\subsection{Proof of Proposition \ref{prop:pivar}} \label{proof:pivar}
\textbf{Proof of claim 1).} Recall the squared Hellinger distance
\begin{align}
\notag
\Hell(\pi, \apdf_f)^2 &=\frac{1}{2}\int \left(  \sqrt{\frac{\pi(x)}{\mu(x)}}- \sqrt{\frac{\apdf_f(x)}{\mu(x)}} \right)^2 \mu(x)dx \\
\label{tmp:11}
&= \frac{1}{Z}\int \left( \frac12\int (f(x_r, x_\bot)^\frac12 - \bar{f}(x_r)^\frac12)^2 \mu(x_\bot|x_r)dx_\bot \right) \bar{\mu}(x_r)dx_r ,
\end{align}
and definitions of $\bar{f}(x_r)$ and $\bar{g}(x_r)$:  
\[
\bar{f}(x_r)  = \int g(x_r, x_\bot)^2 \mu(x_\bot|x_r)dx_\bot,\quad
\bar{g}(x_r)  = \int g(x_r, x_\bot) \mu(x_\bot|x_r)dx_\bot.
\]
We have the identity:
\begin{equation}
\label{rev:1}
\var_{\mu(x_\bot|x_r)} \big[ g \big] = \bar{f}(x_r) - \bar{g}(x_r)^2.
\end{equation}
The inner integral in \eqref{tmp:11} can be expressed as
\begin{align*}
&\frac12 \int (g(x_r, x_\bot) - \bar{f}(x_r)^\frac12)^2 \mu(x_\bot|x_r)dx_\bot\\
 & = \frac12 \int \big(g(x_r, x_\bot)^2 + \bar{f}(x_r) - 2g(x_r, x_\bot)\bar{f}(x_r)^\frac12 \big) \mu(x_\bot|x_r)dx_\bot \\
& = \bar{f}(x_r) - \bar{f}(x_r)^\frac12\bar{g}(x_r) .
\end{align*}
By the Cauchy--Schwarz inequality, we have $\bar{f}(x_r)^\frac12 \geq \bar{g}(x_r) \geq 0$, and therefore
\[
\bar{f}(x_r) \geq \bar{f}(x_r)^\frac12\bar{g}(x_r) \geq \bar{g}(x_r)^2 \geq 0.
\]
This leads to the inequality
\begin{align*}
\bar{f}(x_r) - \bar{f}(x_r)^\frac12\bar{g}(x_r) & \leq \bar{f}(x_r) - \bar{g}(x_r)^2 = \var_{\mu(x_\bot|x_r)} \big[ g \big].
\end{align*}
Applying this bound above to \eqref{tmp:11}, we find that:
\begin{align*}
\Hell(\pi, \apdf_f)^2 & \leq  \frac{1}{Z}  \int  \var_{\mu(x_\bot|x_r)} \big[ g \big] \bar{\mu}(x_r)dx_r .
\end{align*}
\textbf{Proof of claim 2).} Recall that the normalizing constant of $\apdf_g$ takes the form
\[
Z_g:=\int \bar{g}(x_r)^2 \mu(x)dx. 
\] 
The squared Hellinger distance can be written as 
\begin{equation}
\label{tmp:12}
\Hell(\pi, \apdf_g)^2 = \frac12 \int \Big( \int \Big(\frac{1}{Z^\frac12} g(x_r, x_\bot) - \frac{1}{Z_g^\frac12} \bar{g}(x_r) \Big)^2 \mu(x_\bot|x_r)dx_\bot \Big) \bar{\mu}(x_r)dx_r.
\end{equation}
Using the identity \eqref{rev:1} and $\bar{f}(x_r)^\frac12 \geq \bar{g}(x_r) \geq 0$,
we have
\[
Z = \int f(x) \mu(x) dx = \int \bar{f}(x_r) \bar{\mu}(x_r) dx_r \geq \int \bar{g}(x_r)^2 \bar{\mu}(x_r) dx_r = \int \bar{g}(x_r)^2 \mu(x) dx =  Z_g > 0,
\]
and therefore $\frac{Z_g}{Z} \leq 1$. Then, we can bound the inner integral in \eqref{tmp:12} by
\begin{align*}
\int \Big(\frac{1}{Z^\frac12} g(x_r, x_\bot) - \frac{1}{Z_g^\frac12} \bar{g}(x_r) \Big)^2 \mu(x_\bot|x_r)dx_\bot & = \frac{1}{Z} \Big( \bar{f}(x_r) + \frac{Z}{Z_g} \bar{g}(x_r)^2 - 2\sqrt{\frac{Z}{Z_g}} \bar{g}(x_r)^2 \Big) \\
& \leq \frac{1}{Z} \Big( \bar{f}(x_r) + \frac{Z}{Z_g} \bar{g}(x_r)^2 - 2 \bar{g}(x_r)^2 \Big) \\
& = \frac{1}{Z} \Big( \bar{f}(x_r) + \frac{Z - Z_g}{Z_g} \bar{g}(x_r)^2 - \bar{g}(x_r)^2 \Big) \\
& =\frac{1}{Z} \Big( \var_{\mu(x_\bot|x_r)} \big[ g \big] + \frac{Z - Z_g}{Z_g} \bar{g}(x_r)^2 \Big).
\end{align*}
Substituting this upper bound into \eqref{tmp:12}, we have
\begin{align*}
\Hell(\pi, \apdf_g)^2 & \leq \frac1{2Z} \int \Big( \var_{\mu(x_\bot|x_r)} \big[ g \big] + \frac{Z - Z_g}{Z_g} \bar{g}(x_r)^2 \Big) \bar{\mu}(x_r)dx_r \\
 & = \frac1{2Z} \Big( \int \var_{\mu(x_\bot|x_r)} \big[ g \big] \bar{\mu}(x_r)dx_r + (Z - Z_g) \Big) .
\end{align*}
The term $Z - Z_g$ satisfies
\begin{align}
Z - Z_g & = \int \Big(\int (g(x_r, x_\bot)^2 - \bar{g}(x_r)^2 ) \mu(x_\bot|x_r)dx_\bot\Big) \bar{\mu}(x_r) dx_r\nonumber\\
& = \int \var_{\mu(x_\bot|x_r)} \big[ g \big] \bar{\mu}(x_r) dx_r.\label{eq:zzg}
\end{align}
In summary, we have
\begin{align*}
\Hell(\pi, \apdf_g)^2 & \leq \frac1{Z} \Big( \int \var_{\mu(x_\bot|x_r)} \big[ g \big] \bar{\mu}(x_r)dx_r \Big).
\end{align*}
\qed

\subsection{Proof of Theorem \ref{thm:pig}}\label{proof:pig}

\textbf{Claim 1).} Recall that in \eqref{eqn:decompose}, the projector $P_r$ satisfies ${\rm range}(P_r) = \calX_r$. By the Poincar\'{e} inequality of $\mu(x_\bot|x_r)$, the expected conditional variance of $g$ satisfies
\[
\var_{\mu(x_\bot|x_r)} \big[ g \big] \leq \kappa\int \|(I-P_r)\nabla g(x)\|^2 \mu(x_\bot|x_r)dx_\bot.
\]
Applying Proposition \ref{prop:pivar}, we have
\begin{align*}
\Hell(\pi, \apdf_f)^2 & \leq  \frac{\kappa}{Z}  \int \Big( \int \|(I-P_r)\nabla g(x)\|^2 \mu(x_\bot|x_r)dx_\bot \Big) \mu(x_r)dx_r \\
& = \frac{\kappa}{Z}  \int \|(I-P_r)\nabla g(x)\|^2 \mu(x)dx \\
& = \frac{\kappa}{Z} \int \|P_\bot\nabla \log g(x)\|^2 g(x)^2\mu(x)dx\\
& =\frac{\kappa}{4}  \int \|P_\bot \nabla \log f(x)\|^2 \pi(x)dx.
\end{align*}
Since $\|P_\bot \nabla \log f(x)\|^2 = P_\bot \nabla \log f(x)\nabla \log f(x)^\top P_\bot $, the result follows from
\[
\int \|P_\bot \nabla \log f(x)\|^2 \pi(x)dx = \res(\calX_r,H_1).
\] 
\textbf{Claim 2).} This result follows from Lemma \ref{lem:hel2exp} and claim 1).\\
\textbf{Claim 3).} The same proofs of claims 1) and 2) can be applied.
\hfill \qed

\subsection{Proof of Theorem \ref{thm:pil}}
\label{proof:pil}
\textbf{Proof of claim 1).}
Note that 
\begin{align*}
\int (\log f(x)-\bar{l}(x_r))^2\mu(x_\bot|x_r)dx_\bot&=\var_{\mu(x_\bot|x_r)}[\log f(x)]\\
&\leq  \kappa\int \|\nabla_{x_\bot}\log f(x_\bot,x_r)\|^2 \mu(x_\bot|x_r)dx_\bot\\
&=  \kappa\int \|P_\bot\nabla \log f(x)\|^2 \mu(x_\bot|x_r)dx_\bot.
\end{align*}
Integrating both sides with respect to $\bar{\mu}(x_r)$  yields 
\[
\int (\log f(x)-\bar{l}(x_r))^2\mu(x)dx
\leq \kappa \res(\calX_r, H_{0}).
\]
Then by the Cauchy--Schwarz inequality, we find
\[
\int (\log f(x)-\bar{l}(x_r))^2\mu(x)dx \int f^2(x)\mu(x)dx
\geq Z^2\left(\int \log\frac{\pi(x) Z}{\apdf_l(x) Z_l}\pi(x) dx\right)^2,
\]
where 
\begin{align*}
Z_l&:=\int \exp(\bar{l}(x_r))\bar{\mu}(x_r)dx_r\\
&=\int \exp \left(\int\log f(x)\mu(x_\bot|x_r)dx_\bot\right)\bar{\mu}(x_r)dx_r\\
&\leq \int \left(\int f(x)\mu(x_\bot|x_r) dx_\bot\right)\bar{\mu}(x_r)dx_r=Z. \mbox{\quad(by Jensen's ineq.)}
\end{align*}
Moreover, it is well known that $D_{KL}(\pi,\apdf_l)=\int \log\frac{\pi(x) }{\apdf_l(x) }\pi(x) dx\geq 0$, so 
\[
\left(\int \log\frac{\pi(x) Z}{\apdf_l(x) Z_l}\pi(x) dx\right)^2= 
\left(D_{KL}(\pi,\apdf_l)+\log Z/Z_l\right)^2\geq D_{KL}^2(\pi,\apdf_l).
\]
In conclusion, we have claim 1) by
\begin{align*}
D_{KL}^2(\pi,\apdf_l)&\leq \left(\int \log\frac{\pi(x) Z}{\apdf_l(x) Z_l}\pi(x) dx\right)^2\\
&\leq \frac1{Z^2}\int (\log f(x)-\bar{l}(x_r))^2\mu(x)dx \int f^2(x)\mu(x)dx\\
&\leq\frac{\|f\|_{2,\mu}^2\kappa}{Z^2}\int  \|P_\bot\nabla \log f(x)\|^2\mu(x)dx=\frac{\|f\|_{2,\mu}^2\kappa}{Z^2}\res(\calX_r,H_0).  
\end{align*}
\noindent\textbf{Proof of claim 2).} Applying claim 1) of Lemma \ref{lem:hel2exp},  claim 2) of Lemma \ref{lem:hel2exp}, and then  claim 1) of Theorem \ref{thm:pil},  we have
\begin{align*}
|\E_{\pi} [h]-\E_{\apdf_l} [h]|&\leq\sqrt{ 2(\E_{\pi} [h^2]+\E_{\apdf_l}[h^2])} \Hell(\pi,\apdf_l)\\
&\leq (\E_{\pi} [h^2]+\E_{\apdf_l}[h^2])^{\frac12} \sqrt{D_{KL}(\pi,\apdf_l)}\\
&\leq (\E_{\pi} [h^2]+\E_{\apdf_l}[h^2])^{\frac12}\sqrt{\frac{\|f\|_{2,\mu}}{Z}} (\kappa\res(\calX_r, H_{0}))^{\frac14}. 
\end{align*}
Thus, the result follows.\qed

\section{Proofs in Section \ref{sec:MCerror}}

\subsection{Proof of Theorem \ref{thm:MCerrorfg}} \label{sec:proof_thm3fg}
\noindent\textbf{Proof of claim 1).} 
Recalling the function $g$ takes the form $g=\sqrt{f}$ and using
\[
\bar{g}^M(x_r)=\frac1{M}\sum_{i=1}^{M} g(x_r, X^i_\bot)=\frac1{M}\sum_{i=1}^{M} g(x_r, T(x_r,W^i)), \quad {\rm and}\quad \bar{g}(x_r)=\int g(x_r,x_\bot)\mu(x_\bot|x_r)dx_\bot,
\]
we have the corresponding approximate target densities $\apdf_g^M(x_r, x_\bot) \propto \bar{g}^M(x_r)^2 \mu(x_r, x_\bot) $ and $\apdf_g(x_r, x_\bot) \propto \bar{g}(x_r)^2 \mu(x_r, x_\bot) $, respectively. Note that 
\[
\E_M[\bar{g}^M(x_r)]=\frac1{M}\sum_{i=1}^{M} \E_M g(x_r, T(x_r,W^i)) =\bar{g}(x_r),
\]
\[
\E_M[(\bar{g}^M(x_r)-\bar{g}(x_r))^2]=\frac1{M}\sum_{i=1}^{M} \var_M [g(x_r, T(x_r,W^i))] =\text{var}_{\mu(x_\bot|x_r)}[g(x_r, x_\bot)].
\]
Applying Lemma \ref{lem:l2hell} claim 2), we have
\begin{align*}
\E_M \left[ \Hell(\apdf_g^M,\apdf_g)^2 \right]&\leq 
\frac{2}{Z_g}\E_M\left[\int \int \big( \bar{g}^M(x_r) - \bar{g}(x_r)\big)^2 \bar{\mu}(x_r) dx_r \, \mu(x_\bot|x_r) d x_\bot\right]\\ 
& = \frac{2}{Z_gM} \int \text{var}_{\mu(x_\bot|x_r)}(g(x_r, x_\bot)) \bar{\mu}(x_r) d x_r\\
&\leq \frac{2\kappa}{Z_gM}  \int \|\nabla_{x_\bot}g(x)\|^2 \mu(x)dx\\
&= \frac{2\kappa}{Z_gM}  \int \|\nabla_{x_\bot}\log g(x)\|^2 f(x)\mu(x)dx\\
& =\frac{2\kappa Z}{Z_g M} \res(\calX_r, H_1).   
\end{align*}
Thus, the result follows from Jensen's inequality.

\noindent\textbf{Proof of claim 2).} 
We have the corresponding approximate target densities $\apdf_f^M(x_r, x_\bot) \propto \bar{f}^M(x_r) \mu(x_r, x_\bot) $ and $\apdf_f(x_r, x_\bot) =\frac{1}{Z} \bar{f}(x_r) \mu(x_r, x_\bot) $ where
\[
\bar{f}^M(x_r)=\frac1{M}\sum_{i=1}^{M} f(x_r, X^i_\bot)=\frac1{M}\sum_{i=1}^{M} f(x_r, T(x_r,W^i)), \quad {\rm and}\quad \bar{f}(x_r)=\int f(x_r,x_\bot)\mu(x_\bot|x_r)dx_\bot.
\]
Similar to the proof of claim 1), we apply Lemma \ref{lem:l2hell} claim 2) and find
\begin{align}
\E_M \left[ \Hell(\apdf_f^M,\apdf_f)^2 \right]&\leq \frac{2}{Z}\E_M\left[\int \int \big( \sqrt{\bar{f}^M(x_r)} - \sqrt{\bar{f}(x_r)}\big)^2 \, \mu(x_r,x_\bot) dx_r d x_\bot \right] \nonumber\\
&= \frac{2}{Z}\E_M\left[\int \int \big( \sqrt{\bar{f}^M(x_r)} - \sqrt{\bar{f}(x_r)}\big)^2 \bar{\mu}(x_r) dx_r \, \mu(x_\bot|x_r) d x_\bot\right] \nonumber \\ 
& = \frac{2}{Z}\int \E_M\left[\big( \sqrt{\bar{f}^M(x_r)} - \sqrt{\bar{f}(x_r)}\big)^2\right] \bar{\mu}(x_r) dx_r ,\label{eq:pifm1}
\end{align}
here. Considering the identity 
\begin{align*}
(\sqrt{\bar{f}^M(x_r)}-\sqrt{\bar{f}(x_r)})^2&\leq (\sqrt{\bar{f}^M(x_r)}-\sqrt{\bar{f}(x_r)})^2(\sqrt{\bar{f}^M(x_r)/\bar{f}(x_r)}+1)^2\\
&=(\bar{f}^M(x_r)-\bar{f}(x_r))^2/\bar{f}(x_r),
\end{align*}
the inequality in \eqref{eq:pifm1} also satisfies 
\begin{align}
\E_M \left[ \Hell(\apdf_f^M,\apdf_f)^2 \right] & \leq  \frac{2}{Z} \int \frac{\E_M \left[ (\bar{f}^M(x_r)-\bar{f}(x_r))^2\right]}{\bar{f}(x_r)} \bar{\mu}(x_r) dx_r .\label{eq:pifm2}
\end{align}
Then for each given $x_r$, by independence of $x_\bot^i$ we have 
\begin{align*}
\E_M \left[(\bar{f}^M(x_r)-\bar{f}(x_r))^2\right]&=\frac{1}{M}\text{var}_{\mu(x_\bot|x_r)} f(x_r,x_\bot)\\
&\leq \frac{\kappa}{M} \int \|\nabla_{x_\bot}f(x_r,x_\bot)\|^2\mu(x_\bot|x_r)dx_\bot\\
&=\frac{\kappa}{M} \int \|\nabla_{x_\bot}\log f(x_r,x_\bot)\|^2f(x)^2\mu(x_\bot|x_r)dx_\bot\\
&=\frac{\kappa}{M} \bar{f}(x_r)^2\int \|\nabla_{x_\bot}\log f(x_r,x_\bot)\|^2f(x_\bot |x_r)^2\mu(x_\bot|x_r)dx_\bot.\\
\end{align*}
Substituting the above identify into \eqref{eq:pifm2}, we have
\begin{align*}
\E_M \left[ \Hell(\apdf_f^M,\apdf_f)^2 \right] & \leq \frac{2\kappa}{ZM} \int \int \|\nabla_{x_\bot}\log f(x)\|^2f(x_\bot |x_r)^2\mu(x_\bot|x_r) dx_\bot \, \bar{f}(x_r) \bar{\mu}(x_r)dx_r \\
& = \frac{2\kappa}{ZM} \int \|\nabla_{x_\bot}\log f(x)\|^2f(x_\bot |x_r)  \, f(x) \mu(x) dx \\
& \leq \frac{2 \kappa C_f}{M}\res(\calX_r, H_1),
\end{align*}
where $C_f = \sup_{x_r} \sup_{x_\bot}f(x_\bot|x_r)$. Then, the result follows from Jensen's inequality. \qed

\subsection{Proof of Theorem \ref{thm:MCerrorl}} \label{sec:proof_thm3l}
By the independence of the Monte Carlo samples, we have
\begin{align*}
\E_M \left[(\bar{l}^M(x_r)-\bar{l}(x_r))^2\right]
&=\E_M \left[\left(\frac{1}{M}\sum_{i=1}^Ml(x_r,T(x_r,W^i))-\bar{l}(x_r)
\right)^2\right]\\
&=\frac1{M} \var_{\mu(x_\bot|x_r)} \left[ l(x_r,x_\bot)\right],
\end{align*}
which leads to the following using Assumption \ref{aspt:poincare},
\[
\E_M \left[\int (\bar{l}^M(x_r)-\bar{l}(x_r))^2\bar{\mu}(x_r)dx_r \right]\leq \frac{\kappa}{M} \int \var_{\mu(x_\bot|x_r)} \left[ l(x_r,x_\bot)\right] \bar{\mu}(x_r)dx_r.
\]
By Jensen's inequality, the expected $L^2$ error of $\bar{l}^M(x_r)$ satisfies
\begin{equation}
\E_M\left[ \left(\int (\bar{l}^M(x_r)-\bar{l}(x_r))^2\bar{\mu}(x_r)dx_r \right)^\frac12\right]\leq \frac{\sqrt{\kappa}}{\sqrt{M}} \left(\int \var_{\mu(x_\bot|x_r)} \left[ l(x_r,x_\bot)\right] \bar{\mu}(x_r)dx_r\right)^\frac12.\label{eq:l2lfm}
\end{equation}
Assumption \ref{aspt:poincare} states that
\[
\int \var_{\mu(x_\bot|x_r)} \left[ l(x_r,x_\bot)\right] \bar{\mu}(x_r)dx_r \leq \kappa \int\int \|\nabla_{x_\bot}l(x_r,x_\bot)\|^2\mu(x_\bot|x_r)dx_\bot \bar{\mu}(x_r)dx_r= \kappa \res(\calX_r,H_0),
\]
together with \eqref{eq:l2lfm}, we have 
\[
\E_M\left[ \left(\int (\bar{l}^M(x_r)-\bar{l}(x_r))^2\bar{\mu}(x_r)dx_r \right)^\frac12\right]\leq \frac{\sqrt{\kappa}}{\sqrt{M}} \sqrt{\res(\calX_r,H_0)}.
\]

To obtain the KL divergence, we note that
\begin{align*}
\E_M\left[ D_{KL}(\pi , \apdf^M_{l}) \right] & =  \E_M\left[ \int \log\frac{\pi(x)}{\apdf^M_{l}(x)} \pi(x) d x \right]\\
& =   \E_M\left[ \int  \left( l(x) - \bar{l}^M(x_r) \right) \pi(x) d x \right] + \E_M\left[ \log\frac{Z_M}{Z} \right] ,
\end{align*}
where 
\(
Z_M = \int  e^{\bar{l}^M(x_r)} \bar{\mu}(x_r) d x_r.
\)
The expectation of $Z_M$ satisfies
\begin{align*}
\E_M[Z_M] & = \E_M\left[  \int e^{\bar{l}^M(x_r)} \bar{\mu}(x_r) d x_r  \right] \\
& =  \int \E_M\left[ \exp\left( \frac{1}{M} \sum_{i = 1}^{M} l(x_r, x_\bot^i)\right) \right] \bar{\mu}(x_r) d x_r \\
& \leq \int \left( \E_M\left[ \prod_{i = 1}^{M} \exp\left( l(x_r, x_\bot^i)\right) \right] \right)^{1/M} \bar{\mu}(x_r) d x_r \\
& = \int \left( \prod_{i = 1}^{M} \int \exp\left( l(x_r, x_\bot^i)\right) \mu(x^i_\bot|x_r) dx^i_\bot\right)^{1/M} \bar{\mu}(x_r) d x_r \\
& = \int \int f(x_r, x_\bot) \mu(x_\bot|x_r) dx_\bot \bar{\mu}(x_r) d x_r = Z.
\end{align*}
Therefore, by Jensen's inequality, we have
\[
\E_M\left[ \log\frac{Z_M}{Z} \right] \leq \log \E_M\left[\frac{Z_M}{Z} \right] \leq 0.
\]
Thus, the expected KL satisfies
\begin{align}
\notag
\E_M\left[ D_{KL}(\pi, \apdf_l^M) \right] & \leq  \E_M\left[ \int  \left( l(x) - \bar{l}(x_r) +  \bar{l}(x_r) - \bar{l}^M(x_r) \right) \pi(x) d x \right] \\
\label{tmp:KL}
& =  \int  \left( l(x) - \bar{l}(x_r) \right) \pi(x) d x  + \E_M\left[ \int  \left( \bar{l}(x_r) - \bar{l}^M(x_r) \right) \pi(x) d x \right].
\end{align}
Applying the Cauchy--Schwarz inequality, the first term in \eqref{tmp:KL} can be bounded by
\begin{align*}
\int  \left( l(x) - \bar{l}(x_r) \right) \pi(x) d x & \leq \frac{\|f\|_{2,\mu}}{Z} \left( \int  \left( l(x) - \bar{l}(x_r) \right)^2 \mu(x) d x \right)^\frac12 \\
& = \frac{\|f\|_{2,\mu}}{Z}  \left(\int \text{var}_{\mu(x_\bot|x_r)}\left[l(x_r, x_\bot)\right]  \bar{\mu}(x_r) d x_r\right)^{\frac12},
\end{align*}
and the  second term in \eqref{tmp:KL} can be bounded by 
\begin{align*}
\E_M\left[ \int  \left( \bar{l}(x_r) - \bar{l}^M(x_r) \right) \pi(x) d x \right] & \leq  \frac{\|f\|_{2,\mu}}{Z} \E_M\left[\left( \int  \left( \bar{l}(x_r) - \bar{l}^M(x_r) \right)^2 \bar{\mu}(x_r) d x_r \right)^\frac12 \right] .
\end{align*}
Thus, applying the bound on $\E_M[\cdot]$ in \eqref{eq:l2lfm} and Assumption \ref{aspt:poincare}, we have
\begin{align*}
\E_M\left[ D_{KL}(\pi, \apdf_l^M) \right] & \leq  \frac{\|f\|_{2,\mu}}{Z} \left( 1 + \frac1{\sqrt{M}}\right) \left(\int \text{var}_{\mu(x_\bot|x_r)}\left[l(x_r, x_\bot)\right]   \bar{\mu}(x_r) d x_r\right)^{1/2}\\
&\leq \frac{\|f\|_{2,\mu}}{Z} \left( 1 + \frac1{\sqrt{M}} \right) \left(\kappa \int \|\nabla_{x_\bot} l(x_r,x_\bot)\|^2 \mu(x)dx\right)^{1/2}\\
 &=\frac{\sqrt{\kappa} \|f\|_{2,\mu}}{Z} \left( 1 + \frac1{\sqrt{M}} \right) \sqrt{\res(\calX_r,H_0)}. 
\end{align*}

\qed

\newcommand\numleq[1]{\stackrel{\scriptscriptstyle(\mkern-1.5mu#1\mkern-1.5mu)}{\leq}}

\section{Proofs in Section \ref{sec:MCsubspace}}

\subsection{Proof of Lemma \ref{lem:matrix}}\label{proof:matrix}
Given two positive semidefinite matrices $\Sigma, \widehat{\Sigma} \in \R^{d\times d}$, 
let  $\Vhat_r$ be the matrix consisting of the $d_r$ leading orthonormal eigenvectors of $\Sigmahat$ such that $\Vhat_r^\top\Vhat_r=I_{d_r}$, and $\Lambdahat_r$ be the $d_r\times d_r$ diagonal matrices consisting of the $d_r$ leading eigenvalues of $\Sigmahat$ as its diagonal entries.
Similarly, let $V_r$ and $\Lambda_r$ be the matrices consisting of the $d_r$ leading orthonormal eigenvectors of $\Sigma$ and the $d_r$ leading eigenvalues of $\Sigma$, respectively.
We can define the orthogonal projectors $\Phat_r = \Vhat_r \Vhat_r^\top$ and $\Phat_\bot = I - \Phat_r$.

\noindent\textbf{Proof of claim 1).} 
Since $\tr(\Phat_\bot \Sigma \Phat_\bot)=\tr(\Sigma \Phat_\bot^2)=\tr(\Sigma \Phat_\bot) = \tr(\Sigma )-\tr(\Sigma \Phat_r)$, we have
\begin{align*}
\tr(\Phat_\bot \Sigma \Phat_\bot)&=\tr(\Sigma )-\tr(\Lambda_r )+\tr(\Lambda_r )-\tr(\Sigma \Phat_r) + \tr(\Sigmahat \Phat_r) - \tr(\Sigmahat \Phat_r) \nonumber \\
&=\tr(\Sigma )-\tr(\Lambda_r )+\tr(\Lambda_r ) - \tr(\Sigmahat \Phat_r)  + \tr((\Sigmahat-\Sigma) \Phat_r). 
\end{align*}
The definition of the eigenvalue problem $\Sigmahat \Vhat_r  = \Vhat_r\Lambdahat_r$ gives $\tr(\Sigmahat \Phat_r) = \tr(\Sigmahat \Vhat_r \Vhat_r^\top) = \tr(\Lambdahat_r)$. Together with $\tr(\Sigma ) = \sum_{i=1}^d \lambda_i(\Sigma)$, we have
\begin{align}
\tr(\Phat_\bot \Sigma \Phat_\bot) & = \sum_{i={d_r+1}}^d \lambda_i(\Sigma) + \tr(\Lambda_r - \Lambdahat_r)  + \tr((\Sigmahat-\Sigma) \Phat_r)%
. \label{eq:trineq0}
\end{align}
The term $\tr(\Lambda_r - \Lambdahat_r) $ satisfies
\begin{align*}
\tr(\Lambda_r - \Lambdahat_r) & \leq \sum_{i =1}^{d_r} \left| \lambda_i(\Sigma) -\lambda_i(\Sigmahat)  \right| \leq  \sqrt{d_r} \left(\sum_{i =1}^{d_r} \left( \lambda_i(\Sigma) -\lambda_i(\Sigmahat) \right)^2 \right)^{1/2} .
\end{align*}
Since $\Sigma$ and $\Sigmahat$ are both symmetric, applying Theorem 6.11 of  \cite{Kat82}, we have 
\begin{equation*}
\sum_{i =1}^{d_r} \left( \lambda_i(\Sigma) -\lambda_i(\Sigmahat) \right)^2 \leq \sum_{i =1}^{d} \left( \lambda_i(\Sigma) -\lambda_i(\Sigmahat) \right)^2 \leq \sum_{i =1}^{d} \lambda_i(\Sigma - \Sigmahat) ^2 = \|\Sigma - \Sigmahat\|_F^2,
\end{equation*}
which leads to
\begin{equation}\label{eq:tr1}
\tr(\Lambda_r - \Lambdahat_r) \leq \sqrt{d_r}\|\Sigma - \Sigmahat\|_F.
\end{equation}
Since for any matrix $A \in \R^{d_r \times d_r}$, it satisfies 
\(
\tr(A) \leq \sqrt{d_r}\sqrt{\sum_{i=1}^{d_r} A_{ii}^2} \leq \sqrt{d_r}\|A\|_F,
\)
the term $\tr(\Phat_r ( \Sigmahat - \Sigma))$ satisfies
\begin{equation}\label{eq:tr2}
\tr(\Phat_r ( \Sigmahat - \Sigma)) = \tr(\Vhat_r^\top (\Sigmahat- \Sigma) \Vhat_r ) \leq  \sqrt{d_r} \| \Vhat_r^\top (\Sigmahat - \Sigma) \Vhat_r \|_F\leq \sqrt{d_r} \|\Sigmahat - \Sigma \|_F,
\end{equation}
where the last inequality follows from the property 
\begin{equation*}
\|AB\|^2_F= \tr(ABB^\top A^\top) =\tr(BB^\top A^\top A)\leq \|BB^\top\| \tr( A^\top A)=\|B\|^2\|A\|_F^2. 
\end{equation*}
Substituting \eqref{eq:tr1} and \eqref{eq:tr2} into \eqref{eq:trineq0}, the result of claim 1) follows.

\noindent\textbf{Proof of claim 2).} 
The approximation residual can be expressed as
\begin{align*}
\tr(\Phat_\bot \Sigma \Phat_\bot) & = \tr(\Phat_\bot \Sigmahat \Phat_\bot) + \tr(\Phat_\bot (\Sigma-\Sigmahat) \Phat_\bot) \\
& = \sum_{i=d_r+1}^d\lambda_i(\Sigmahat) + \tr(\Phat_\bot ( \Sigma - \Sigmahat)) + \tr(\Phat_r ( \Sigma - \Sigmahat))- \tr(\Phat_r ( \Sigma - \Sigmahat))\\
& = \sum_{i=d_r+1}^d\lambda_i(\Sigmahat) + \tr(\Sigma - \Sigmahat) + \tr(\Phat_r ( \Sigmahat - \Sigma)).
\end{align*}
Applying \eqref{eq:tr2}, then the result of claim 2) follows. 
\qed

\subsection{Proof of Proposition \ref{prop:Vbound}}\label{proof:Vbound}
For any $i,j$-th component of $\text{V}(H_0, \nu)$, we have
\[
\text{var}_{X\sim \nu}\left[\partial_i \log f(X)\partial_j \log f(X) \frac{\mu(X)}{\nu(X)}\right]
\leq \E_{X\sim \nu}\left[\left(\partial_i \log f(X)\partial_j \log f(X) \frac{\mu(X)}{\nu(X)}\right)^2\right].
\]
This way, summing over all indices, we have
\[
\text{V}(H_0, \nu)\leq \E_{X\sim \nu} \left[\sum_{i,j=1}^d \left(\partial_i \log f(X)\partial_j f(X) \frac{\mu(X)}{\nu(X)}\right)^2\right]
= \E_{X\sim \nu}\left[\|\nabla \log f(X)\|^4 \left(\frac{\mu(X)}{\nu(X)}\right)^2\right].
\]
The bound on $V(H_1,\nu)$ can be shown similarly by replacing $\mu$ with $\pi$.
\qed

\section{Proofs in Section \ref{sec:MCMC}}

\subsection{Proof of Proposition \ref{prop:MCMC}}\label{proof:MCMC}
We denote the composite transition density of lines 3--5 of Algorithm \ref{alg:MCMC} by $Q(x,x')$.
We first verify the detailed balance condition $\pi(x)Q(x,x')=\pi(x')Q(x',x)$. Note that for $x_r\neq x_r', x_\bot\neq x_\bot'$, the overall transition density is
\[
Q(x,x')=p(x_r,x'_r)\beta(x_r,x_r')\mu(x'_\bot|x'_r) \alpha(x,x').
\]
Note that by the formulation of $\beta$, the following detailed balance condition holds
\begin{equation}
\label{rev:2}
p(x_r,x_r')\bar{\apdf}_s(x_r)\beta(x_r,x_r')=p(x'_r,x_r)\bar{\apdf}_s(x'_r)\beta(x'_r,x_r).
\end{equation}
These lead to 
\begin{align*}
\pi(x)Q(x,x')&=\pi(x)p(x_r,x'_r)\beta(x_r,x_r')\mu(x'_\bot|x'_r) \alpha(x,x')\\
&=\pi(x)\frac{\bar{\apdf}_s(x'_r)}{\bar{\apdf}_s(x_r)}p(x'_r,x_r)\beta(x'_r,x_r)\mu(x'_\bot|x'_r) \alpha(x,x')\quad \mbox{( by \eqref{rev:2})}\\
&=p(x'_r,x_r)\beta(x'_r,x_r) \pi(x') \frac{\pi(x)}{\pi(x')} \frac{\bar{\apdf}_s(x'_r)}{\bar{\apdf}_s(x_r)} \mu(x'_\bot|x'_r) \left(1 \wedge \frac{f(x') \bar{\apdf}_s(x_r)\bar{\mu}(x'_r)}{f(x) \bar{\apdf}_s(x'_r)\bar{\mu}(x_r)}\right) \\
&=p(x'_r,x_r)\beta(x'_r,x_r) \pi(x') \frac{f(x)\mu(x)}{f(x')\mu(x')} \frac{\bar{\apdf}_s(x'_r)\mu(x')}{\bar{\apdf}_s(x_r) \bar{\mu}(x'_r)}\left(1 \wedge \frac{f(x') \bar{\apdf}_s(x_r)\bar{\mu}(x'_r)}{f(x) \bar{\apdf}_s(x'_r)\bar{\mu}(x_r)}\right) \\
&=p(x'_r,x_r)\beta(x'_r,x_r) \mu(x_\bot|x_r)  \pi(x') \frac{f(x)\bar{\mu}(x_r)\bar{\apdf}_s(x'_r)}{f(x')\bar{\apdf}_s(x_r)\bar{\mu}(x'_r)}  \left(1 \wedge \frac{f(x') \bar{\apdf}_s(x_r)\bar{\mu}(x'_r)}{f(x) \bar{\apdf}_s(x'_r)\bar{\mu}(x_r)}\right) \\
&=p(x'_r,x_r)\beta(x'_r,x_r) \mu(x_\bot|x_r)  \pi(x')  \left(1 \wedge \frac{f(x) \bar{\apdf}_s(x'_r)\bar{\mu}(x_r)}{f(x') \bar{\apdf}_s(x_r)\bar{\mu}(x'_r)}\right) \\
&=p(x'_r,x_r)\beta(x'_r,x_r) \mu(x_\bot|x_r) \pi(x')\alpha(x',x)=\pi(x')Q(x',x). 
\end{align*}
For the case $x_r=x_r', x_\bot\neq x_\bot'$, the overall transition density is
\[
Q(x,x')=\delta_{x_r=x'_r}\beta_c(x_r)\mu(x'_\bot|x_r) \alpha(x,x'),\quad \beta_c(x_r)=1-\int p(x_r,y_r)\beta(x_r,y_r)dy_r. 
\]
Therefore, as $x_r = x'_r$, we have
\begin{align*}
\pi(x)Q(x,x')&=\pi(x)\delta_{x_r=x_r'}\beta_c(x_r)\mu(x'_\bot|x_r) \alpha(x,x')\\
&=\delta_{x_r=x'_r}\beta_c(x_r)\frac{\bar{\apdf}_s(x'_r)\bar{\pi}(x_r)}{\bar{\apdf}_s(x_r)\bar{\pi}(x_r')} \frac{\pi(x)}{\pi(x')}\pi(x') \mu(x'_\bot|x_r)\alpha(x,x')\\
&=\delta_{x_r=x'_r}\beta_c(x_r)\frac{\bar{\apdf}_s(x'_r)\bar{\pi}(x_r)}{\bar{\apdf}_s(x_r)\bar{\pi}(x_r')} \frac{f(x) \mu(x_\bot|x_r) }{f(x')\mu(x'_\bot|x_r')}\pi(x') \mu(x'_\bot|x_r)\alpha(x,x')\\
&=\pi(x')\delta_{x_r=x'_r}\beta_c(x_r)\mu(x_\bot|x_r) \alpha(x',x)=\pi(x')Q(x',x).
\end{align*}
Note that if the proposal is rejected for the $x_\bot$ part, then the $x_r$ is also rejected. So the case that $x_r\neq x_r', x_\bot= x_\bot'$ can be ignored. Finally, the detailed balance condition is trivial if $x=x'$. In conclusion, the detailed balance condition holds, so $\pi$ is the invariant density of Algorithm \ref{alg:MCMC}.

Next, we investigate the acceptance rate of the complement transition. If we denote the MCMC transition probability for the $x_r$ part as
\[
P(x_r,x'_r)=p(x_r,x'_r)\beta(x_r,x'_r)+\beta_c(x_r)\delta_{x_r=x'_r}. 
\]
Note that the acceptance probability can also be written as
\[
 1 \wedge \frac{f(x')}{f(x)}\frac{\bar{\apdf}_s(x_r)\bar{\mu}(x_r')}{\bar{\apdf}_s(x'_r)\bar{\mu}(x_r)} =
  1 \wedge \frac{f(x')\mu(x')}{f(x)\mu(x)}\frac{\bar{\apdf}_s(x_r)\mu(x_r|x_\bot)}{\bar{\apdf}_s(x'_r)\mu(x'_r|x'_\bot)}=
  1\wedge \frac{\pi(x')\apdf_s(x)}{\pi(x)\apdf_s(x')}.
\]
Then, the acceptance rate is given by
\begin{align*}
\E\left[\alpha(X,X')\right]=&\int\int \pi(x)P(x_r,x'_r)\mu(x'_\bot|x'_r)\left(1\wedge \frac{\pi(x')\apdf_s(x)}{\pi(x)\apdf_s(x')} \right)dxdx'\\
=&\int \int P(x_r,x'_r) \mu(x'_\bot|x_r') \apdf_s(x)\left[ \frac{\pi(x)}{\apdf_s(x)}\wedge  \frac{\pi(x')}{\apdf_s(x')}\right]dxdx'
\end{align*}
Therefore if we denote the likelihood ratio $b(x)=\frac{\pi(x)}{\apdf_s(x)}$, then the average rejection probability is
\[
1-\E\left[ \alpha(X,X')\right]=\int\int P(x_r,x'_r)\apdf_s(x)\mu(x'_\bot|x'_r) \left[(1-b(x'))\vee (1-b(x))\right]dxdx'.
\]
To continue, we note that for any $b\geq 0$, $1-b\leq 2-2\sqrt{b}\leq |2-2\sqrt{b}|$, therefore
\[
(1-b(x))\vee (1-b(x'))\leq |2-2\sqrt{b(x)}|\vee |2-2\sqrt{b(x')}|\leq |2-2\sqrt{b(x)}|+|2-2\sqrt{b'(x)}|. 
\]
As a consequence, 
\begin{align*}
1-\E\left[\alpha(X,X')\right]&\leq \int\int P(x_r,x'_r)\apdf_s(x)\mu(x'_\bot|x'_r) |2-2\sqrt{b(x)}|dxdx'\\
&\quad+\int\int P(x_r,x'_r)\apdf_s(x)\mu(x'_\bot|x'_r) |2-2\sqrt{b(x')}|dxdx'\\
&=2\int \int P(x_r,x'_r)\apdf_s(x)\mu(x'_\bot|x'_r) |2-2\sqrt{b(x)}|dxdx'=4\int \apdf_s(x)|1-\sqrt{b(x)}|dx.
\end{align*}
Above, the first identity is obtained by observing  that $P(x_r,x'_r)\apdf_s(x)\mu(x'_\bot|x'_r)=P(x'_r,x_r)\apdf_s(x')\mu(x_\bot|x_r)$. The second identity is obtained by observing that 
\[
\int\int P(x_r,x'_r)\mu(x'_\bot|x'_r)dx'_r dx'_\bot=\int P(x_r,x'_r)\left(\int \mu(x'_\bot|x'_r)dx'_\bot\right) dx_r'=\int P(x_r,x'_r)dx'_r=1. 
\]
Then by the Cauchy--Schwarz inequality,
\begin{align*}
\int \apdf_s(x)|1-\sqrt{b(x)}|dx\leq\sqrt{\int \apdf_s(x)(1-\sqrt{b(x)})^2dx}\sqrt{\int \apdf_s(x)dx}\leq \sqrt{2}\Hell(\pi,\apdf_s).
\end{align*}
In summary, we have 
\(
\E\left[ \alpha(X,X')\right]\geq 1-4\sqrt{2}\Hell(\pi,\apdf_s). 
\)
\qed

\subsection{Proof of Proposition \ref{prop:SMC}}\label{proof:SMC}
Recall that 
\begin{equation*}
\pi_k(x)=\frac{1}{Z_k} f(x)^{\beta_k} \mu(x), \quad \pi_{k+1}(x)=\frac{1}{Z_{k+1}} f(x)^{\beta_{k+1}} \mu(x),
\end{equation*}
and
\begin{align*}
  V_{k+1}(H_1,\pi_{k}) & = \sum_{i,j=1}^d\text{var}_{X\sim \pi_{k}}\left[\partial_i\log f(X)\partial_j \log f(X)\frac{\pi_{k+1}(X)}{\pi_k(X)}\right]. 
\end{align*}
For all $x$ such that $\pi_k(x)>0$, we have
\begin{equation*}
  \frac{\pi_{k+1}(x)}{\pi_k(x)} = \frac{Z_k}{Z_{k+1}} f(x)^{(\beta_{k+1}-\beta_k)} = \frac{Z_k}{Z_{k+1}} f(x)^\delta,
\end{equation*}
where $\delta=\beta_{k+1}-\beta_k$. Thus, the variance $V_{k+1}(H_1,\pi_{k})$ satisfies
\begin{align}
  V_{k+1}(H_1,\pi_{k}) & = \sum_{i,j=1}^d\text{var}_{X\sim \pi_{k}}\left[ \big[H^{(k,k+1)}(X)\big]_{ij}\right] \nonumber \\
  & \leq \sum_{i,j=1}^d \E_{X\sim \pi_k}\left[ \big[H^{(k,k+1)}(X)\big]_{ij}^2\right] \nonumber \\
  & = \E_{X\sim \pi_k}\left[ \| H^{(k,k+1)}(X)\|_{F}^2\right]. \label{eq:bound_smc}
\end{align}
Since the square of the $i,j$-th component of $H^{(k,k+1)}(x)$ is given by 
\[
\big[H^{(k,k+1)}(x)\big]_{i,j}^2=
\left(\frac{Z_k}{Z_{k+1}}\partial_i \log f(x)\partial_j \log f(x) f(x)^\delta\right)^2,
\]
the variance upper bound in \eqref{eq:bound_smc} can also be expressed as
\begin{align*}
  \E_{X\sim \pi_k}\left[ \| H^{(k,k+1)}(X)\|_{F}^2\right] & = \frac{Z^2_k}{Z^2_{k+1}}\sum_{i,j=1}^d \E_{X\sim \pi_k}\left[(\partial_i \log f(X))^2(\partial_j \log f(X))^2 f(X)^{2\delta}\right] \\
  &=\frac{Z^2_k}{Z^2_{k+1}}\int \|\nabla \log f(x)\|^4 f(x)^{2\delta}\pi_k(x)dx\\ 
  &= \frac{Z_k}{Z^2_{k+1}}\int  \|\nabla \log f(x)\|^4f(x)^{2\delta}f(x)^{\beta_k}\mu(x)dx\\
  &= \frac{Z_k}{Z^2_{k+1}}\int  \|\nabla \log f(x)\|^4 f(x)^{\beta_{k+1}+\delta}\mu(x)dx.\\
  & = \frac{Z_k}{Z_{k+1}}\int  \|\nabla \log f(x)\|^4 f(x)^\delta\pi_{k+1}(x)dx.
\end{align*}
Thus, the variance upper bound is finite if $\|\nabla \log f(x)\|^4f(x)^{\beta_{k+1}+\delta}$ is integrable under $\mu$.
\qed

\section{Proofs in Section \ref{sec:dimension}}

\subsection{Proof of Proposition \ref{prop:linear}}
\label{proof:linear}
For the linear inverse problem,  the  likelihood function and its log gradient are given by 
\[
f(x)= \exp\left(-\frac12 \|Ax-y\|^2\right),\quad \nabla \log f(x)= A^\top (Ax-y). 
\]
The posterior distribution of $X$ is given by 
\[
\pi(x)\sim \mathcal{N}((C_A+I)^{-1}u, (C_A+I)^{-1}),\quad C_A=A^\top A,\quad u=A^\top y.
\]

\noindent\textbf{Proof of claim 1).} The $H_0$ matrix is given by 
\[
H_0=\E_{\mu} \left[ \nabla \log f(X)\nabla \log f(X)^\top\right] =A^\top (AA^\top +yy^\top )A=C_A^2+U,
\]
where  $U=A^\top  yy^\top A^\top $ is a rank 1 matrix. Apply Lemma \ref{lem:eigpert} with $H_0=C_A^2+U$, we have
\[
\lambda_{k+1}(H_0)\leq \lambda_{k}(C_A^2). 
\]
\noindent\textbf{Proof of claim 2).}  The $H_1$ matrix is given by 
\[
H_1 =\E_{\pi} \left[ \nabla \log f(X)\nabla \log f(X)^\top \right] = \E_{\pi} \left[ A^\top (AX-y) (AX-y)^\top A\right].
\]
For a given $y$ and $X\sim \pi$, the vector $A^\top(AX-y)$ follows a Gaussian distribution with the mean 
\[
A^\top \left( A (C_A+I)^{-1} u - y \right)  = C_A (C_A+I)^{-1} u - A^\top y = C_A (C_A+I)^{-1} u - u = - (C_A+I)^{-1} u
\]
and the covariance $C_A (C_A+I)^{-1} C_A$. This way, we have
\[
H_1 = C_A (C_A+I)^{-1} C_A + (C_A+I)^{-1} U (C_A+I)^{-1},
\]
where $(C_A+I)^{-1} U (C_A+I)^{-1}$ is a rank-$1$ matrix. Thus, by Lemma \ref{lem:eigpert}, we have
\[
\lambda_{k+1}(H_1)\leq \lambda_{k}(C_A(C_A+I)^{-1}C_A).  
\]
Finally note that eigenvectors of $C_A(I+C_A)^{-1}C_A$ are identical with the eigenvectors of $C_A$, with 
\[
v_k^\top C_A(I+C_A)^{-1}C_Av_k= \frac{\lambda_k(C_A)^2}{1+\lambda_k(C_A)}, 
\]
and we have our claim.

\noindent\textbf{Proof of claim 3).} 
Note that for a fixed $y$, the normalizing constant $Z$ is the integral
\begin{align*}
Z = \int (2\pi)^{-d/2}\exp\left( -\frac12 \|Ax - y\|^2 - \frac12 \|x\|^2\right) dx ,
\end{align*}
which is equivalent to 
\begin{align*}
(2\pi)^{-d_y/2} Z = \int (2\pi)^{-d/2 - d_y/2}\exp\left( -\frac12 \|Ax - y\|^2 - \frac12 \|x\|^2\right) dx ,
\end{align*}
where the integrand of the right hand side is the joint probability density of $x$ and $y$. In other words, if we view $(2\pi)^{-d_y/2}Z$ as a function of $y$, then the right-hand side of the equation above is the marginal probability density of $y$. Since the marginal of a Gaussian distribution is still Gaussian,  it is easy to see that $y$ follows the Gaussian distribution $\mathcal{N}(0, I + A A^\top)$. Thus, we have
\begin{equation*}%
Z = \det(I+AA^\top)^{-1/2} \exp\left(-\frac12 y^\top (I + A A^\top)^{-1} y\right).
\end{equation*}
Since $\det(I+AA^\top) = \det(I+A^\top A) = \det(I+C_A)$, we have
\[
\frac{1}{\sqrt{ \det(I+C_A) }}\geq Z \geq \frac{1}{\sqrt{ \det(I+C_A) }} \exp\left(-\frac12 \|y\|^2\right).
\]
Then, the result follows. 

\noindent\textbf{Proof of claim 4).} 
We have
\begin{align*}
\|f\|^2_{2,\mu} & = (2\pi)^{-d/2}\int \exp\left(- \|Ax - y\|^2 -\frac12 \|x\|^2\right) dx \\
& = (2\pi)^{d_y/2} \left[(2\pi)^{-d/2-d_y/2}\int \exp\left(- \frac12 \|\tilde{A}x - \tilde{y}\|^2 -\frac12 \|x\|^2\right) dx\right]
\end{align*}
with $\tilde{A} = \sqrt{2}A$ and $\tilde{y} = \sqrt{2} y$. Note that the term inside the square brackets is the normalizing constant of the posterior defined by the prior $\mathcal{N}(0, I_d)$, the parameter-to-observable map $\tilde{A}$, and the data $\tilde{y}$. This way, applying a similar identity to that in claim 3), we have
\begin{align*}
\|f\|^2_{2,\mu} & = \det(I+\tilde{A}\tilde{A}^\top)^{-1/2} \exp\left(-\frac12 \tilde{y}^\top (I + \tilde{A} \tilde{A}^\top)^{-1} \tilde{y}\right) \\
& = \frac{1}{\sqrt{\det(I+2C_A)} } \exp\left(- y^\top (I + 2 A A^\top)^{-1} y \right).
\end{align*}
This leads to 
\begin{equation}\label{eq:rationfz}
\frac{\|f\|_{2,\mu}}{Z} = \frac{\det(I+C_A)^{1/2}}{\det(I+2C_A)^{1/4}} \exp\left( \frac12 y^\top (I + A A^\top)^{-1} y - \frac12 y^\top (I + 2 A A^\top)^{-1} y \right).
\end{equation}
Suppose the eigenvalue decomposition of $AA^\top$ is given by $A A^\top=U\Lambda U^\top$, then 
\begin{align*}
(I + A A^\top)^{-1} - (I + 2 A A^\top)^{-1}&=U(I + \Lambda)^{-1} U^\top- U(I + 2 \Lambda)^{-1}U^\top\\
&=U\Lambda (I + \Lambda)^{-1}(I + 2 \Lambda)^{-1} U^\top.
\end{align*}
Therefore if $AA^\top$ has an eigenvector $u_i$ so that $AA^\top u_i=\lambda_i u_i$, then it is also an eigenvector of  $(I + A A^\top)^{-1} - (I + 2 A A^\top)^{-1}$:
\[
((I + A A^\top)^{-1} - (I + 2 A A^\top)^{-1})u_i=\frac{\lambda_i}{(1+\lambda_i)(1+2\lambda_i)}u_i. 
\]
Since $\frac{\lambda_i}{(1+\lambda_i)(1+2\lambda_i)}=\frac{1}{2\lambda_i+1/\lambda_i+3}\leq \frac{\sqrt{2}-1}{2(\sqrt{2}+1)}$,
we obtain the upper bound following \eqref{eq:rationfz}
\begin{align*}
\frac{\|f\|_{2,\mu}}{Z} & \leq \left( \prod_{k=1}^d  \frac{1 + 2\lambda_i(C_A) + \lambda_i(C_A)^2}{ 1 + 2\lambda_i(C_A)}\right)^{1/4} \exp\left( \frac{\sqrt{2}-1}{2 (\sqrt{2}+1)} \|y\|^2 \right)\\
& \leq \prod_{k=1}^d  \left( 1 + \lambda_i(C_A)^2\right)^{1/4} \exp\left( \frac{\sqrt{2}-1}{2 (\sqrt{2}+1)} \|y\|^2 \right) \\
& = \det(I + C_A^2)^{1/4} \exp\left( \frac{\sqrt{2}-1}{2 (\sqrt{2}+1)} \|y\|^2 \right) \\
& = \det(I + C_A^2)^{1/4} \exp\left( \frac12 (\sqrt{2}-1)^2  \|y\|^2 \right).
\end{align*}

\noindent\textbf{Proof of claim 5).} 
Recall that $\nabla \log f(x)=A^\top (Ax-y)$. We introduce the random variable $\zeta = A^\top AX- A^\top y$ that follows the Gaussian distribution $p(\zeta) = \mathcal{N}(-A^\top y, C_A^2)$.
Employing the upper bound established in Proposition \ref{prop:Vbound}, we have
\[
V(H_0,\mu)\leq \E_{\mu} \left[ \|\nabla \log f(X)\|^4 \right] = \E_{\zeta\sim p} \Big[ \Big(\sum_{i = 1}^d \zeta_i^2 \Big)^2\Big] = \E_{\zeta\sim p} \Big[ \sum_{i, j} \zeta_i^2 \zeta_j^2 \Big].
\]
Note that if $X\sim \mathcal{N}(m,\sigma^2)$, we can write it as $X=m+\sigma Z$ where $Z\sim \mathcal{N}(0,1)$. With this we find that 
\[
\E X^4=m^4+2\sigma^2m^2+3\sigma^4\leq 3(m^2+\sigma^2)^2=3(\E X^2)^2.
\]
For $i = j$, we have $\E [\zeta^4_i]=3\big(\E[\zeta^2_i]\big)^2$ because $\zeta_i$ is Gaussian distributed and all real-valued Gaussian random variables have zero excess kurtosis, and for $i \neq j$, we have
\(
\E [\zeta_i^2 \zeta_j^2] \leq \sqrt{ \E [\zeta_i^4] \E[\zeta_j^4] } = 3 \E [\zeta_i^2] \E[\zeta_j^2].
\)
This leads to
\begin{equation}
\label{rev:5}
\E_{\zeta\sim p} [ \|\zeta\|^4]=\E_{\zeta\sim p} \Big[ \sum_{i, j} \zeta_i^2 \zeta_j^2 \Big] \leq 3 \sum_{i, j} \E [\zeta_i^2] \E[\zeta_j^2] = 3 \Big( \sum_{i}^d \E [\zeta_i^2] \Big)^2.
\end{equation}
Since $\zeta \sim \mathcal{N}(-A^\top y, C_A^2)$, we have
\[
\sum_{i}^d \E [\zeta_i^2] = \tr(C_A^2) + \tr( A^\top y y^\top A ) = \sum_{i = 1}^d \lambda_i(C_A)^2 + \|A^\top y\|^2.
\]
Thus, we have
\[
V(H_0,\mu)\leq 3 \Big( \sum_{i = 1}^d \lambda_i(C_A)^2 + \|A^\top y\|^2 \Big)^2 \leq 6 \Big(\sum_{i = 1}^d \lambda_i(C_A)^2 \Big)^2 + 6 \|A^\top y\|^4.
\]
For the case $k=1$, employing the upper bound established in Proposition \ref{prop:Vbound}, we have
\begin{align*}
V(H_1,\mu)\leq \E_\mu\left[ \|\nabla \log f(X)\|^4 \frac{\pi(X)^2}{\mu(X)^2} \right] & =  \frac{1}{Z^2} \int \|\nabla \log f(x)\|^4 f(x)^2 \mu(x) dx.
\end{align*}
Define a new distribution with the density 
\begin{align}
\notag
\pi_2(x) = \frac{1}{Z_2} f(x)^2 \mu(x)& \propto\exp\left( - \|Ax - y\| - \frac12 \|x\|^2\right)\\
\label{rev:3}
&\propto \exp\left( -\frac12 x^\top (I+2A^\top A) x - y^\top A x\right).
\end{align}
Clearly $\pi_2$ is a Gaussian distribution, its mean is 
$2 (I + 2 C_A)^{-1} A^\top y$, its covariance is $ (I + 2 C_A)^{-1} $, and the normalizing constant $Z_2 = \|f\|_{2,\mu}^2$ satisfies
\[
\frac{Z_2}{Z^2}  \leq \prod_{k=1}^d  \left( 1 + \lambda_i(C_A)^2\right)^{1/2} \exp\left( \frac{\sqrt{2}-1}{\sqrt{2}+1} \|y\|^2 \right) = \sqrt{\det(I + C_A^2)} \exp\left( (\sqrt{2}-1)^2\|y\|^2 \right),
\]
as the result of claim 4). 
We express the upper bound on the variance as
\[
V(H_1,\mu) \leq \E_\mu\left[ \|\nabla \log f(X)\|^4 \frac{\pi(X)^2}{\mu(X)^2} \right] = \frac{Z_2}{Z^2} \E_{\pi_2}\left[ \|\nabla \log f(X)\|^4 \right] = \frac{Z_2}{Z^2} \E_{\pi_2}\left[ \|C_A X - A^\top y\|^4 \right].
\]
Similar to the proof of the first part, we can introduce $\zeta = C_A X - A^\top y$, for $X\sim \pi_2.$ Then $\zeta$ follows the Gaussian distribution with the mean
\[
2 C_A (I + 2 C_A)^{-1} A^\top y - A^\top y = - (I + 2 C_A)^{-1} A^\top y,
\]
and the covariance $C_A (I + 2 C_A)^{-1}  C_A$. This leads to
\begin{align}
\notag
\sum_{i}^d \E [\zeta_i^2] & = \tr(C_A (I + 2 C_A)^{-1}  C_A) + \|(I + 2 C_A)^{-1} A^\top y\|^2 \\
\notag
& = \sum_{i = 1}^d \frac{\lambda_i(C_A)^2}{1 + 2\lambda_i(C_A)} + \|(I + 2 C_A)^{-1} A^\top y\|^2 \\
\label{rev:4}
& \leq \sum_{i = 1}^d \frac{\lambda_i(C_A)^2}{1 + 2\lambda_i(C_A)} + \| A^\top y\|^2.
\end{align}
Thus, following a similar derivation to the case $k=0$, we have
\[
V(H_1,\mu)\leq  6 \sqrt{\det(I + C_A^2)}  \exp\left( (\sqrt{2}-1)^2 \|y\|^2 \right) \left( \left( \sum_{i = 1}^d \frac{\lambda_i(C_A)^2}{1 + 2\lambda_i(C_A)}\right)^2 + \| A^\top y\|^4 \right).
\]

\noindent\textbf{Proof of claim 6).} 
Consider the tempered target density
\[
\pi_\beta(x) = \frac{1}{Z_\beta} f(x)^\beta \mu(x) = \frac{1}{(2\pi)^{d/2}  Z_\beta} \exp\left( - \frac{1}{2} \|\sqrt{\beta}A x - \sqrt{\beta}y\|^2 - \frac{1}{2}\|x\|^2\right),
\]
where the normalizing constant takes the form
\begin{equation}
\label{rev:6}
Z_\beta = \det(I+\beta A A^\top)^{-1/2} \exp\left(-\frac{\beta}2 y^\top (I + \beta A A^\top)^{-1} y\right).
\end{equation}
We use the shorthand notations $\pi_k$ and $Z_k$ to denote $\pi_{\beta_k}$ and $Z_{\beta_k}$, respectively. Let $\delta = \beta_{k+1}-\beta_k$, we have
\begin{align*}
V_{k+1}(H_1,\pi_k)&\leq \frac{Z_{k}^2}{Z_{k+1}^2}\E_{\pi_k}\left[  \|\nabla \log f(X)\|^4 f(X)^{2\delta}\right] \\
& = \frac{Z_{k}}{Z_{k+1}^2}\E_{\mu}\left[  \|\nabla \log f(X)\|^4 f(X)^{2\delta + \beta_k}\right].
\end{align*}
Following a similar procedure in the proof of claim 5), we define a new distribution with the density 
\[
\pi_\tau(x) = \frac{1}{ Z_\tau} f(x)^{\tau} \mu(x) dx,
\]
where $\tau = 2\delta + \beta_k = 2\beta_{k+1} - \beta_k = \beta_{k+1} + \delta$. The density $\pi_\tau(x)$ has the mean $\tau (I + \tau C_A)^{-1} A^\top y$ and the covariance $ (I + \tau C_A)^{-1} $.
We express the upper bound on the variance as
\begin{align}
V_{k+1}(H_1,\mu) & \leq \frac{Z_{k}}{Z_{k+1}^2}\E_{\mu}\left[  \|\nabla \log f(X)\|^4 f(X)^\tau\right] \nonumber \\
& = \frac{Z_\tau Z_{k}}{Z_{k+1}^2} \E_{\pi_\tau}\left[ \|\nabla \log f(X)\|^4 \right] \nonumber \\
& = \frac{Z_\tau Z_{k}}{Z_{k+1}^2} \E_{\pi_\tau}\left[ \|C_A X - A^\top y\|^4 \right].\label{eq:claim6}
\end{align}
Then, for $X\sim \pi_\tau$, we introduce $\zeta = C_A X - A^\top y$, which follows the Gaussian distribution with the mean
\[
\E [\zeta]=(\tau C_A (I + \tau C_A)^{-1}-I) A^\top y=- (I + \tau C_A)^{-1} A^\top y,
\]
and the covariance $C_A (I + \tau C_A)^{-1}  C_A$. Similar to the derivation of \eqref{rev:4}, this leads to
\begin{align*}
\sum_{i=1}^d \E [\zeta_i^2]=\E [\|\zeta\|^2] 
&=\|\E[\zeta]\|^2+\tr(\text{cov}[\zeta])\\
& = \tr(C_A (I + \tau C_A)^{-1}  C_A) + \|(I + \tau C_A)^{-1} A^\top y\|^2 \\
& \leq \sum_{i = 1}^d \frac{\lambda_i(C_A)^2}{1 + \tau \lambda_i(C_A)} + \| A^\top y\|^2,
\end{align*}
which yields the following using the same argument used in \eqref{rev:5}:
\begin{equation}
\E_{\pi_\tau}\left[ \|C_A X - A^\top y\|^4 \right]=\E_{\pi_\tau}\left[ \|\zeta\|^4 \right]\leq 6 \left( \Big(\sum_{i = 1}^d \frac{\lambda_i(C_A)^2}{1 + \tau \lambda_i(C_A)}\Big)^2 + \| A^\top y\|^4 \right).\label{eq:claim60}
\end{equation}
Recalling \eqref{rev:6}, the ratio between normalizing constants in \eqref{eq:claim6} can be expressed as
\begin{align*}
\frac{Z_\tau Z_{k}}{Z_{k+1}^2} & = \frac{\det(I+\beta_{k+1} A A^\top) }{\sqrt{ \det(I+\beta_{k} A A^\top) \det(I+\tau A A^\top) } } \exp\left(\frac12 y^\top T y \right) ,
\end{align*}
where
\[
T = 2 \beta_{k+1} (I + \beta_{k+1} A A^\top)^{-1} - \beta_k (I + \beta_{k} A A^\top)^{-1} - \tau (I + \tau A A^\top)^{-1}.
\]
In the above equation, the ratio between determinants can be expressed as
\begin{align*}
\frac{\det(I+\beta_{k+1} A A^\top) }{\sqrt{ \det(I+\beta_{k} A A^\top) \det(I+\tau A A^\top) } } & = \left( \prod_{i = 1}^d \frac{1 + 2\beta_{k+1}\lambda_i(C_A) + \beta_{k+1}^2\lambda_i(C_A)^2}{1 + 2\beta_{k+1}\lambda_i(C_A) + \beta_k \tau \lambda_i(C_A)^2 } \right)^{1/2} .
\end{align*}
Since $\beta_{k+1}^2 - (\beta_{k+1} - \beta_k)^2 =  \beta_k (2\beta_{k+1} - \beta_{k})$, which gives  $\beta_{k+1}^2 = \delta^2 + \beta_k \tau $, and thus
\begin{align}
\frac{\det(I+\beta_{k+1} A A^\top) }{\sqrt{ \det(I+\beta_{k} A A^\top) \det(I+\tau A A^\top) } } & =  \prod_{i = 1}^d \left( 1 + \frac{\delta^2 \lambda_i(C_A)^2}{1 + 2\beta_{k+1}\lambda_i(C_A) + \beta_k \tau \lambda_i(C_A)^2 }\right) ^{1/2} \nonumber \\
& \leq \prod_{i = 1}^d \left( 1 + \delta^2 \lambda_i(C_A)^2\right) ^{1/2}.\label{eq:claim61}
\end{align}
Since the matrix $T$ and the matrix $AA^\top$ share the same eigenvectors, the eigenvalues of $T$ can be expressed as
\begin{align*}
\lambda_i(T) & = \frac{\beta_{k} }{1 + \beta_{k+1} \lambda_i(AA^\top)} - \frac{\beta_{k} }{1 + \beta_{k} \lambda_i(AA^\top)} + \frac{\tau }{1 + \beta_{k+1} \lambda_i(AA^\top)} - \frac{\tau }{1 + \tau \lambda_i(AA^\top)} \\
& = \frac{ - \beta_{k}\delta \lambda_i(AA^\top) }{(1 + \beta_{k+1} \lambda_i(AA^\top))(1 + \beta_{k} \lambda_i(AA^\top))} + \frac{\tau\delta \lambda_i(AA^\top)}{(1 + \beta_{k+1} \lambda_i(AA^\top))(1 + \tau \lambda_i(AA^\top))}\\
& = \frac{\delta\lambda_i(AA^\top)}{1 + \beta_{k+1} \lambda_i(AA^\top)}\left( \frac{\tau }{1 + \tau \lambda_i(AA^\top)} -  \frac{\beta_{k} }{1 + \beta_{k} \lambda_i(AA^\top)} \right)\\
& = \frac{2 \delta ^2 \lambda_i(AA^\top)}{(1 + \beta_{k+1} \lambda_i(AA^\top))(1 + \tau \lambda_i(AA^\top))(1 + \beta_{k} \lambda_i(AA^\top))} .
\end{align*}
This way, we have $\lambda_i(T)\leq 2 \delta ^2 \lambda_i(AA^\top)$, and thus $T \preceq 2 \delta ^2 AA^\top$. This leads to 
\[
\exp\left(\frac12 y^\top T y \right) \leq \exp\left(\delta ^2 \|A^\top y\|^2 \right) .
\]
Substituting the above inequality, \eqref{eq:claim60}, and \eqref{eq:claim61} into \eqref{eq:claim6}, we have
\[
V_{k+1}(H_1,\mu) \leq 6 \sqrt{\det(I + \delta^2 C_A^2)} \exp\left(\delta^2 \|A^\top y\|^2\right) \left( \Big(\sum_{i = 1}^d \frac{\lambda_i(C_A)^2}{1 + \tau \lambda_i(C_A)}\Big)^2 + \| A^\top y\|^4 \right).
\]
\qed

\subsection{Proof of Corollary \ref{cor:lin}}
\label{proof:lin}
We will first show that the eigenvalues of $C_A$ are bounded by those of $\Gamma$. 
By the Courant--Fischer--Weyl min-max principle, we note that for any symmetric matrix $C$
\[
\lambda_r(C)=\min_V\{\max_{x\in V}\{x^\top Cx, \|x\|=1\}, \text{dim}(V)=d-r+1\}.
\]
Let the eigenvectors of $\Gamma$ be $v_1,\ldots, v_d$. Now we pick 
\(
V=\text{span}\{v_{r},\ldots,v_d\}.
\)
Then for any $x\in V$ of unit $\ell^2$ norm, 
\[
x^\top C_Ax=x^\top  \Gamma^{1/2} G^\top G\Gamma^{1/2}x=\|G\Gamma^{1/2}x\|^2\leq \|G\|^2 \|\Gamma^{1/2}x\|^2\leq \|G\|^2\lambda_r(\Gamma). 
\]
Therefore, $\lambda_r(C_A)\leq \|G\|^2 \lambda_r(\Gamma)$. 

The following identities are useful: For $r\geq 2$, we have
\[
 \sum_{j=r}^d j^{-\alpha}\leq \int^{\infty}_{r-1} x^{-\alpha} dx\leq \frac{(r-1)^{1 - \alpha}}{\alpha-1},
\]
and for $r=1$, we have
\[
1+ \sum_{j=2}^d j^{-\alpha}\leq \frac{\alpha}{\alpha-1}.
\]
For any $a>0$, we have
\begin{align*}
\prod_{j=r}^d(1+a\lambda_j(C_A))&\leq\prod_{j=r}^d(1+a\|G\|^2\lambda_j(\Gamma))\\
&\leq \exp\left(a\|G\|^2\sum_{j=r}^d\lambda_j(\Gamma)\right)=\exp\left(a\|G\|^2C_\Gamma \sum_{j=r}^d j^{-\alpha} \right).
\end{align*}
Thus, we have
\begin{equation}\label{eq:coro61}
\sum_{i = 1}^d \lambda_i(C_A)^2 \leq  \|G\|^4 C_\Gamma^2 \frac{2 \alpha}{ 2 \alpha - 1},
\end{equation}
and
\begin{equation}\label{eq:coro62}
\det(I + a C_A^2) \leq \exp\left(a \|G\|^4 C_\Gamma^2 \frac{2 \alpha}{ 2 \alpha - 1} \right).
\end{equation}
Then, replacing the estimates in Proposition \ref{prop:linear} with these upper bounds, the results follow. \qed

\section*{Acknowledgments}
X. Tong's research is supported by MOE Academic Research Funds R-146-000-292-114. T. Cui acknowledges support from the Australian Research Council.

\bibliographystyle{imsart-nameyear}
\bibliography{ref}

\begin{thebibliography}{65}

\bibitem[\protect\citeauthoryear{Agapiou, Dashti and
  Helin}{2018}]{agapiou2018rates}
\begin{barticle}[author]
\bauthor{\bsnm{Agapiou},~\bfnm{Sergios}\binits{S.}},
  \bauthor{\bsnm{Dashti},~\bfnm{Masoumeh}\binits{M.}} \AND
  \bauthor{\bsnm{Helin},~\bfnm{Tapio}\binits{T.}}
(\byear{2018}).
\btitle{Rates of contraction of posterior distributions based on $ p
  $-exponential priors}.
\bjournal{arXiv preprint arXiv:1811.12244}.
\end{barticle}
\endbibitem

\bibitem[\protect\citeauthoryear{Agapiou et~al.}{2017}]{agapiou2017importance}
\begin{barticle}[author]
\bauthor{\bsnm{Agapiou},~\bfnm{Sergios}\binits{S.}},
  \bauthor{\bsnm{Papaspiliopoulos},~\bfnm{Omiros}\binits{O.}},
  \bauthor{\bsnm{Sanz-Alonso},~\bfnm{Daniel}\binits{D.}},
  \bauthor{\bsnm{Stuart},~\bfnm{AM}\binits{A.}} \betal{et~al.}
(\byear{2017}).
\btitle{Importance sampling: Intrinsic dimension and computational cost}.
\bjournal{Statistical Science}
\bvolume{32}
\bpages{405--431}.
\end{barticle}
\endbibitem

\bibitem[\protect\citeauthoryear{Agapiou et~al.}{2018}]{agapiou2018unbiased}
\begin{barticle}[author]
\bauthor{\bsnm{Agapiou},~\bfnm{Sergios}\binits{S.}},
  \bauthor{\bsnm{Roberts},~\bfnm{Gareth~O}\binits{G.~O.}},
  \bauthor{\bsnm{Vollmer},~\bfnm{Sebastian~J}\binits{S.~J.}} \betal{et~al.}
(\byear{2018}).
\btitle{Unbiased Monte Carlo: Posterior estimation for
  intractable/infinite-dimensional models}.
\bjournal{Bernoulli}
\bvolume{24}
\bpages{1726--1786}.
\end{barticle}
\endbibitem

\bibitem[\protect\citeauthoryear{Andrieu and Roberts}{2009}]{andrieu2009pseudo}
\begin{barticle}[author]
\bauthor{\bsnm{Andrieu},~\bfnm{Christophe}\binits{C.}} \AND
  \bauthor{\bsnm{Roberts},~\bfnm{Gareth~O}\binits{G.~O.}}
(\byear{2009}).
\btitle{The pseudo-marginal approach for efficient Monte Carlo computations}.
\bjournal{The Annals of Statistics}
\bvolume{37}
\bpages{697--725}.
\end{barticle}
\endbibitem

\bibitem[\protect\citeauthoryear{Andrieu and Vihola}{2015}]{AV15}
\begin{barticle}[author]
\bauthor{\bsnm{Andrieu},~\bfnm{Christophe}\binits{C.}} \AND
  \bauthor{\bsnm{Vihola},~\bfnm{Matti}\binits{M.}}
(\byear{2015}).
\btitle{Convergence properties of pseudo-marginal {M}arkov chian {M}onte
  {C}arlo algorithms}.
\bjournal{Ann. Appl. Probab.}
\bvolume{25}
\bpages{1030-1077}.
\end{barticle}
\endbibitem

\bibitem[\protect\citeauthoryear{Beskos et~al.}{2014}]{beskos2014stability}
\begin{barticle}[author]
\bauthor{\bsnm{Beskos},~\bfnm{Alexandros}\binits{A.}},
  \bauthor{\bsnm{Crisan},~\bfnm{Dan}\binits{D.}},
  \bauthor{\bsnm{Jasra},~\bfnm{Ajay}\binits{A.}} \betal{et~al.}
(\byear{2014}).
\btitle{On the stability of sequential Monte Carlo methods in high dimensions}.
\bjournal{The Annals of Applied Probability}
\bvolume{24}
\bpages{1396--1445}.
\end{barticle}
\endbibitem

\bibitem[\protect\citeauthoryear{Beskos et~al.}{2017}]{beskos2017geometric}
\begin{barticle}[author]
\bauthor{\bsnm{Beskos},~\bfnm{Alexandros}\binits{A.}},
  \bauthor{\bsnm{Girolami},~\bfnm{Mark}\binits{M.}},
  \bauthor{\bsnm{Lan},~\bfnm{Shiwei}\binits{S.}},
  \bauthor{\bsnm{Farrell},~\bfnm{Patrick~E}\binits{P.~E.}} \AND
  \bauthor{\bsnm{Stuart},~\bfnm{Andrew~M}\binits{A.~M.}}
(\byear{2017}).
\btitle{Geometric MCMC for infinite-dimensional inverse problems}.
\bjournal{Journal of Computational Physics}
\bvolume{335}
\bpages{327--351}.
\end{barticle}
\endbibitem

\bibitem[\protect\citeauthoryear{Beskos et~al.}{2018}]{beskos2018multilevel}
\begin{barticle}[author]
\bauthor{\bsnm{Beskos},~\bfnm{Alexandros}\binits{A.}},
  \bauthor{\bsnm{Jasra},~\bfnm{Ajay}\binits{A.}},
  \bauthor{\bsnm{Law},~\bfnm{Kody}\binits{K.}},
  \bauthor{\bsnm{Marzouk},~\bfnm{Youssef}\binits{Y.}} \AND
  \bauthor{\bsnm{Zhou},~\bfnm{Yan}\binits{Y.}}
(\byear{2018}).
\btitle{Multilevel sequential Monte Carlo with dimension-independent
  likelihood-informed proposals}.
\bjournal{SIAM/ASA Journal on Uncertainty Quantification}
\bvolume{6}
\bpages{762--786}.
\end{barticle}
\endbibitem

\bibitem[\protect\citeauthoryear{Bigoni et~al.}{2019}]{bigoni2019greedy}
\begin{barticle}[author]
\bauthor{\bsnm{Bigoni},~\bfnm{Daniele}\binits{D.}},
  \bauthor{\bsnm{Zahm},~\bfnm{Olivier}\binits{O.}},
  \bauthor{\bsnm{Spantini},~\bfnm{Alessio}\binits{A.}} \AND
  \bauthor{\bsnm{Marzouk},~\bfnm{Youssef}\binits{Y.}}
(\byear{2019}).
\btitle{Greedy inference with layers of lazy maps}.
\bjournal{arXiv preprint arXiv:1906.00031}.
\end{barticle}
\endbibitem

\bibitem[\protect\citeauthoryear{Bobkov}{1999}]{bobkov1999isoperimetric}
\begin{barticle}[author]
\bauthor{\bsnm{Bobkov},~\bfnm{Sergey~G}\binits{S.~G.}}
(\byear{1999}).
\btitle{Isoperimetric and analytic inequalities for log-concave probability
  measures}.
\bjournal{The Annals of Probability}
\bvolume{27}
\bpages{1903--1921}.
\end{barticle}
\endbibitem

\bibitem[\protect\citeauthoryear{Bobkov and Ledoux}{1997}]{bobkov1997poincare}
\begin{barticle}[author]
\bauthor{\bsnm{Bobkov},~\bfnm{Sergey}\binits{S.}} \AND
  \bauthor{\bsnm{Ledoux},~\bfnm{Michel}\binits{M.}}
(\byear{1997}).
\btitle{Poincar{\'e}'s inequalities and Talagrand's concentration phenomenon
  for the exponential distribution}.
\bjournal{Probability Theory and Related Fields}
\bvolume{107}
\bpages{383--400}.
\end{barticle}
\endbibitem

\bibitem[\protect\citeauthoryear{Bobkov and Ledoux}{2000}]{bobkov2000brunn}
\begin{barticle}[author]
\bauthor{\bsnm{Bobkov},~\bfnm{Sergey~G}\binits{S.~G.}} \AND
  \bauthor{\bsnm{Ledoux},~\bfnm{Michel}\binits{M.}}
(\byear{2000}).
\btitle{From Brunn-Minkowski to Brascamp-Lieb and to logarithmic sobolev
  inequalities}.
\bjournal{Geometric \& Functional Analysis GAFA}
\bvolume{10}
\bpages{1028--1052}.
\end{barticle}
\endbibitem

\bibitem[\protect\citeauthoryear{Bobkov and Ledoux}{2009}]{bobkov2009weighted}
\begin{barticle}[author]
\bauthor{\bsnm{Bobkov},~\bfnm{Sergey~G}\binits{S.~G.}} \AND
  \bauthor{\bsnm{Ledoux},~\bfnm{Michel}\binits{M.}}
(\byear{2009}).
\btitle{Weighted Poincar{\'e}-type inequalities for Cauchy and other convex
  measures}.
\bjournal{The Annals of Probability}
\bvolume{37}
\bpages{403--427}.
\end{barticle}
\endbibitem

\bibitem[\protect\citeauthoryear{Bolley and Villani}{2005}]{bolley2005weighted}
\begin{binproceedings}[author]
\bauthor{\bsnm{Bolley},~\bfnm{Fran{\c{c}}ois}\binits{F.}} \AND
  \bauthor{\bsnm{Villani},~\bfnm{C{\'e}dric}\binits{C.}}
(\byear{2005}).
\btitle{Weighted Csisz{\'a}r-Kullback-Pinsker inequalities and applications to
  transportation inequalities}.
In \bbooktitle{Annales de la Facult{\'e} des sciences de Toulouse:
  Math{\'e}matiques}
\bvolume{14}
\bpages{331--352}.
\end{binproceedings}
\endbibitem

\bibitem[\protect\citeauthoryear{Brascamp and Lieb}{1976}]{brascamp1976}
\begin{barticle}[author]
\bauthor{\bsnm{Brascamp},~\bfnm{Herm~Jan}\binits{H.~J.}} \AND
  \bauthor{\bsnm{Lieb},~\bfnm{Elliott~H}\binits{E.~H.}}
(\byear{1976}).
\btitle{On extensions of the Brunn-Minkowski and Pr{\'e}kopa-Leindler theorems,
  including inequalities for log concave functions, and with an application to
  the diffusion equation}.
\bjournal{Journal of Functional Analysis}
\bvolume{22}
\bpages{366--389}.
\end{barticle}
\endbibitem

\bibitem[\protect\citeauthoryear{Bui-Thanh et~al.}{2012}]{bui2012extreme}
\begin{binproceedings}[author]
\bauthor{\bsnm{Bui-Thanh},~\bfnm{Tan}\binits{T.}},
  \bauthor{\bsnm{Burstedde},~\bfnm{Carsten}\binits{C.}},
  \bauthor{\bsnm{Ghattas},~\bfnm{Omar}\binits{O.}},
  \bauthor{\bsnm{Martin},~\bfnm{James}\binits{J.}},
  \bauthor{\bsnm{Stadler},~\bfnm{Georg}\binits{G.}} \AND
  \bauthor{\bsnm{Wilcox},~\bfnm{Lucas~C}\binits{L.~C.}}
(\byear{2012}).
\btitle{Extreme-scale UQ for Bayesian inverse problems governed by PDEs}.
In \bbooktitle{SC'12: Proceedings of the International Conference on High
  Performance Computing, Networking, Storage and Analysis}
\bpages{1--11}.
\bpublisher{IEEE}.
\end{binproceedings}
\endbibitem

\bibitem[\protect\citeauthoryear{Bui-Thanh et~al.}{2013}]{bui2013computational}
\begin{barticle}[author]
\bauthor{\bsnm{Bui-Thanh},~\bfnm{Tan}\binits{T.}},
  \bauthor{\bsnm{Ghattas},~\bfnm{Omar}\binits{O.}},
  \bauthor{\bsnm{Martin},~\bfnm{James}\binits{J.}} \AND
  \bauthor{\bsnm{Stadler},~\bfnm{Georg}\binits{G.}}
(\byear{2013}).
\btitle{A computational framework for infinite-dimensional Bayesian inverse
  problems Part I: The linearized case, with application to global seismic
  inversion}.
\bjournal{SIAM Journal on Scientific Computing}
\bvolume{35}
\bpages{A2494--A2523}.
\end{barticle}
\endbibitem

\bibitem[\protect\citeauthoryear{Cai and Hall}{2008}]{CH08}
\begin{barticle}[author]
\bauthor{\bsnm{Cai},~\bfnm{T.}\binits{T.}} \AND
  \bauthor{\bsnm{Hall},~\bfnm{P.}\binits{P.}}
(\byear{2008}).
\btitle{Prediction in function linear regression}.
\bjournal{Ann. Statist.}
\bvolume{34}
\bpages{2159-2179}.
\end{barticle}
\endbibitem

\bibitem[\protect\citeauthoryear{Constantine, Kent and Bui-Thanh}{2016}]{CKB16}
\begin{barticle}[author]
\bauthor{\bsnm{Constantine},~\bfnm{Paul~G}\binits{P.~G.}},
  \bauthor{\bsnm{Kent},~\bfnm{Carson}\binits{C.}} \AND
  \bauthor{\bsnm{Bui-Thanh},~\bfnm{Tan}\binits{T.}}
(\byear{2016}).
\btitle{Accelerating Markov chain Monte Carlo with active subspaces}.
\bjournal{SIAM Journal on Scientific Computing}
\bvolume{38}
\bpages{A2779--A2805}.
\end{barticle}
\endbibitem

\bibitem[\protect\citeauthoryear{Cotter et~al.}{2013}]{cotter2013mcmc}
\begin{barticle}[author]
\bauthor{\bsnm{Cotter},~\bfnm{Simon~L}\binits{S.~L.}},
  \bauthor{\bsnm{Roberts},~\bfnm{Gareth~O}\binits{G.~O.}},
  \bauthor{\bsnm{Stuart},~\bfnm{Andrew~M}\binits{A.~M.}} \AND
  \bauthor{\bsnm{White},~\bfnm{David}\binits{D.}}
(\byear{2013}).
\btitle{MCMC methods for functions: modifying old algorithms to make them
  faster}.
\bjournal{Statistical Science}
\bpages{424--446}.
\end{barticle}
\endbibitem

\bibitem[\protect\citeauthoryear{Cui and Dolgov}{2020}]{cui2020deep}
\begin{barticle}[author]
\bauthor{\bsnm{Cui},~\bfnm{Tiangang}\binits{T.}} \AND
  \bauthor{\bsnm{Dolgov},~\bfnm{Sergey}\binits{S.}}
(\byear{2020}).
\btitle{Deep Composition of Tensor Trains using Squared Inverse Rosenblatt
  Transports}.
\bjournal{arXiv preprint arXiv:2007.06968}.
\end{barticle}
\endbibitem

\bibitem[\protect\citeauthoryear{Cui, Fox and
  O'Sullivan}{2011}]{cui2011bayesian}
\begin{barticle}[author]
\bauthor{\bsnm{Cui},~\bfnm{Tiangang}\binits{T.}},
  \bauthor{\bsnm{Fox},~\bfnm{Colin}\binits{C.}} \AND
  \bauthor{\bsnm{O'Sullivan},~\bfnm{Michael~J.}\binits{M.~J.}}
(\byear{2011}).
\btitle{Bayesian calibration of a large-scale geothermal reservoir model by a
  new adaptive delayed acceptance Metropolis Hastings algorithm}.
\bjournal{Water Resources Research}
\bvolume{47}.
\end{barticle}
\endbibitem

\bibitem[\protect\citeauthoryear{Cui, Law and Marzouk}{2016}]{Cuietal16b}
\begin{barticle}[author]
\bauthor{\bsnm{Cui},~\bfnm{Tiangang}\binits{T.}},
  \bauthor{\bsnm{Law},~\bfnm{Kody J.~H.}\binits{K.~J.~H.}} \AND
  \bauthor{\bsnm{Marzouk},~\bfnm{Youssef~M.}\binits{Y.~M.}}
(\byear{2016}).
\btitle{Dimension-independent likelihood-informed {MCMC}}.
\bjournal{J. Comput. Phys.}
\bvolume{304}
\bpages{109-137}.
\end{barticle}
\endbibitem

\bibitem[\protect\citeauthoryear{Cui and Zahm}{2020}]{cui2020data}
\begin{barticle}[author]
\bauthor{\bsnm{Cui},~\bfnm{Tiangang}\binits{T.}} \AND
  \bauthor{\bsnm{Zahm},~\bfnm{Olivier}\binits{O.}}
(\byear{2020}).
\btitle{Data-Free Likelihood-Informed Dimension Reduction of Bayesian Inverse
  Problems}.
\bjournal{hal preprint: hal-02938064}.
\end{barticle}
\endbibitem

\bibitem[\protect\citeauthoryear{Cui et~al.}{2014}]{cui2014likelihood}
\begin{barticle}[author]
\bauthor{\bsnm{Cui},~\bfnm{Tiangang}\binits{T.}},
  \bauthor{\bsnm{Martin},~\bfnm{James}\binits{J.}},
  \bauthor{\bsnm{Marzouk},~\bfnm{Youssef~M}\binits{Y.~M.}},
  \bauthor{\bsnm{Solonen},~\bfnm{Antti}\binits{A.}} \AND
  \bauthor{\bsnm{Spantini},~\bfnm{Alessio}\binits{A.}}
(\byear{2014}).
\btitle{Likelihood-informed dimension reduction for nonlinear inverse
  problems}.
\bjournal{Inverse Problems}
\bvolume{30}
\bpages{114015}.
\end{barticle}
\endbibitem

\bibitem[\protect\citeauthoryear{Dashti and
  Stuart}{2011}]{dashti2011uncertainty}
\begin{barticle}[author]
\bauthor{\bsnm{Dashti},~\bfnm{Masoumeh}\binits{M.}} \AND
  \bauthor{\bsnm{Stuart},~\bfnm{Andrew~M}\binits{A.~M.}}
(\byear{2011}).
\btitle{Uncertainty quantification and weak approximation of an elliptic
  inverse problem}.
\bjournal{SIAM Journal on Numerical Analysis}
\bvolume{49}
\bpages{2524--2542}.
\end{barticle}
\endbibitem

\bibitem[\protect\citeauthoryear{Detommaso et~al.}{2018}]{detommaso2018stein}
\begin{barticle}[author]
\bauthor{\bsnm{Detommaso},~\bfnm{Gianluca}\binits{G.}},
  \bauthor{\bsnm{Cui},~\bfnm{Tiangang}\binits{T.}},
  \bauthor{\bsnm{Marzouk},~\bfnm{Youssef}\binits{Y.}},
  \bauthor{\bsnm{Spantini},~\bfnm{Alessio}\binits{A.}} \AND
  \bauthor{\bsnm{Scheichl},~\bfnm{Robert}\binits{R.}}
(\byear{2018}).
\btitle{A Stein variational Newton method}.
\bjournal{Advances in Neural Information Processing Systems}
\bvolume{31}
\bpages{9169--9179}.
\end{barticle}
\endbibitem

\bibitem[\protect\citeauthoryear{Dodwell et~al.}{2019}]{dodwell2019multilevel}
\begin{barticle}[author]
\bauthor{\bsnm{Dodwell},~\bfnm{Tim~J}\binits{T.~J.}},
  \bauthor{\bsnm{Ketelsen},~\bfnm{Christian}\binits{C.}},
  \bauthor{\bsnm{Scheichl},~\bfnm{Robert}\binits{R.}} \AND
  \bauthor{\bsnm{Teckentrup},~\bfnm{Aretha~L}\binits{A.~L.}}
(\byear{2019}).
\btitle{Multilevel markov chain monte carlo}.
\bjournal{Siam Review}
\bvolume{61}
\bpages{509--545}.
\end{barticle}
\endbibitem

\bibitem[\protect\citeauthoryear{Drineas and Ipsen}{2019}]{drineas2019low}
\begin{barticle}[author]
\bauthor{\bsnm{Drineas},~\bfnm{Petros}\binits{P.}} \AND
  \bauthor{\bsnm{Ipsen},~\bfnm{Ilse~CF}\binits{I.~C.}}
(\byear{2019}).
\btitle{Low-rank matrix approximations do not need a singular value gap}.
\bjournal{SIAM Journal on Matrix Analysis and Applications}
\bvolume{40}
\bpages{299--319}.
\end{barticle}
\endbibitem

\bibitem[\protect\citeauthoryear{Flath et~al.}{2011}]{FlathEtAl11}
\begin{barticle}[author]
\bauthor{\bsnm{Flath},~\bfnm{H.~P.}\binits{H.~P.}},
  \bauthor{\bsnm{Wilcox},~\bfnm{L.~C.}\binits{L.~C.}},
  \bauthor{\bsnm{Ak\c{c}elik},~\bfnm{V.}\binits{V.}},
  \bauthor{\bsnm{Hill},~\bfnm{J.}\binits{J.}}, \bauthor{\bparticle{van}
  \bsnm{Bloemen~Waander},~\bfnm{B.}\binits{B.}} \AND
  \bauthor{\bsnm{Ghattas},~\bfnm{O.}\binits{O.}}
(\byear{2011}).
\btitle{Fast Algorithms for {B}ayesian Uncertainty Quantification in
  Large-Scale Linear Inverse Problems Based on Low-Rank Partial {H}essian
  Approximations}.
\bjournal{SIAM J. Sci. Comput.}
\bvolume{33}
\bpages{407--432}.
\end{barticle}
\endbibitem

\bibitem[\protect\citeauthoryear{Gross}{1975}]{gross1975logarithmic}
\begin{barticle}[author]
\bauthor{\bsnm{Gross},~\bfnm{Leonard}\binits{L.}}
(\byear{1975}).
\btitle{Logarithmic sobolev inequalities}.
\bjournal{American Journal of Mathematics}
\bvolume{97}
\bpages{1061--1083}.
\end{barticle}
\endbibitem

\bibitem[\protect\citeauthoryear{Haario et~al.}{2004}]{haario2004markov}
\begin{barticle}[author]
\bauthor{\bsnm{Haario},~\bfnm{Heikki}\binits{H.}},
  \bauthor{\bsnm{Laine},~\bfnm{Marko}\binits{M.}},
  \bauthor{\bsnm{Lehtinen},~\bfnm{Markku}\binits{M.}},
  \bauthor{\bsnm{Saksman},~\bfnm{Eero}\binits{E.}} \AND
  \bauthor{\bsnm{Tamminen},~\bfnm{Johanna}\binits{J.}}
(\byear{2004}).
\btitle{Markov chain Monte Carlo methods for high dimensional inversion in
  remote sensing}.
\bjournal{Journal of the Royal Statistical Society: series B (statistical
  methodology)}
\bvolume{66}
\bpages{591--607}.
\end{barticle}
\endbibitem

\bibitem[\protect\citeauthoryear{Hall and Horowitz}{2007}]{HH07}
\begin{barticle}[author]
\bauthor{\bsnm{Hall},~\bfnm{P.}\binits{P.}} \AND
  \bauthor{\bsnm{Horowitz},~\bfnm{J.~L.}\binits{J.~L.}}
(\byear{2007}).
\btitle{Methodology and convergence rates for funcational linear regeression}.
\bjournal{Ann. Statist.}
\bvolume{35}
\bpages{70-91}.
\end{barticle}
\endbibitem

\bibitem[\protect\citeauthoryear{Iglesias, Lin and
  Stuart}{2014}]{iglesias2014well}
\begin{barticle}[author]
\bauthor{\bsnm{Iglesias},~\bfnm{Marco~A}\binits{M.~A.}},
  \bauthor{\bsnm{Lin},~\bfnm{Kui}\binits{K.}} \AND
  \bauthor{\bsnm{Stuart},~\bfnm{Andrew~M}\binits{A.~M.}}
(\byear{2014}).
\btitle{Well-posed Bayesian geometric inverse problems arising in subsurface
  flow}.
\bjournal{Inverse Problems}
\bvolume{30}
\bpages{114001}.
\end{barticle}
\endbibitem

\bibitem[\protect\citeauthoryear{Kaipio et~al.}{2000}]{kaipio2000statistical}
\begin{barticle}[author]
\bauthor{\bsnm{Kaipio},~\bfnm{Jari~P}\binits{J.~P.}},
  \bauthor{\bsnm{Kolehmainen},~\bfnm{Ville}\binits{V.}},
  \bauthor{\bsnm{Somersalo},~\bfnm{Erkki}\binits{E.}} \AND
  \bauthor{\bsnm{Vauhkonen},~\bfnm{Marko}\binits{M.}}
(\byear{2000}).
\btitle{Statistical inversion and Monte Carlo sampling methods in electrical
  impedance tomography}.
\bjournal{Inverse problems}
\bvolume{16}
\bpages{1487}.
\end{barticle}
\endbibitem

\bibitem[\protect\citeauthoryear{Karhunen}{1947}]{DimRedu:Karhunen_1947}
\begin{barticle}[author]
\bauthor{\bsnm{Karhunen},~\bfnm{Kari}\binits{K.}}
(\byear{1947}).
\btitle{\"{U}ber lineare Methoden in der Wahrscheinlichkeitsrechnung}.
\bjournal{Ann. Acad. Sci. Fennicae. Ser. A. I. Math.-Phys}
\bvolume{37}
\bpages{1--79}.
\end{barticle}
\endbibitem

\bibitem[\protect\citeauthoryear{Kato}{1982}]{Kat82}
\begin{bbook}[author]
\bauthor{\bsnm{Kato},~\bfnm{Tosio}\binits{T.}}
(\byear{1982}).
\btitle{A Short Introduction to Perturbation Theory for Linear Operators}.
\bpublisher{Springer-Verlag}.
\end{bbook}
\endbibitem

\bibitem[\protect\citeauthoryear{Ledoux}{1994}]{ledoux1994simple}
\begin{barticle}[author]
\bauthor{\bsnm{Ledoux},~\bfnm{Michel}\binits{M.}}
(\byear{1994}).
\btitle{A simple analytic proof of an inequality by P. Buser}.
\bjournal{Proceedings of the American mathematical society}
\bvolume{121}
\bpages{951--959}.
\end{barticle}
\endbibitem

\bibitem[\protect\citeauthoryear{Lie, Sullivan and
  Teckentrup}{2019}]{lie2019error}
\begin{barticle}[author]
\bauthor{\bsnm{Lie},~\bfnm{Han~Cheng}\binits{H.~C.}},
  \bauthor{\bsnm{Sullivan},~\bfnm{Timothy~John}\binits{T.~J.}} \AND
  \bauthor{\bsnm{Teckentrup},~\bfnm{Aretha}\binits{A.}}
(\byear{2019}).
\btitle{Error bounds for some approximate posterior measures in Bayesian
  inference}.
\bjournal{arXiv preprint arXiv:1911.05669}.
\end{barticle}
\endbibitem

\bibitem[\protect\citeauthoryear{Liu and Wang}{2016}]{liu2016stein}
\begin{barticle}[author]
\bauthor{\bsnm{Liu},~\bfnm{Qiang}\binits{Q.}} \AND
  \bauthor{\bsnm{Wang},~\bfnm{Dilin}\binits{D.}}
(\byear{2016}).
\btitle{Stein variational gradient descent: A general purpose bayesian
  inference algorithm}.
\bjournal{Advances in neural information processing systems}
\bvolume{29}
\bpages{2378--2386}.
\end{barticle}
\endbibitem

\bibitem[\protect\citeauthoryear{Lo\`{e}ve}{1978}]{DimRedu:Loeve_1978}
\begin{bbook}[author]
\bauthor{\bsnm{Lo\`{e}ve},~\bfnm{Michel}\binits{M.}}
(\byear{1978}).
\btitle{Probability theory, Vol. II},
\bedition{4} ed.
\bseries{Graduate Texts in Mathematics}
\bvolume{46}.
\bpublisher{Springer-Verlag}, \baddress{Berlin}.
\end{bbook}
\endbibitem

\bibitem[\protect\citeauthoryear{Martin et~al.}{2012}]{martin2012stochastic}
\begin{barticle}[author]
\bauthor{\bsnm{Martin},~\bfnm{James}\binits{J.}},
  \bauthor{\bsnm{Wilcox},~\bfnm{Lucas~C}\binits{L.~C.}},
  \bauthor{\bsnm{Burstedde},~\bfnm{Carsten}\binits{C.}} \AND
  \bauthor{\bsnm{Ghattas},~\bfnm{Omar}\binits{O.}}
(\byear{2012}).
\btitle{A stochastic Newton MCMC method for large-scale statistical inverse
  problems with application to seismic inversion}.
\bjournal{SIAM Journal on Scientific Computing}
\bvolume{34}
\bpages{A1460--A1487}.
\end{barticle}
\endbibitem

\bibitem[\protect\citeauthoryear{Marzouk et~al.}{2016}]{marzouk2016sampling}
\begin{barticle}[author]
\bauthor{\bsnm{Marzouk},~\bfnm{Youssef}\binits{Y.}},
  \bauthor{\bsnm{Moselhy},~\bfnm{Tarek}\binits{T.}},
  \bauthor{\bsnm{Parno},~\bfnm{Matthew}\binits{M.}} \AND
  \bauthor{\bsnm{Spantini},~\bfnm{Alessio}\binits{A.}}
(\byear{2016}).
\btitle{Sampling via measure transport: An introduction}.
\bjournal{Handbook of uncertainty quantification}
\bpages{1--41}.
\end{barticle}
\endbibitem

\bibitem[\protect\citeauthoryear{Menz et~al.}{2014}]{menz2014poincare}
\begin{barticle}[author]
\bauthor{\bsnm{Menz},~\bfnm{Georg}\binits{G.}},
  \bauthor{\bsnm{Schlichting},~\bfnm{Andr{\'e}}\binits{A.}} \betal{et~al.}
(\byear{2014}).
\btitle{Poincar{\'e} and logarithmic Sobolev inequalities by decomposition of
  the energy landscape}.
\bjournal{Annals of Probability}
\bvolume{42}
\bpages{1809--1884}.
\end{barticle}
\endbibitem

\bibitem[\protect\citeauthoryear{Morzfeld, Tong and
  Marzouk}{2019}]{morzfeld2019localization}
\begin{barticle}[author]
\bauthor{\bsnm{Morzfeld},~\bfnm{Matthias}\binits{M.}},
  \bauthor{\bsnm{Tong},~\bfnm{Xin~T}\binits{X.~T.}} \AND
  \bauthor{\bsnm{Marzouk},~\bfnm{Youssef~M}\binits{Y.~M.}}
(\byear{2019}).
\btitle{Localization for MCMC: sampling high-dimensional posterior
  distributions with local structure}.
\bjournal{Journal of Computational Physics}
\bvolume{380}
\bpages{1--28}.
\end{barticle}
\endbibitem

\bibitem[\protect\citeauthoryear{Murray, MacKay and
  Adams}{2008}]{murray2008gaussian}
\begin{barticle}[author]
\bauthor{\bsnm{Murray},~\bfnm{Iain}\binits{I.}},
  \bauthor{\bsnm{MacKay},~\bfnm{David}\binits{D.}} \AND
  \bauthor{\bsnm{Adams},~\bfnm{Ryan~P}\binits{R.~P.}}
(\byear{2008}).
\btitle{The Gaussian process density sampler}.
\bjournal{Advances in Neural Information Processing Systems}
\bvolume{21}
\bpages{9--16}.
\end{barticle}
\endbibitem

\bibitem[\protect\citeauthoryear{Otto and
  Villani}{2000}]{otto2000generalization}
\begin{barticle}[author]
\bauthor{\bsnm{Otto},~\bfnm{Felix}\binits{F.}} \AND
  \bauthor{\bsnm{Villani},~\bfnm{C{\'e}dric}\binits{C.}}
(\byear{2000}).
\btitle{Generalization of an inequality by Talagrand and links with the
  logarithmic Sobolev inequality}.
\bjournal{Journal of Functional Analysis}
\bvolume{173}
\bpages{361--400}.
\end{barticle}
\endbibitem

\bibitem[\protect\citeauthoryear{Parente et~al.}{2020}]{parente2020generalized}
\begin{barticle}[author]
\bauthor{\bsnm{Parente},~\bfnm{Mario~Teixeira}\binits{M.~T.}},
  \bauthor{\bsnm{Wallin},~\bfnm{Jonas}\binits{J.}},
  \bauthor{\bsnm{Wohlmuth},~\bfnm{Barbara}\binits{B.}} \betal{et~al.}
(\byear{2020}).
\btitle{Generalized bounds for active subspaces}.
\bjournal{Electronic Journal of Statistics}
\bvolume{14}
\bpages{917--943}.
\end{barticle}
\endbibitem

\bibitem[\protect\citeauthoryear{Petra et~al.}{2014}]{petra2014computational}
\begin{barticle}[author]
\bauthor{\bsnm{Petra},~\bfnm{Noemi}\binits{N.}},
  \bauthor{\bsnm{Martin},~\bfnm{James}\binits{J.}},
  \bauthor{\bsnm{Stadler},~\bfnm{Georg}\binits{G.}} \AND
  \bauthor{\bsnm{Ghattas},~\bfnm{Omar}\binits{O.}}
(\byear{2014}).
\btitle{A computational framework for infinite-dimensional Bayesian inverse
  problems, Part II: Stochastic Newton MCMC with application to ice sheet flow
  inverse problems}.
\bjournal{SIAM Journal on Scientific Computing}
\bvolume{36}
\bpages{A1525--A1555}.
\end{barticle}
\endbibitem

\bibitem[\protect\citeauthoryear{Ramsay and Silverman}{2005}]{RS05}
\begin{bbook}[author]
\bauthor{\bsnm{Ramsay},~\bfnm{J.~O.}\binits{J.~O.}} \AND
  \bauthor{\bsnm{Silverman},~\bfnm{B.~W.}\binits{B.~W.}}
(\byear{2005}).
\btitle{Functional data analysis (2nd ed.)}.
\bpublisher{Springer}.
\end{bbook}
\endbibitem

\bibitem[\protect\citeauthoryear{Rudolf and
  Sprungk}{2018}]{rudolf2018generalization}
\begin{barticle}[author]
\bauthor{\bsnm{Rudolf},~\bfnm{Daniel}\binits{D.}} \AND
  \bauthor{\bsnm{Sprungk},~\bfnm{Bj{\"o}rn}\binits{B.}}
(\byear{2018}).
\btitle{On a generalization of the preconditioned Crank--Nicolson Metropolis
  algorithm}.
\bjournal{Foundations of Computational Mathematics}
\bvolume{18}
\bpages{309--343}.
\end{barticle}
\endbibitem

\bibitem[\protect\citeauthoryear{Sanz-Alonso}{2018}]{sanz2018importance}
\begin{barticle}[author]
\bauthor{\bsnm{Sanz-Alonso},~\bfnm{Daniel}\binits{D.}}
(\byear{2018}).
\btitle{Importance sampling and necessary sample size: an information theory
  approach}.
\bjournal{SIAM/ASA Journal on Uncertainty Quantification}
\bvolume{6}
\bpages{867--879}.
\end{barticle}
\endbibitem

\bibitem[\protect\citeauthoryear{Spantini, Bigoni and
  Marzouk}{2018}]{spantini2018inference}
\begin{barticle}[author]
\bauthor{\bsnm{Spantini},~\bfnm{Alessio}\binits{A.}},
  \bauthor{\bsnm{Bigoni},~\bfnm{Daniele}\binits{D.}} \AND
  \bauthor{\bsnm{Marzouk},~\bfnm{Youssef}\binits{Y.}}
(\byear{2018}).
\btitle{Inference via low-dimensional couplings}.
\bjournal{The Journal of Machine Learning Research}
\bvolume{19}
\bpages{2639--2709}.
\end{barticle}
\endbibitem

\bibitem[\protect\citeauthoryear{Spantini et~al.}{2015}]{spantini2015optimal}
\begin{barticle}[author]
\bauthor{\bsnm{Spantini},~\bfnm{Alessio}\binits{A.}},
  \bauthor{\bsnm{Solonen},~\bfnm{Antti}\binits{A.}},
  \bauthor{\bsnm{Cui},~\bfnm{Tiangang}\binits{T.}},
  \bauthor{\bsnm{Martin},~\bfnm{James}\binits{J.}},
  \bauthor{\bsnm{Tenorio},~\bfnm{Luis}\binits{L.}} \AND
  \bauthor{\bsnm{Marzouk},~\bfnm{Youssef}\binits{Y.}}
(\byear{2015}).
\btitle{Optimal low-rank approximations of Bayesian linear inverse problems}.
\bjournal{SIAM Journal on Scientific Computing}
\bvolume{37}
\bpages{A2451--A2487}.
\end{barticle}
\endbibitem

\bibitem[\protect\citeauthoryear{Steerneman}{1983}]{steerneman1983total}
\begin{barticle}[author]
\bauthor{\bsnm{Steerneman},~\bfnm{Ton}\binits{T.}}
(\byear{1983}).
\btitle{On the total variation and Hellinger distance between signed measures;
  an application to product measures}.
\bjournal{Proceedings of the American Mathematical Society}
\bvolume{88}
\bpages{684--688}.
\end{barticle}
\endbibitem

\bibitem[\protect\citeauthoryear{Stewart}{1980}]{stewart1980efficient}
\begin{barticle}[author]
\bauthor{\bsnm{Stewart},~\bfnm{Gilbert~W}\binits{G.~W.}}
(\byear{1980}).
\btitle{The efficient generation of random orthogonal matrices with an
  application to condition estimators}.
\bjournal{SIAM Journal on Numerical Analysis}
\bvolume{17}
\bpages{403--409}.
\end{barticle}
\endbibitem

\bibitem[\protect\citeauthoryear{Stuart}{2010}]{St10}
\begin{barticle}[author]
\bauthor{\bsnm{Stuart},~\bfnm{A.}\binits{A.}}
(\byear{2010}).
\btitle{Inverse problems: a Bayesian perspective}.
\bjournal{Acta numerica}
\bvolume{19}
\bpages{451-559}.
\end{barticle}
\endbibitem

\bibitem[\protect\citeauthoryear{Sullivan}{2015}]{sullivan2015introduction}
\begin{bbook}[author]
\bauthor{\bsnm{Sullivan},~\bfnm{Timothy~John}\binits{T.~J.}}
(\byear{2015}).
\btitle{Introduction to uncertainty quantification}
\bvolume{63}.
\bpublisher{Springer}.
\end{bbook}
\endbibitem

\bibitem[\protect\citeauthoryear{Tabak, Trigila and
  Zhao}{2020}]{tabak2020conditional}
\begin{barticle}[author]
\bauthor{\bsnm{Tabak},~\bfnm{Esteban~G}\binits{E.~G.}},
  \bauthor{\bsnm{Trigila},~\bfnm{Giulio}\binits{G.}} \AND
  \bauthor{\bsnm{Zhao},~\bfnm{Wenjun}\binits{W.}}
(\byear{2020}).
\btitle{Conditional density estimation and simulation through optimal
  transport}.
\bjournal{Machine Learning}
\bpages{1--24}.
\end{barticle}
\endbibitem

\bibitem[\protect\citeauthoryear{Tabak and Turner}{2013}]{tabak2013family}
\begin{barticle}[author]
\bauthor{\bsnm{Tabak},~\bfnm{Esteban~G}\binits{E.~G.}} \AND
  \bauthor{\bsnm{Turner},~\bfnm{Cristina~V}\binits{C.~V.}}
(\byear{2013}).
\btitle{A family of nonparametric density estimation algorithms}.
\bjournal{Communications on Pure and Applied Mathematics}
\bvolume{66}
\bpages{145--164}.
\end{barticle}
\endbibitem

\bibitem[\protect\citeauthoryear{Tong, Morzfeld and
  Marzouk}{2020}]{tong2020mala}
\begin{barticle}[author]
\bauthor{\bsnm{Tong},~\bfnm{Xin~T}\binits{X.~T.}},
  \bauthor{\bsnm{Morzfeld},~\bfnm{Mathias}\binits{M.}} \AND
  \bauthor{\bsnm{Marzouk},~\bfnm{Youssef~M}\binits{Y.~M.}}
(\byear{2020}).
\btitle{MALA-within-Gibbs samplers for high-dimensional distributions with
  sparse conditional structure}.
\bjournal{SIAM Journal on Scientific Computing}
\bvolume{42}
\bpages{A1765--A1788}.
\end{barticle}
\endbibitem

\bibitem[\protect\citeauthoryear{Trigila and Tabak}{2016}]{trigila2016data}
\begin{barticle}[author]
\bauthor{\bsnm{Trigila},~\bfnm{Giulio}\binits{G.}} \AND
  \bauthor{\bsnm{Tabak},~\bfnm{Esteban~G}\binits{E.~G.}}
(\byear{2016}).
\btitle{Data-driven optimal transport}.
\bjournal{Communications on Pure and Applied Mathematics}
\bvolume{69}
\bpages{613--648}.
\end{barticle}
\endbibitem

\bibitem[\protect\citeauthoryear{Tsybakov}{2008}]{tsybakov2008introduction}
\begin{bbook}[author]
\bauthor{\bsnm{Tsybakov},~\bfnm{Alexandre~B}\binits{A.~B.}}
(\byear{2008}).
\btitle{Introduction to nonparametric estimation}.
\bpublisher{Springer Science \& Business Media}.
\end{bbook}
\endbibitem

\bibitem[\protect\citeauthoryear{Yu, Wang and Samworth}{2015}]{yu2015useful}
\begin{barticle}[author]
\bauthor{\bsnm{Yu},~\bfnm{Yi}\binits{Y.}},
  \bauthor{\bsnm{Wang},~\bfnm{Tengyao}\binits{T.}} \AND
  \bauthor{\bsnm{Samworth},~\bfnm{Richard~J}\binits{R.~J.}}
(\byear{2015}).
\btitle{A useful variant of the Davis--Kahan theorem for statisticians}.
\bjournal{Biometrika}
\bvolume{102}
\bpages{315--323}.
\end{barticle}
\endbibitem

\bibitem[\protect\citeauthoryear{Zahm et~al.}{2018}]{zahm2018certified}
\begin{barticle}[author]
\bauthor{\bsnm{Zahm},~\bfnm{Olivier}\binits{O.}},
  \bauthor{\bsnm{Cui},~\bfnm{Tiangang}\binits{T.}},
  \bauthor{\bsnm{Law},~\bfnm{Kody}\binits{K.}},
  \bauthor{\bsnm{Spantini},~\bfnm{Alessio}\binits{A.}} \AND
  \bauthor{\bsnm{Marzouk},~\bfnm{Youssef}\binits{Y.}}
(\byear{2018}).
\btitle{Certified dimension reduction in nonlinear Bayesian inverse problems}.
\bjournal{arXiv preprint arXiv:1807.03712}.
\end{barticle}
\endbibitem

\end{thebibliography}

\end{document}